\documentclass[11pt,draftcls,onecolumn,romanappendices]{IEEEtran}%

\usepackage{amsmath}%
\usepackage{amsfonts}%
\usepackage{amssymb}%
\usepackage{amsthm}
\usepackage{graphicx}
\usepackage{mathrsfs}
\usepackage{cite}
\usepackage{tikz}
\usepackage{xcolor}
\usepackage[colorlinks=true,linkcolor=blue,citecolor=red]{hyperref}
\usepackage{setspace}
\usepackage{geometry}
\usepackage{algorithm}
\usepackage{algorithmic}
\usepackage{psfrag}
\usepackage{pgfplots}
\pgfplotsset{compat=1.5}

\newtheorem{theorem}{Theorem}

\newtheorem{corollary}{Corollary}

\newtheorem{lemma}{Lemma}

\newtheorem{proposition}{Proposition}

\def\bh{\mathbf{h}}
\def\bH{\mathbf{H}}
\def\bx{\mathbf{x}}

\def\ih{\widetilde{\mathbf{h}}}
\def\hH{\widehat{\mathbf{H}}}
\def\hh{\widehat{\mathbf{h}}}

\def\bI{\mathbf{I}}

\def\be{\begin{equation}}
\def\ee{\end{equation}}
\def\bea{\begin{align}}
\def\eea{\end{align}}
\def\bean{\begin{align*}}
\def\eean{\end{align*}}

\def\Tr{\text{Tr}}
\def\ba{\left[\begin{array}}
\def\ea{\end{array}\right]}
\def\as{\overset{a.s}{\longrightarrow}}
\def\ip{\overset{i.p.}{\longrightarrow}}
\def\ms{\overset{L_2}{\longrightarrow}}
\newcommand{\mb}[1]{
\mathbf{#1}
}
\newcommand{\TR}[1]{
\text{Tr}\left(#1\right)
}

\def\tSINR{\text{SINR}}
\def\bs{\mathbf{s}}

\def\bu{\mathbf{u}}
\def\bg{\mathbf{g}}
\def\sG{\mathsf{G}}
\def\bd{\mathbf{d}}
\def\sD{\mathsf{D}}

\def\bv{\mathbf{v}}
\def\bp{\mathbf{p}}
\def\by{\mathbf{y}}

\def\bw{\mathbf{w}}
\def\hw{\hat{\mathbf{w}}}

\def\bO{\mathbf{O}}

\def\bA{\mathbf{A}}
\def\bX{\mathbf{X}}
\def\bg{\mathbf{g}}

\def\bQ{\mathbf{Q}}
\def\b1{\mathbf{1}}

\def\bPi{\mathbf{\Pi}}
\def\bL{\mathbf{\Lambda}}

\def\bPhi{\mathbf{\Phi}}

\def\half{\frac{1}{2}}
\def\bryu{\breve{\gamma}_u}

\begin{document}
\title{Base Station Cooperation with Feedback Optimization: A Large System Analysis}

\author{\IEEEauthorblockN{Rusdha Muharar, Randa Zakhour and Jamie Evans}
\thanks{The material in this paper was presented in part at the IEEE International Conference on Communications (ICC), Ottawa, Canada, June 2012 and at the IEEE International Symposium on Information Theory (ISIT), Cambridge, MA, USA, July 2012.}
\thanks{Rusdha Muharar is with the Department of Electrical and Computer Systems Engineering, Monash University, Clayton, VIC 3800, Australia and Syiah Kuala University, Banda Aceh, Indonesia (e-mail: rusdha.muharar@monash.edu; r.muharar@unsyiah.ac.id).}%
\thanks{Randa Zakhour is part-time faculty at the Schools of Engineering of the American University of Science and Technology and of the Lebanese International University, Lebanon  (email: randa.zakhour@gmail.com).}
\thanks{Jamie Evans is with the Department of Electrical and Computer Systems Engineering, Monash University, Clayton, VIC 3800, Australia (email: jamie.evans@monash.edu).}
}

\maketitle \thispagestyle{plain} \pagestyle{plain}

\begin{abstract}
In this paper, we study feedback optimization problems that maximize the 
users' signal to interference plus noise ratio (SINR) in a two-cell MIMO 
broadcast channel. Assuming the users learn their direct and interfering 
channels perfectly, they can feed back this information to the base stations 
(BSs) over the uplink channels. The BSs then use the channel information to 
design their transmission scheme. Two types of feedback are considered: 
analog and digital.   In the analog feedback case, the users send their 
unquantized and uncoded CSI over the uplink channels. In this context, given 
a user's fixed transmit power, we  investigate how he/she should optimally 
allocate it to feed back the direct and interfering (or cross) CSI 
for two types of base station cooperation schemes, namely, Multi-Cell Processing (MCP) and 
Coordinated Beamforming (CBf). In the digital feedback case, the direct and cross link channel vectors of each user are quantized separately, each using RVQ,
with different size codebooks. The users then send the index of the quantization vector in the corresponding codebook to the BSs. Similar 
to the feedback optimization problem in the analog feedback, we investigate 
the optimal bit partitioning for the direct and interfering link for both 
types of cooperation.

We focus on regularized channel inversion precoding structures and perform 
our analysis in the large system limit in which the number of users per cell ($K$) 
and the number of antennas per BS ($N$) tend to infinity with their 
ratio $\beta=\frac{K}{N}$  held fixed. 
We show that for both types of cooperation, for some values of interfering 
channel gain, usually at low values, no cooperation between the base stations 
is preferred: This is because, for these values of cross channel gain, the 
channel estimates for the cross link are not accurate enough for their knowledge 
to contribute to improving the SINR and there is 
no benefit in doing base station cooperation under that condition. We also show that for the MCP scheme, unlike in the perfect CSI case, 
the SINR improves only when the interfering channel gain is 
above a certain threshold.      
\end{abstract}

\section{Introduction}
\subsection{Background}
Recently, many applications that require high data rates such as high quality 
video streaming and huge volume data transfers through wireless
communication systems have emerged. MIMO communication systems have 
arisen as a promising candidate to support this requirement and have been adopted for existing 
and future wireless communication standards such as in IEEE 802.11n and 4G networks.
Current MIMO technological advancements can be considered as 
the results of research works started about fifteen years ago.
So far, there has been a considerable amout of work focusing on single 
user and single-cell multiuser MIMO systems.
Only recently, researchers have started to put more attention to investigate 
how to maximize data rates in \emph{multi-cell} MIMO networks, particularly 
in the downlink \cite[and references therein]{Gesbert_jsac10}.

The main challenge that limits the spectral efficiency in the downlink of 
multi-cell networks, besides intra-cell interference, is the inter-cell 
interference (ICI). The conventional approach to mitigate this interference 
is to use spatial reuse of resources such as frequency and time \cite{Gesbert_
 jsac10}. The move towards aggressive frequency or time reuse will cause the 
networks to be interference limited especially for the users at the cell 
edge. The current view is to mitigate ICI through base station (BS) 
cooperations. Within this scheme, the BSs share the control signal, channel 
state  information (CSI) and data symbols for all users via a central 
processing unit or wired backhaul links \cite{Ramya_tsp11}. 

It has been established in \cite{Shamai_vtc01, Zhang_eura04, Karakayali_wc06 
,Foschini_iee06, Jing_eura08, Somekh_it09}, to name a few,  that  MIMO 
cooperation schemes  provide a significant increase in spectral efficiency 
compared to conventional cellular networks. BS cooperation can be implemented 
at different levels \cite{Gesbert_jsac10}.  In the \textit{Multi-Cell Processing} 
setup, also known as \textsl{Network MIMO} or \textit{Coordinated Multi-Point} (CoMP) transmission, the BSs 
fully cooperate and share both the channel state information (CSI) and 
transmission data. This full cooperation requires high capacity backhaul 
links which are sometimes not viable in practical settings. To alleviate this 
requirement, only CSI (including direct and interfering channels) is shared 
amongst base stations in the \textit{interference coordination} scheme \cite{
Gesbert_jsac10}. Several works have addressed  coordinated beamforming 
and power control schemes to improve the spectral efficiency in  interference-limited downlink multi-cell networks. Detailed discussions regarding these 
topics can be found in \cite{Gesbert_jsac 10} and references therein. 

In both base station cooperation schemes, the CSI at the base stations plays 
an important role in maximizing the system performance. The  
base stations use this information to adapt their transmission strategies to the channel 
conditions.  The benefit 
of having CSI at the transmitter (CSIT) with respect to the capacity in single and 
multi-cell multi-antenna systems is nicely summarized in \cite{Goldsmith_jsac
03, Biglieri_book07}. However, these advantages are also accompanied by the 
overhead cost for the CSI acquisition via channel training and feedback in 
frequency division duplex (FDD) systems. It needs to scale proportionally to 
the number of transmit and  receive antennas and the number of users in the 
system  in order to maintain a constant gap of the sum-rate with respect to the full CSI case 
\cite{Love_jsac08}. Moreover, in practical systems, the backhaul-link 
capacity for CSI and user data exchanges and feedback-link bandwidth are 
limited \cite{Ramya_tsp11}.  Considering the CSI signaling overhead from 
channel training and CSI feedback, references \cite{Ramprashad_pimrc09, 
Ramprashad_asil09} (see also \cite{JZhang_arxiv11}) suggested that the 
conventional single-cell processing (SCP) without coordination may outperform the 
cooperative systems, even the MCP scheme. In this paper, to reduce 
the complexity in the analysis, we ignore the (important) constraints of 
limited backhaul-link and CSI training overhead. We assume a perfect CSI 
training so that all users know their CSI perfectly. We focus on studying 
how to allocate feedback resources, that depend on the feedback schemes, to 
send the CSI for the direct channel and interfering (cross) channel to BSs so 
that the users' SINR are maximized. Two feedback schemes are considered in 
our study: the analog feedback scheme, introduced in \cite{Marzetta_tsp06} 
and the limited (quantized) feedback via random vector quantization (RVQ), 
introduced in \cite{Santipach_it09}. In the analog feedback scheme, each user 
sends its unquantized and uncoded channel state information through the 
uplink channel. Hence, we ask the question, for a given uplink power 
constraint, what fraction of this uplink power is allocated optimally to 
transmit the direct and interfering channel information? For the digital 
feedback scheme, the number of feedback bits determines the quality of the 
CSI. Hence, we can ask, how many bits are optimally needed to feedback the 
direct and cross CSI? 

\subsection{Contributions} 

The main goal of this paper is to optimize and 
investigate the effect of feedback  for MCP and CBf cooperation schemes under analog 
and quantized feedback (via RVQ). We consider a \emph{symmetric} two-cell 
Multi-Input Single-Output (MISO) network where the base stations have multiple antennas and each 
user has a single antenna. We assume that the users in each cell know their 
own channel perfectly: they feed back this information through the uplink 
channel and the base stations form the users' channel estimates. The BSs 
use these estimates to construct a regularized channel inversion (RCI) 
type beamformer, also called regularized zero-forcing (RZF), to precode 
the data symbols of the users. The precoders follow the structures proposed in  \cite{Zakhour_it12}.
Unlike \cite{Ramya_tsp11,JZhang_arxiv11}, we 
assume several users are simultaneously active in each cell so that the users 
experience both intra- and inter-cell interference. To mitigate ICI through 
base station cooperation, we consider both full cooperation (MCP) and 
interference coordination via CBf.  

Our contributions can be summarized as follows. First, under both feedback 
models and both cooperation schemes, we derive the SINR expression in the 
large system limit, also called the \textit{limiting SINR}, where the number 
of antennas at base stations and the number of users in each cell go to 
infinity with their ratio kept fixed: As our numerical results will show, \emph{this is indicative of the average 
performance for even finite numbers of antennas}. Then, we formulate a joint 
optimization problem that performs the feedback optimization for both 
feedback models and  both cooperation schemes and finds the optimal 
regularization parameter of the corresponding RCI-structured precoder.  
The regularization parameter is an important design parameter for 
the precoder because it controls the amount of interference introduced 
to the users. Optimizing this parameter, as discussed later, 
will allow the precoder to adapt to the changes of the CSIT quality and 
consequently produces a '\textit{robust beamformer}'.

We analyze the behavior of the maximum limiting SINR as a function of the cross channel gains and the available feedback resources, and identify, for both the analog and quantized feedback models, regions where SCP processing is optimal. We also show that whereas in the perfect CSI case, MCP performance always improves with epsilon, this only occurs after a certain threshold is crossed in both analog and limited feedback cases.

Parts of this work appeared in \cite{Muharar_icc12,Muharar_isit12}, but  without the proofs.

\subsection{Related Works}
In the last decade, there has been a large volume of research discussing
feedback schemes in multi-antenna systems. A summary of digital feedback (also known as limited or finite-rate feedback) 
schemes  in multi-antenna (also single-antenna) and multi-user systems in the single-cell setup can be found 
in \cite{Love_jsac08}. Since the optimal codebook for the limited feedback is 
not known yet \cite{Santipach_it09, Jindal_it06, Caire_it10}, the use of RVQ, 
which is based on a random codebook, as the feedback scheme becomes popular. 
Furthermore, the RVQ-based system performance analysis is also more tractable.  
In multiple-antenna and multiuser systems, works on the analog feedback commonly 
refer to \cite{Marzetta_tsp06} (sometimes \cite{Samar
_comm06}).      

The paper by Jindal \cite{Jindal_it06} sparked the use of RVQ in analyzing 
broadcast channels. Considering a MISO broadcast channel with a zero-forcing 
(ZF) precoder and  assuming that each user knows its own channel, the main 
result in the paper is that the feedback rate should be increased linearly 
with the signal-to-noise ratio (SNR) to maintain the full multiplexing 
gain.   Caire et al. in \cite{Caire_it10} investigate achievable ergodic sum rates of BC with ZF precoder
under several practical scenarios. The CSI acquisition involves four steps; downlink training, CSI 
feedback, beamformer selection and dedicated training where each user will 
try to estimate the coupling between its channel and the beamforming vectors. 
They derive and compare the lower bound and upper bound of the achievable ergodic sum- 
rate of the analog feedback as in \cite{Marzetta_tsp06} and RVQ-based digital feedback under different considerations, e.g., 
feedback transmission over AWGN and MAC channel, feedback delay and feedback errors for the digital 
feedback scheme. A subsequent work by Kobayashi et al. in \cite{Koba_comm11} 
studies training and feedback optimizations for the same system setup as in 
\cite{Caire_it10} except without dedicated training. The optimal period for 
the training and feedback that minimized achievable rate gap (with and 
without perfect CSI) are derived under different scenarios as in \cite{Caire_ 
it10}. The authors also show that the digital feedback can give a significant 
advantage over the analog feedback.   In the same spirit as \cite{Jindal_it06}
, reference \cite{Wagner_it12} discusses the feedback scaling (as SNR 
increases) in order to maintain a constant rate gap for a broadcast channel 
with  RCI precoder. The analysis has been done in 
the large system limit since the analysis the finite-size turns out to be 
difficult \cite{Jindal_it06}. Moreover, besides analyzing for the case $K=N$, 
as in \cite{ Jindal_it06}, the authors also investigate 
the case $K < N$.         

While channel state feedback in the single-cell system has received a 
considerable amount of attention so far, fewer works have addressed this 
problem in multi-cell settings. The effect of channel uncertainty, 
specifically the channel estimation error, in the multi-cell setup is studied 
in \cite{Huh_it12, Jose_wcomm11}. In \cite{Jose_wcomm11}, the authors 
conclude that when channel estimates at one base station contain 
interferences from the users in other cells, also called as pilot 
contamination phenomenon, the inter-cell interference increases. Thus, this phenomenon 
could severely impact the performance of the systems. Huh et al. in \cite{Huh_it12}
investigate optimal user scheduling strategies to reduce the feedback and also the 
effects of channel estimation error on the ergodic sum-rate of the clustered 
Network MIMO systems. They consider the ZF  precoder at the base 
stations and derive the optimal power allocation that maximizes the 
weighted sum-rate. In deriving the results, it is assumed that 
the BSs received perfectly (error-free) the CSI fed back by the users.  
The overhead caused by the channel training is also investigated and they observe that there is a trade-off 
between the number of cooperating antennas and the cost of estimating the 
channel. Based on the trade-off, the optimal cooperation 
cluster size can be determined. By incorporating 
the channel training cost, no-coordination amongst the base stations could be 
preferable. The same conclusion is also obtained in \cite{Ramprashad_asil09 
,Ramprashad_pimrc09}.  

For the interference coordination scheme, \cite{Ramya_tsp11} presents the 
RVQ-based limited feedback in an infinite  Wyner cellular model using 
generalized eigenvector beamforming at the base stations. The work adopts the 
intra-cell TDMA mechanism where a single user is active in each cell per time 
slot. Each user in each cell is also assumed to know its downlink channel 
perfectly. Based on that system model, an optimal bit partitioning strategy 
for direct and interfering channels that minimizes the sum-rate gap is 
proposed. Explicitly, it is a function of the received SNR from the direct 
and cross links. It is observed that as the received SNR from the cross link 
increases, more bits are allocated to quantize the cross channel. A better 
quality of the cross channel estimate will help to reduce 
the inter-cell interference. The authors also show that the proposed bit 
partitioning scheme reduces the average sum-rate loss. Also in the 
interference coordination setting, \cite{JZhang_arxiv11} takes into account 
both CSI training and feedback in analyzing the system what they called the inter-cell 
interference cancellation (ICIC) scheme. In ICIC, the precoding vector of a user is 
the projection of its channel in the null-space of the others users' channels 
in other cells so that the transmission from this user will not cause interference 
to the users in other cells. The work also assumes the intra-cell TDMA and 
presents the training optimization and feedback optimization for both analog 
and digital feedback (RVQ). Based on that system setup, the most interesting 
result is  that the training optimization is more important than the feedback 
optimization for the analog feedback while the opposite holds for the digital feedback.

For different levels of cooperation, i.e., MCP, CBf and SCP, 
 \cite{Zakhour_it12} investigates an optimization 
problem to minimize the total downlink  transmit power while satisfying a 
specified SINR target. The authors derived the optimal transmit power, 
beamforming vectors, cell loading and achieved SINR for those different 
cooperation schemes in a symmetric two-cell network. The resulted optimal 
beamforming vectors have a structure related to RCI. 

The current work is closely related to \cite{Zakhour_it12} in the sense we 
use the same cooperative schemes and precoder structure. We extend the work 
by analyzing the optimal feedback strategies for analog and digital feedback 
under MCP and  CBf schemes. The results in this work are obtained by 
performing the analysis in the large system limit where the dimensions of the 
system i.e., the numbers of users and transmit antennas tend to infinity with 
their ratio being fixed. The large system analysis mainly exploits the 
eigenvalue distribution of large random matrices. For examples, it has 
been used to derive the asymptotic performance of linear multiuser receivers 
in CDMA communications in early 2000 (see \cite{Tulino_FnT_it04}),  single-cell broadcast channels with RCI for various channel conditions \cite{Nguyen_
globecom08, Muharar_ausctw11, Muharar_icc11, Wagner_it12}, base station 
cooperations in downlink multi-cell networks (see e.g., \cite{Huh_it12,Zakhour_
it12}). The asymptotic performance measure becomes a deterministic quantity 
and can have close-form/compact expressions. Hence, it can be used to derive 
the optimal parameters for the system design. Moreover, it can provide a good 
approximation of, hence insights on, the performance of the finite-size (or even small-size) systems. 

Similar to \cite{Ramya_tsp11} and \cite{JZhang_arxiv11}, we perform the feedback optimization in interference-coordination scheme (CBf). As in \cite{JZhang_arxiv11}, we also investigate the feedback optimization for the analog and digital feedback schemes. However, different from those works, we do not assume the intra-cell TDMA in each cell, and hence each user experiences both intra-cell and inter-cell interference. We also consider a different type of precoder i.e., the RCI. Moreover, we also analyze the feedback optimization for different level of cooperations between the base stations, including the MCP setup, and try to capture how we allocate resources available at the user side as the the interfering channel gain varies.

\subsection{Paper Organization and Notation}
The rest of the paper is structured as follows. The system model is described in Section \ref{sec:system_model}. It starts with the channel model, and the expressions of the transmit signal, precoder and the corresponding SINR for each MCP and CBf. In the end of the section, the feedback schemes and true channel model in term of the channel estimate at the BSs and the channel uncertainty for the analog and digital feedback are presented. The main results for the noisy analog feedback and digital feedback and for different types of coordination are discussed in Section \ref{sec:ana_feedback} and \ref{sec:RVQ_feedback}, respectively. In each section, we begin by discussing the large system result of the SINR for the MCP and CBf and then followed by deriving the corresponding optimal feedback allocation; optimal (uplink) power for the analog feedback and optimal bit partitioning for the digital feedback. The optimal regularization parameter for the RCI precoder is also derived for both types of feedback and cooperation. The end of each section provides numerical results that depict how the optimal feedback allocation and the SINR of each user behave as the interfering channel gain varies.  In Section \ref{sec:compare_ana_rvq}, we provide some numerical simulations that compare the performance of the system under the analog feedback and digital feedback. The conclusion are drawn in the Section \ref{sec:conclusion} and some of the proofs go to the appendices.

Throughout the paper, the following notations are used. $\mathbb{E}[\cdot]$ denotes the statistical expectation. The almost sure convergence, convergence in probability, and mean-square convergence are denoted by $\as$, $\ip$, $\ms$ respectively. The partial derivative of $f$ with respect to (w.r.t.) $x$  is denoted by $\frac{\partial f}{\partial x}$. The circularly symmetric complex Gaussian (CSCG) vector with mean $\pmb{\mu}$ and covariance matrix $\mathbf{\Sigma}$ is denoted by $\mathcal{CN}(\pmb{\mu},\mathbf{\Sigma})$. $|a|$ and $\Re[a]$ denote the magnitude and the real part of the complex variable $a$, respectively. $\|\cdot\|$ represents the Euclidean norm and $\TR{\cdot}$ denotes the trace of a matrix. $\bI_N$ and $\mathbf{0}_N$ denote an $N\times N$ identity matrix and a $1\times N$ zero entries vector, respectively. $(\cdot)^T$ and $(\cdot)^H$ refer to the transpose and Hermitian transpose, respectively. The angle between vector $\bx$ and $\mathbf{y}$ is denoted by $\angle(\bx,\mathbf{y})$. LHS and RHS refer to the left-hand and right-hand side of an equation, respectively.

\section{System Model}\label{sec:system_model}
We consider a symmetric two-cell broadcast channel, as shown in Figure \ref{system_model}, where each cell has $K$ single antenna users and a base station equipped with $N$ antennas. The channel between user $k$ in cell $j$ and the BS in cell $i$ is denoted by row vector $\bh_{k,j,i}$ where $\bh_{k,j,j} \sim \mathcal{CN}(\mathbf{0},\mathbf{I}_N)$
and $\bh_{k,j,\bar{j}} \sim \mathcal{CN}(\mathbf{0},\epsilon\mathbf{I}_N)$, for $j = 1, 2$ and $\bar{j} = \mod(j, 2) + 1$.
We refer to the $\bh_{k,j,j}$ as direct channels and $\bh_{k,j,\bar{j}}$ as cross or ``interfering" channels.
We find it useful to group these into a single channel vector $\bh_{k,j}=[\bh_{k,j,1}\ \bh_{k,j,2}]$.


We consider an FDD system and assume that the users have perfect knowledge of their downlink channels, $\bh_{k,j,j}$ and $\bh_{k,j,\bar{j}}$. 
Each user feeds back the channel information to the direct BS and neighboring BS through the corresponding uplink channels. The BSs  estimate or recover these channel states and use them to construct the precoder. 

The received signal of user $k$ in cell $j$ can be written as
\[
y_{k,j} = \bh_{k,j,1}\bx_{1} + \bh_{k,j,2}\bx_{2} + n_{k,j}, 
\]
where $\bx_i \in \mathbb{C}^{N\times 1}$, $i=1,2$ is the transmitted data from BS $i$,  and $n_{k,j} \sim \mathcal{CN}(0,\sigma^2_d)$ is the noise at the user's receiver. 
The transmitted data $\bx_i$ depends on the level of cooperation assumed, and will be described in more details in Sections \ref{sec:NMIMO} and \ref{sec:CBf}: we restrict ourselves to linear precoding schemes, more specifically RCI precoder.
We assume each BS's transmission is subject to a power constraint $\mathbb{E}\left[\|\bx_i\|^2\right]=P_{i}$.  In the MCP case, we relax this constraint to a sum power constraint so that $\mathbb{E}\left[\|\bx\|^2\right]=\sum_{i=1}^2 P_{i} = P_t$. In the analysis, we assume $P_1=P_2=P$ and denote $\gamma_d=P/\sigma^2_d$.

As already mentioned, in practical scenarios, perfect CSI is difficult to obtain and the CSI at the BSs is obtained through feedback from the users. We are particularly interested in the channel model where we can express the downlink channel between the user $k$
in cell $j$ and BS $i$ as
\be\label{eq:ch_model}
\bh_{k,j,i}=\sqrt{\phi_{k,j,i}}\hh_{k,j,i}+ \ih_{k,j,i},
\ee
where $\hh_{k,j,i}$ represents the channel estimate, and $\ih_{k,j,i}$ the channel uncertainty or estimation error. 
The channel estimates are used by the BSs to construct the precoder. 

The transmitted signal, precoder and SINR for each user for each cooperation scheme will be presented in the following subsections. 

\begin{figure}[!t]
\centering
\begin{tikzpicture}[scale=1.5]
\node at (0,1) {\includegraphics[scale=0.4]{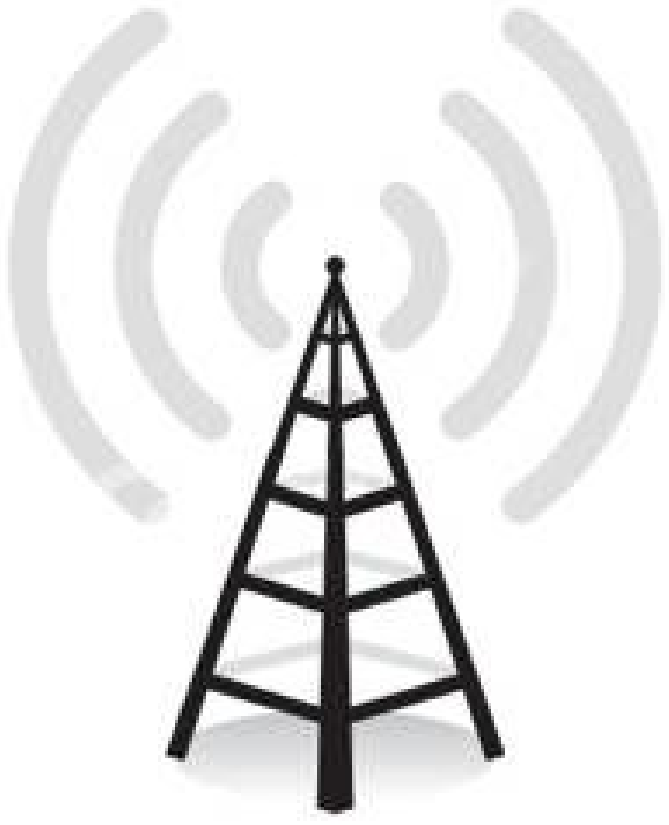}};
\draw[thick] (0,0) ellipse (2.2cm and 1cm);
\node at (0,-1.3) {Cell 1};
\draw[thick,fill=white] (0,1.52) ellipse (0.32cm and 0.12cm);
\draw[very thick,brown] (0.34,1.37) -- (0.34,1.67);
\draw[very thick,brown] (-0.34,1.37) -- (-0.34,1.67);
\draw[very thick,brown] (0.17,1.27) -- (0.17,1.57);
\draw[very thick,brown] (-0.17,1.27) -- (-0.17,1.57);
\draw[very thick,brown] (0.12,1.47) -- (0.12,1.77);
\draw[very thick,brown] (-0.12,1.47) -- (-0.12,1.77);
\node at (4.4,1) {\includegraphics[scale=0.4]{BSeps.eps}};
\draw[thick] (4.4,0) ellipse (2.2cm and 1cm);
\node at (4.4,-1.3) {Cell 2};
\draw[thick,fill=white] (4.4,1.52) ellipse (0.32cm and 0.12cm);
\draw[very thick,brown] (4.74,1.37) -- (4.74,1.67);
\draw[very thick,brown] (4.06,1.37) -- (4.06,1.67);
\draw[very thick,brown] (4.57,1.27) -- (4.57,1.57);
\draw[very thick,brown] (4.23,1.27) -- (4.23,1.57);
\draw[very thick,brown] (4.52,1.47) -- (4.52,1.77);
\draw[very thick,brown] (4.28,1.47) -- (4.28,1.77);
\node at (1.3,0) {\includegraphics[scale=0.2]{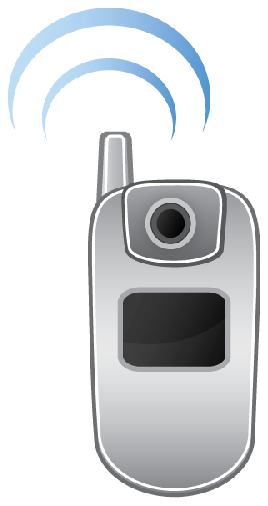}};
\node at (1.3,-0.5) {\tiny UT $k$};
\node at (0,-0.5) {\includegraphics[scale=0.2]{Mobileeps.eps}};
\node at (0.4,-0.7) {\tiny UT $1$};
\node at (-1.3,0) {\includegraphics[scale=0.2]{Mobileeps.eps}};
\node at (-1.3,-0.5) {\tiny UT $K$};
\node at (5.7,0) {\includegraphics[scale=0.2]{Mobileeps.eps}};
\node at (5.7,-0.5) {\tiny UT $K$};
\node at (4.4,-0.5) {\includegraphics[scale=0.2]{Mobileeps.eps}};
\node at (4.8,-0.7) {\tiny UT $1$};
\node at (3.1,0) {\includegraphics[scale=0.2]{Mobileeps.eps}};
\node at (3.1,-0.5) {\tiny UT $k$};
\draw[->,very thick,>=stealth,blue]  (3.9,1.52) -- (1.6,0.1);
\draw[->,very thick,>=stealth,blue]  (0.4,1.3) -- (1,0.2);
\node at (1.22, 0.6){\footnotesize $\bh_{k,1,1}$};
\node at (2.9, 1.3){\footnotesize $\bh_{k,1,2}$};
\draw[->,very thick,>=stealth,dashed,red]  (2.9,0.2) -- (0.6,1.4);
\node at (1.9, 0.9){\footnotesize $\epsilon$};
\draw[->,very thick,>=stealth,dashed,red]  (3.3,0.2) -- (4,1.3);
\node at (3.7, 0.5){\footnotesize $1$};
\end{tikzpicture}
\caption{System model}\label{system_model}
\vspace{-13pt}
\end{figure}

\subsection{MCP}\label{sec:NMIMO}
As previously mentioned,  in the MCP, both BSs share the channel information and data symbols for all users in the network. Therefore, we may consider the network as a broadcast channel with $2N$ transmit antennas and $2K$ single antenna users.   The BSs construct the precoding matrix by using the users' channel estimates. 
 In this work, we consider the RCI precoding, for which the precoding or beamforming vector for user $k$ in cell $j$, $\bw_{kj}$, can be written as \cite{Peel_tcomm05}
\[
\bw_{kj}= c \hw_{kj} =c \left(\hH^H\hH+\alpha \mathbf{I}_{2N}\right)^{-1}\hh^H_{k,j},
\]
where $\hh_{k,j} = [\hh_{k,j,1}\ \hh_{k,j,2}]$ and $\hH=[\hh_{1,1}^H \ \hh_{2,1}^H \cdots \ \hh_{K,1}^H \ \hh_{1,2}^H \ \hh_{2,2}^H \ \cdots \ \hh_{K,2}^H]^H$. The transmitted data vector can be expressed as 
\[
\bx = c \sum_{j=1}^2 \sum_{k=1}^K  \hw_{kj} s_{kj}, 
\]
where 
 $ s_{kj} \sim \mathcal{CN}(0,1)$ denotes the symbol to be transmitted to user $k$ in cell $j$. It is also assumed that the data symbols across the users are independent, i.e, $\mathbb{E}[\mb{ss}^H]=\bI_{2N}$, with $\mb{s}=[\mb{s}_1\ \mb{s}_2 ]^T$ and $\mb{s}_j=[s_{1j}\ s_{2j}\ \cdots\ s_{Kj}]^T$. $c$ is a scaling factor ensuring the total power constraint 
 is met with equality:
\[
c^2=\frac{P_t}{\TR{\left(\hH^H\hH+\alpha \mathbf{I}_{2N}\right)^{-2}\hH^H\hH}}.
\]

The received signal at user $k$ in cell $j$ can be written as
\begin{align*}
y_{kj} &= \bh_{k,j}\bx + n_{k,j}=c\bh_{k,j}\left(\hH^H\hH+\alpha \mathbf{I}_{2N}\right)^{-1}\hH^H\bs + n_{k,j}\\
				&=c\bh_{k,j}\left(\hH^H\hH+\alpha \mathbf{I}_{2N}\right)^{-1}\hh^H_{k,j}s_{k,j} + c\bh_{k,j}\left(\hH^H\hH+\alpha \mathbf{I}_{2N}\right)^{-1}\hH^H_k \bs_{k,j} + n_{k,j}\ ,
\end{align*}
where $\bh_{k,j}$ follows the channel model \eqref{eq:ch_model} with $\ih_{k,j} = [\ih_{k,j,1}\ \ih_{k,j,2}]$. The term $\hH_{k,j}$ and $\bs_{k,j}$ are obtained from $\hH$ and $\bs$ by removing the row corresponding to user $k$ in cell $j$ respectively. Hence, the SINR for user $k$ in cell $j$ can be expressed as
\begin{align}\label{eq:sinr_mcp}
\tSINR_{k,j} =\frac{c^2\left|\bh_{k,j}\left(\hH^H\hH+\alpha \mathbf{I}_{2N}\right)^{-1}\hh^H_{k,j}\right|^2}{c^2\bh_{k,j}\left(\hH^H\hH+\alpha \mathbf{I}_{2N}\right)^{-1}\hH^H_{k,j}\hH_{k,j}\left(\hH^H\hH+\alpha \mathbf{I}_{2N}\right)^{-1}\bh_{k,j}^H + \sigma^2_d}.
\end{align}

\subsection{Coordinated Beamforming}\label{sec:CBf}
In this scheme, the base stations only share the channel information, so that, for cell $j$, $\bx_j$ can be expressed as
\[
\bx_j=c_j\sum_{k=1}^K \hat{\bw}_{kj}s_{kj},
\]
where as in the MCP case $ s_{kj} \sim \mathcal{CN}(0,1)$ denotes the symbol to be transmitted to user $k$ in cell $j$. The constant $c_j$ is chosen to satisfy the per-BS power constraint, that is, $\mathbb{E}\left[\|\bx_j\|^2\right]=P_{j}$. Hence, $c_j^2=\frac{P_j}{\sum_{k=1}^K \|\hat{\bw}_{kj}\|^2}$.
We let
\[
\hw_{kj}=\left(\alpha \mathbf{I}_N +\sum_{(l,m)\neq(k,j)} \hh_{l,m,j}^H\hh_{l,m,j}\right)^{-1}\hh_{k,j,j},
\]
which is an extension of regularized zero-forcing to the coordinated beamforming setup \cite{Zakhour_it12}. Note that designing the precoding matrix at BS $j$ 
requires \emph{local} CSI only (the $\hh_{k,i,j}$ from BS $j$ to all users, but not the channels from the other BS to the users).
The SINR of user $k$ in cell $j$ can be expressed as
\be\label{eq:sinr_coord}
\tSINR_{k,j}=\frac{\displaystyle c_j^2|\bh_{k,j,j}\bw_{kj}|^2}{\displaystyle \sum_{(k',j')\neq (k,j)} c_{j'}^2 |\bh_{k,j,j'}\bw_{k'j'}|^2 + \sigma^2_d},
\ee
where, once again, $\bh_{k,j,j}$ and $\bh_{k,j,j'}$ follow \eqref{eq:ch_model}. 

\subsection{Analog Feedback through AWGN Channel}\label{ss:AF_model}
In the \textit{analog feedback} scheme, proposed in \cite{Marzetta_tsp06},  each user feeds back the CSI to the base stations using the linear analog modulation.
Since we skip quantizing and coding the channel information, we can convey this information very rapidly \cite{Marzetta_tsp06}. We also consider a simple uplink channel model, an AWGN channel.
A more realistic multiple access (MAC) uplink channel model could be a subject for future investigation. 
Each user in cell $j$ feeds back its CSI $\bh_{k,j}$ orthogonally (in time). Since each user has to transmit $2N$ symbols (its channel coefficients), it needs $2 \kappa N$ channel uses to feed back the CSI, where $\kappa \geq 1$. User $k$ in cell $j$ sends 
			\be\label{eq:anaup_signal}
					\bh_{k,j}\bL_j^\half,
			\ee 
where $\mathbf{\Lambda}_j$ is a diagonal matrix such that the first $N$ diagonal entries are equal to $\lambda_{j1}$ and the remaining diagonal entries are equal to $\lambda_{j2}$,
 with $\lambda_{jj}=2 \nu \kappa P_{u}$, 
$\lambda_{j\bar{j}}=2\epsilon^{-1}(1-\nu) \kappa P_{u}$ and $P_u$ is the user's average transmit power per channel use. 
Equation \eqref{eq:anaup_signal} satisfies the uplink power constraint $\mathbb{E}[\|\bh_{k,j}\bL_j^\half\|^2]=2\kappa N P_u$.
Thus, the power allocated to feedback the direct and interfering channel is controlled by $\nu \in [0,1]$. We should note that in \eqref{eq:anaup_signal}, it is assumed that $\kappa$ is an integer. If $\kappa N$ is an integer, we can modulate the signal \eqref{eq:anaup_signal} with $2N\times 2\kappa N$ spreading matrix \cite{Marzetta_tsp06,Caire_it10} and the analysis presented below still holds.

Now, let $b_\ell$, $\ell=1, 2,\cdots,2N$, be the $\ell$th element of $\bh_{k,j}$, $\lambda_\ell$ be the corresponding element on the diagonal of $\bL$, and $\epsilon_\ell = \mathbb{E}[b_\ell b_\ell^*]$.
When this channel coefficient is transmitted, the signal received by the coordinating BSs is
\vspace{-1mm}
\begin{align*}
\by_\ell =\sqrt{\lambda_\ell}\ba{c} \b1_N \\ \sqrt{\epsilon}\b1_N\ea b_\ell + \mathbf{n}_u= \sqrt{\lambda_\ell}\mathbf{p} b_\ell + \mathbf{n}_u,
\end{align*}
where $\mathbf{n}_u \in \mathbb{C}^{2N\times 1} \sim \mathcal{CN}(\mathbf{0},\sigma^2_u\bI)$ is the noise vector at the coordinating BSs and $\b1_N$ is a column vector of length $N$ with all 1 entries.
Using  the fact that the path-gain from  the users in cell $j$ to BS $\bar{j}$ is $\epsilon$,
 the MMSE estimate of $b_\ell$ becomes 
\vspace{-1mm}
\[
\hat{b}_\ell=\sqrt{\lambda_\ell}\epsilon_\ell \bp^T\left[\lambda_\ell\epsilon_\ell \bp\bp^T + \sigma^2_u\bI_{2N} \right]^{-1}\by_\ell,
\]
and its MMSE is $\sigma^2_{b_\ell}=\epsilon_\ell -\lambda_\ell\epsilon_\ell^2 \bp^T \left[\lambda_\ell\epsilon_\ell \bp\bp^T + \sigma^2_u\bI_{2N} \right]^{-1} \bp.$
We should note that $\{\hat{b}_\ell\}$ are mutually independent. By using the property of MMSE estimation, we can express 
$\bh_{k,j,i}$ as
\vspace{-1mm}
\be\label{eq:ch_model_AF}
\bh_{k,j,i}=\hh_{k,j,i}+ \ih_{k,j,i},
\ee
where $\hh_{k,j,i}$ represents the channel estimate, and $\ih_{k,j,i}$ the channel uncertainty or estimation error. Note that the entries of each vector $\hh_{k,i,j}$ and $\ih_{k,i,j}$ are independent and identically distributed (i.i.d) and  distributed according to $\mathcal{CN}(0,\omega_{ji})$ and $\mathcal{CN}(0,\delta_{ji})$, respectively, where 
\be\label{eq:def_omega_delta}
\delta_{ji} =\begin{cases}\frac{1}{1+\nu\bar{\gamma}_u  }, & j=i \\
\frac{\epsilon}{1+(1-\nu)\bar{\gamma}_u  }, & j\neq i, 
\end{cases}, \quad 
\omega_{ji} =\begin{cases}\frac{\nu\bar{\gamma}_u }{1+\nu\bar{\gamma}_u  }, & j=i \\
 \frac{\epsilon (1-\nu)\bar{\gamma}_u}{1+(1-\nu)\bar{\gamma}_u  }, & j\neq i, 
\end{cases}
\ee
and $\bar{\gamma}_u=2\gamma_u \kappa (1+\epsilon)$ with $\gamma_u=N P_u/\sigma^2_u$. 
The channel estimates are used by the BSs to construct the precoder. Since each $\delta_{ij}$ and $\omega_{ij}$ are identical for all users then we denote
$\delta_d=\delta_{jj}, \delta_c=\delta_{j\bar{j}}, \omega_d=\omega_{jj}$ and $\omega_c=\omega_{j\bar{j}}$. From \eqref{eq:def_omega_delta}, it follows that $\omega_d=1-\delta_d$ and $\omega_c=\epsilon-\delta_c$.


\subsection{Quantized Feedback via RVQ}\label{ss:RVQ_model}
In the digital feedback case, user $k$ in cell $j$ uses $B_{k,j,j}$ and $B_{k,j,\bar{j}}$ bits to quantize or feedback the direct and interfering channels, respectively. The total number of feedback bits is assumed to be fixed. It is also assumed that each user has different codebooks: $\mathcal{U}_{k,j,j}$ with size $2^{B_{k,j,j}}$ and $\mathcal{U}_{k,j,\bar{j}}$ with size $2^{B_{k,j,\bar{j}}}$,  to quantize the direct and interfering channel, respectively.  Moreover, these codebooks are different for each user. In this work, 
$B_{k,j,j}$ is the same for all users and $B_{k,j,j}=B_d, \forall k, j = 1, 2$. 
Similarly, $B_{k,j,\bar{j}}=B_c, \forall k, j = 1, 2$. The total number of feedback bits is denoted by $B_t$, where $B_t=B_d+B_c$.

Since the optimal codebook design for the quantized feedback is not known yet, therefore in this paper, for  analytical tractability, we consider the well known RVQ scheme.  
As suggested by its name, RVQ uses a random vector quantization codebook where the quantization vectors in the codebook are  independently chosen from the isotropic distribution on the $N$-dimensional unit sphere \cite{Jindal_it06, Santipach_it09}. The codebook is known by the base station and the user. The user quantizes its channel by finding the quantization vector in the codebook which is closest to its channel vector and feedbacks the index of the quantization vector to the BSs. We should note that only the channel direction is quantized. Most of the works that employ RVQ for the feedback model assume that only channel direction information is sent to the BSs.  
 As mentioned in \cite{Jindal_it06}, the channel norm information can also be used for some problems that need channel quality information (CQI) such as power allocation across the channel and users scheduling \cite{Ravindran_icc08}. 

The user $k$ in cell $j$ finds its quantization vector for the channel $\bh_{k,j,i}$ according to
\[
\hat{\bu}_{k,j,i}=\arg \underset{\bu_{k,j,i}\in \ \mathcal{U}_{k,j,i}}{ \max}\ \frac{|\bh_{k,j,i}\bu_{k,j,i}^H|}{\|\bh_{k,j,i}\|}. 
\] 	
The quantization error or distortion $\tau^2_{k,j,i}$ is defined as
\[
\tau^2_{k,j,i}=1-\frac{\left\|\bh_{k,j,i}\hat{\bu}_{k,j,i}\right\|^2}{\|\bh_{k,j,i}\|^2}=\sin^2\left(\angle\left(\bh_{k,j,i}/\|\bh_{k,j,i}\|,\hat{\bu}_{k,j,i}\right)\right).   
\]
It is a random variable whose distribution is equivalent to the minimum of $2^{B_{k,j,i}}$ beta random variables with parameters $N-1$ and 1\ (see \cite{Jindal_it06,Yeung_wc07}). Each realization of $\tau_{k,j,i}$ is different for each user even though the users have the same amount of feedback bits.

Having obtained $\hat{\bu}_{k,j,i}$, each user then sends its index in the codebook and also the channel magnitude $\|\bh_{k,j,i}\|$ (see also \cite{Ravindran_icc08}). By assuming that the BSs can receive the information perfectly, the channel estimate at the BS can be written as 
\be
\hh_{k,j,i}=\|\bh_{k,j,i}\|\hat{\bu}_{k,j,i}.
\ee 
Note that  $\hh_{k,j,i}$ has the same statistical distribution as  $\bh_{k,j,i}$ i.e., $\hh_{k,j,i} \sim \mathcal{CN}(0,\epsilon_{ji}\bI_N)$, where $\epsilon_{ji}=1$ when $i=j$ and otherwise, $\epsilon_{ji}=\epsilon$.

From \cite{Jindal_it06, Jindal_wcom08},  we can model $\bh_{k,j,i}$ as follows
\be\label{eq:chmodel_rvq}
\bh_{k,j,i}=\sqrt{1-\tau^2_{k,j,i}}\hh_{k,j,i} + \tau_{k,j,i}\|\bh_{k,j,i}\|\mathbf{z}_{k,j,i}, 
\ee
where $\mathbf{z}_{k,j,i}$ is isotropically distributed in the null-space of $\hat{\bu}_{k,j,i}$ and
is independent of $\tau_{k,j,i}$. Moreover, $\mathbf{z}_{k,j,i}$ can be rewritten  as
\[
\mathbf{z}_{k,j,i}=\frac{\bv_{k,j,i}\bPi_{\hh_{k,j,i}}^\bot}{\|\bv_{k,j,i}\bPi_{\hh_{k,j,i}}^\bot\|},
\]
where $\bPi_{\hh_{k,j,i}}$ is the projection matrix in the column space of $\hh_{k,j,i}$, $\bPi_{\hh_{k,j,i}}^\bot=\bI_N-\frac{\hh_{k,j,i}^H\hh_{k,j,i}}{\|\hh_{k,j,i}\|^2}$ and $\bv_{k,j,i} \sim \mathcal{CN}(\mathbf{0},\bI_N)$ is independent of $\hh_{k,j,i}$. 
It is clear that the channel model \eqref{eq:chmodel_rvq} has the same structure as \eqref{eq:ch_model} with $\phi_{k,j,i}=1-\tau^2_{k,j,i}$ and $\ih=\tau_{k,j,i}\|\bh_{k,j,i}\|\mathbf{z}_{k,j,i}$. 

\subsection{Achievable and limiting sum-rate}
Besides $\tSINR_{k,j}$, another relevant performance measure is the achievable rate. For the user $k$ at cell $j$, it is defined as
		\be\label{eq:def_rate}
				R_{k,j} =  \log_2(1+\text{SINR}_{k,j}),
		\ee
It is obtained by treating the interferences as noise or equivalently performing single-user decoding at the receiver. Observing \eqref{eq:def_rate}, it is obvious that there is a one-to-one continuous mapping between the SINR and the achievable rate (see also \cite{Tse_it99}). The total sum-rate, or just the sum-rate, can then be defined as follows
   \be\label{eq:def_sumrate}
				R_\text{sum} = \sum_{j=1}^2\sum_{k=1}^{K} R_{kj}.
		\ee

As shown later in Section \ref{sec:ana_feedback} and \ref{sec:RVQ_feedback}, as $K,N \to \infty$, we have
\be\label{eq:conv_SINR}
\tSINR_{kj} - \tSINR^\infty \to 0,
\ee
where $\tSINR^\infty$ is a deterministic quantity and also called  the limiting SINR. It is also shown that the limiting SINR is the same for all users.
By using the result \eqref{eq:conv_SINR} and based on the continuous mapping theorem \cite{Vaart_book98}, the following
 \[
\frac{1}{2N}\mathbb{E}\left[R_\text{sum}\right] - R_\text{sum}^\infty \to 0
\]
holds (see also \cite{Wagner_it12}) where the limiting achievable sum-rate can be expressed as $R_\text{sum}^\infty = \beta\log_2(1+\tSINR^\infty)$. For the numerical simulations, we also introduce
the normalized sum-rate difference, defined as
\be\label{eq:def_ThDiff}
\Delta R_\text{sum} =\frac{\frac{1}{2N}\mathbb{E}\left[R_\text{sum}\right]- R_\text{sum}^\infty}{\frac{1}{2N}\mathbb{E}\left[R_\text{sum}\right]},
\ee
that quantifies the sum-rate difference, $\frac{1}{2N}\mathbb{E}\left[R_\text{sum}\right] - R_\text{sum}^\infty$, compared to the (actual) finite-size system average sum-rate. 

\section{MCP and CBf with Noisy Analog Feedback}\label{sec:ana_feedback}
In this section, we will discuss the large system results and feedback optimization for the MCP and CBf by using the analog feedback model discussed in Section \ref{ss:AF_model}.  First, the large system limit expression for the SINR  is derived. Then, the corresponding optimal regularization parameter that maximizes the limiting SINR is investigated. Finally, the optimal $\nu$ that maximizes the limiting SINR that already incorporates the optimal regularization parameter will be discussed.
 
\subsection{MCP}\label{ss:Network_MIMO_AF}
We start with the theorem that states the large system limit of the SINR \eqref{eq:sinr_mcp}.
\begin{theorem}\label{th:sinr_mcp}
Let $\rho_{\text{\tiny M,AF}}=(\omega_d+\omega_c)^{-1}\alpha/N$ and $g(\beta,\rho)$ be the solution of $g(\beta,\rho)=\left(\rho + \frac{\beta}{1+g(\beta,\rho)}\right)^{-1}$. In the large system limit, the SINR of MCP given in \eqref{eq:sinr_mcp} converges in probability to a deterministic quantity
given by
\be\label{eq:lim_sinr_mcp}
\tSINR^\infty_{\text{\tiny MCP,AF}} = \gamma_e g(\beta,\rho_{\text{\tiny M,AF}})\frac{1+\frac{\rho_{\text{\tiny M,AF}}}{\beta}(1+g(\beta,\rho_{\text{\tiny M,AF}}))^2}{\gamma_e+\left(1+g(\beta,\rho_{\text{\tiny M,AF}})\right)^2},
\ee
where the effective SNR $\gamma_e$ is expressed as
\be\label{eq:def_effSNR_MCPanalog}
\gamma_e=\frac{\omega_d+\omega_c}{\delta_d+\delta_c+\frac{1}{\gamma_d}}=\frac{1-\delta_d+\epsilon-\delta_c}{\delta_d+\delta_c+\frac{1}{\gamma_d}}.
\ee
\end{theorem} 
\begin{IEEEproof}
See Appendix \ref{App:SINR_MCP}
\end{IEEEproof}

It is obvious from above that the limiting SINR is the same for all users in both cells. This is due to the channel statistics of all users in both cells are the same.
The channel uncertainty, captured by $\omega_\bullet$ and $\delta_\bullet$, affects the system performance (limiting SINR) via the effective SNR and regularization parameter  $\rho_{\text{\tiny M,AF}}$. 

As discussed previously, the (effective) regularization parameter $\rho_{\text{\tiny M,AF}}$ controls the amount of interference introduced to the users and provides the trade-off between suppressing the inter-user interference and increasing desired signal energy.  The optimal choice of $\rho_{\text{\tiny M,AF}}$ that maximizes \eqref{eq:lim_sinr_mcp} is given in the following.
\begin{corollary}
The optimal $\rho_{\text{\tiny M,AF}}$ that maximizes $\tSINR^\infty_{\text{\tiny MCP,AF}} $ is 
\be\label{eq:opt_rhoM}
\rho_{\text{\tiny M,AF}}^*=\frac{\beta}{\gamma_e},
\ee 
and the corresponding limiting SINR is
\be\label{eq:optSINR_MCP_AF_funcRho}
\tSINR^{*,\infty}_{\text{\tiny MCP,AF}} =g(\beta,\rho_{\text{\tiny M,AF}}^*).
\ee
\end{corollary}
\begin{IEEEproof}
The proof follows easily from \cite{RusdhaPrep}.
\end{IEEEproof}
It is interesting to see that the limiting SINR expression with $\rho_{\text{\tiny M,AF}}^*$ becomes simpler and it depends only the cell-loading ($\beta$) and the effective SNR. Clearly from \eqref{eq:def_effSNR_MCPanalog}, $\gamma_e$ is a function of the total MSE, $\delta_t=\delta_d+\delta_c$, that can be thought as a reasonable measure of the CSIT quality. Thus, $\rho_{\text{\tiny M,AF}}^*$ adjusts its value as   $\delta_t$ changes. Now, from \eqref{eq:def_effSNR_MCPanalog}, it is obvious that $\gamma_e$ is a decreasing function of $\delta_t$. As a result,  $\rho_{\text{\tiny M,AF}}^*$ is increasing with $\delta_t$. In other words, if the total quality of CSIT improves then the regularization parameter becomes smaller. In the perfect CSIT case, i.e., $\delta_t=0$, and in the high SNR regime, $\rho_{\text{\tiny M,AF}}^*$ goes to zero and we have the ZF precoder.   

Now, we will investigate how to allocate $\nu$ to maximize the limiting SINR \eqref{eq:optSINR_MCP_AF_funcRho}, or equivalently $g(\beta,\rho_{\text{\tiny M,AF}}^*)$. 
$\nu$ is captured by $\gamma_e$ (or  $\rho_{\text{\tiny M,AF}}^*$) via $\delta_d$. It can be shown that $g$ is decreasing (increasing) in $\rho_{\text{\tiny M,AF}}$  ($\gamma_e$). Then, for a fixed $\beta$ the limiting SINR is maximized by solving the following optimization problem
\[
\underset{\nu \in [0,1] }{\max} \quad \gamma_e=\frac{\epsilon-\delta_c+ 1-\delta_d}{(\delta_d+\delta_c)+\frac{1}{\gamma_d}}. 
\]
As mentioned earlier, $\gamma_e$ is a decreasing function of $\delta_t$. Thus, the optimization problem above can be rewritten as
\be\label{opt:nu_min_dist_MCP}
\underset{\nu \in [0,1] }{\min} \quad \delta_t=\delta_d+\delta_c=\frac{1}{\nu\bar{\gamma}_u+1}+\frac{\epsilon}{(1-\nu)\bar{\gamma}_u+1}.
\ee
From the above, it is very interesting to note that \emph{the optimal $\nu$ that maximizes $\tSINR^{*,\infty}_{\text{\tiny MCP}}$ is the same as the one that minimizes the total MSE, $\delta_t$}. 

It is easy to check that the optimization problem above is a convex program and the optimal $\nu$, denoted by $\nu^*$, can be expressed as follows
\be\label{eq:opt_nu_mcp}
\nu^*=\begin{cases}
0, & \sqrt{\epsilon} \geq \bar{\gamma}_u + 1\\
1, & \sqrt{\epsilon} \leq \frac{1}{\bar{\gamma}_u + 1}\\
\frac{1+\frac{1}{\bar{\gamma}_u}(1-\sqrt{\epsilon})}{1+\sqrt{\epsilon}}, & \text{otherwise.}
\end{cases}
\ee
As a result, for $\sqrt{\epsilon} \le \frac{1}{\bar{\gamma}_u+1}$,
the BSs should not waste resources trying to learn about the ``interfering" channel states. In this situation, \textit{the coordination breaks down} and the base stations perform SCP.
The completely opposite scenario, in which the BSs should not learn the ``direct" channels, occurs when $\sqrt{\epsilon} \geq \bar{\gamma}_u + 1$. Clearly, this can only happen if $\epsilon > 1$. 
When $\sqrt{\epsilon} \geq \bar{\gamma}_u + 1$, the BSs also perform SCP but each BS transmits to the users in the neighboring cell. 

We end this subsection by characterizing the behavior of $\gamma_e$ (equivalently $\tSINR^{*,\infty}_{\text{\tiny MCP}}$), after optimal feedback power allocation, as the cross channel gain $\epsilon$ varies. 
This also implicitly shows how the total MSE, $\delta_t$, affects the limiting SINR.
%
Let $\bryu=\frac{\bar{\gamma}_u}{(1+\epsilon)}$. We analyze the different cases in \eqref{eq:opt_nu_mcp} separately.
\subsubsection{\texorpdfstring{$\sqrt{\epsilon} \leq \frac{1}{\bar{\gamma}_u + 1}$}{Gamma VS epsilon case1}}
This is the case when the BSs perform SCP for the users in their own cell. For fixed $\bryu$, this inequality is equivalent  to $\epsilon \le \epsilon^{\text{\tiny SCP}}_{\max}$, where $\epsilon^{\text{\tiny  SCP}}_{\max} \ge 0$ satisfies  
$\sqrt{\epsilon^{\text{\tiny  SCP}}_{\max}} = \frac{1}{\bryu (1+\epsilon^{\text{\tiny  SCP}}_{\max}) + 1}$.
Now, by taking the first derivative  $\frac{\partial \gamma_e}{\partial \epsilon}$ and setting it to zero, the (unique) stationary point is given by
\[
\epsilon^{\text{\tiny SCP}}_{\text{\tiny AF}} = \frac{1}{\sqrt{\gamma_d\bryu}}-1. 
\]

If $\sqrt{\epsilon^{\text{\tiny SCP}}_{\text{\tiny AF}}} \in [0,\sqrt{\epsilon^{\text{\tiny  SCP}}_{\max}}]$, 
it is easy to check that the limiting SINR is increasing until $\epsilon=\epsilon^{\text{\tiny SCP}}_{\text{\tiny AF}}$ and then decreasing. If $ \sqrt{\gamma_d\bryu} > 1 $ then  $\epsilon^{\text{\tiny SCP}}_{\text{\tiny AF}}<0$, or equivalently, $\frac{\partial \gamma_e}{\partial \epsilon} < 0$. Consequently, for this case, the limiting SINR is decreasing in $\epsilon$. Moreover, $\sqrt{\epsilon^{\text{\tiny SCP}}_{\text{\tiny AF}}} \ge \sqrt{\epsilon^{\text{\tiny  SCP}}_{\max}}$ if the following condition holds
\be\label{eq:cond_limSINR_incr}
\sqrt{\gamma_d\bryu}(2-2\gamma_d -\bryu) \geq (2\gamma_d\bryu -\gamma_d -\bryu),
\ee
in which case $\frac{\partial \gamma_e}{\partial \epsilon} > 0$, 
 which implies that the limiting SINR always increases over $\epsilon$ for this case. 

This behavior of $\gamma_e$ as a function of $\epsilon$ can be intuitively explained as follows. When $\nu=1$, the total MSE is $\delta_t=\frac{1}{(1+\epsilon)\bryu+1}+\epsilon$, where the first and second terms are $\delta_d$ and $\delta_c$, respectively. 
 As $\epsilon$ increases, $\delta_d$ decreases whereas $\delta_c$ increases. This shows that there is a trade-off between the quality of the direct channel and the strength of the interference. The trade-off is also influenced by parameters $\gamma_d$ and $\bryu$. As shown in the analysis,
when $\sqrt{\gamma_d\bryu} > 1$, the effect of cross channel to the limiting SINR dominates. In contrast, if the condition in \eqref{eq:cond_limSINR_incr} is satisfied, the effect of the quality of the direct channel ($\delta_t$) becomes dominant. If the aforementioned conditions do not hold, $\delta_t$ causes the SINR to increase until $\epsilon^{\text{\tiny SCP}}_{\text{\tiny AF}}$ and after that the interference from the cross channel takes over as the dominant factor, thereby reducing the limiting SINR. 

\subsubsection{\texorpdfstring{$\bar{\gamma}_u + 1 \geq \sqrt{\epsilon} \geq \frac{1}{\bar{\gamma}_u + 1}$}{Gamma VS epsilon Case2}}
Here, the BSs perform MCP. By taking $\frac{\partial \gamma_e}{\partial \epsilon}$ in that interval of $\epsilon$, it can be shown that we have a unique stationary which we denote as $\sqrt{\epsilon_\text{\tiny AF}^\text{\tiny M}}$. We can also show that $\gamma_e$ is a convex function for $\epsilon \in [0,1]$ and is increasing for $\epsilon \geq 1$. Thus, if $\frac{1}{\bar{\gamma}_u + 1} \leq \sqrt{\epsilon_\text{\tiny AF}^\text{\tiny M}} \leq \bar{\gamma}_u + 1 $, the limiting SINR will decrease for $\sqrt{\epsilon} \in [\frac{1}{\bryu(1+\epsilon)+1},\sqrt{\epsilon_\text{\tiny AF}^\text{\tiny M}}]$ and increase  after that; Otherwise, the limiting SINR increases in the region. Here, for $\sqrt{\epsilon} \in [ \frac{1}{\bar{\gamma}_u + 1},1]$, we still can see the effect of the trade-off within $\delta_t$ to the limiting SINR as $\epsilon$ changes. In that interval, the quality of the direct channel becomes better as $\epsilon$ increases; However, that of the cross channel decreases and this affects the SINR badly until $\epsilon_\text{\tiny AF}^\text{\tiny M}$. After this point, 
the improvement in the quality of the direct channel will outweigh the deterioration of that of the cross channel, causing the SINR to increase. 

\subsubsection{\texorpdfstring{$\sqrt{\epsilon} \geq \bar{\gamma}_u + 1$}{Gamma VS epsilon Case3}}
In this case, each BS performs SCP, but serves the other cell's users. We can establish that $\frac{\partial \gamma_e}{\partial \epsilon}> 0$. Hence, for this case, the limiting SINR is increasing in $\epsilon$.

\subsection{Coordinated Beamforming}\label{ss:Network_MIMO_CBf}
\begin{theorem}\label{th:sinr_coord}
Let $\rho_{\text{\tiny C,AF}} = \frac{\alpha}{N}$, and let $\Gamma_A$ be the solution of the following cubic equation
\begin{align}
\Gamma_A=\frac{1}{\rho_{\text{\tiny C,AF}}+\frac{\beta\omega_c}{1+\omega_c\Gamma_A}+\frac{\beta\omega_d}{1+\omega_d\Gamma_A}}.\label{eq:Gamma_AF}
\end{align}
In the large system limit, the SINR of the coordinated beamforming given in \eqref{eq:sinr_coord} converges almost surely to a deterministic quantity
given by
\begin{align}\label{eq:limSINR_coord_AF}
\tSINR^{\infty}_{\text{\tiny CBf,AF}}
&= \frac{\frac{\omega_d}{\beta}\Gamma_A\left[\rho_{\text{\tiny C,AF}}+\frac{\beta\omega_c}{(1+\omega_c\Gamma_A)^2}+\frac{\beta\omega_d}{(1+\omega_d\Gamma_A)^2}\right]}{\left(\frac{1}{\gamma_d}+\delta_{d} + \delta_{c}+\frac{\omega_{d}}{(1+\omega_{d}\Gamma_A)^2}+\frac{\omega_{c}}{(1+\omega_{c}\Gamma_A)^2}\right)}.
\end{align}
\end{theorem} 
\begin{IEEEproof}
See Appendix \ref{App:limSINR_CBf_AF}
\end{IEEEproof}

Similar to the MCP case, the limiting SINR expression \eqref{eq:limSINR_coord_AF} is the same for all users. By comparing \eqref{eq:opt_rhoM} and \eqref{eq:opt_rhoC}, it is also interesting to see that $\rho_{\text{\tiny C,AF}} = \rho_{\text{\tiny M,AF}}$ for a given $\alpha$.  
The optimal $\rho_{\text{\tiny C,AF}}$ that maximizes the limiting SINR \eqref{eq:limSINR_coord_AF} is given in the following. 
\begin{corollary}\label{col:opt_rhoC}
The limiting SINR \eqref{eq:limSINR_coord_AF} is maximized by choosing the regularization parameter according to
\be\label{eq:opt_rhoC}
\rho_{\text{\tiny C,AF}}^*=\beta\left(\frac{1}{\gamma_d}+\delta_d+\delta_c\right).
\ee
and the corresponding limiting SINR is 
\be\label{eq:opt_limSINR_coord}
\tSINR^{*,\infty}_{\text{\tiny CBf,AF}}=\omega_d\Gamma_A^*,
\ee 
where $\Gamma_A^*$ is $\Gamma_A$ with $\rho_{\text{\tiny C,AF}}=\rho_{\text{\tiny C,AF}}^*$.
\end{corollary}
\begin{IEEEproof}
Let $\gamma_{eff}=\beta\left(\gamma_d^{-1}+\delta_{d} + \delta_{c}\right)$ and $\Psi=\frac{\beta\omega_{d}}{(1+\omega_{d}\Gamma_A)^2}+\frac{\beta\omega_{c}}{(1+\omega_{c}\Gamma_A)^2}$.
It is easy to show that
\begin{align}\label{eq:der_sinr_rhoC}
\frac{\partial \tSINR^{\infty}_{\text{\tiny CBf,AF}}}{\partial \rho_{\text{\tiny C,AF}}} = \omega_d\frac{\gamma_{eff}-\rho_{\text{\tiny C,AF}}}{[\gamma_{eff}+\Psi]^2}\frac{\partial\Psi}{\partial\rho_{\text{\tiny C,AF}}},
\end{align}
where $\frac{\partial\Psi}{\partial\rho_{\text{\tiny C,AF}}}=-2\beta\frac{\partial\Gamma_A}{\partial\rho_{\text{\tiny C,AF}}}\left(\frac{\omega_d^2}{(1+\omega_d\Gamma_A)^3}+\frac{\omega_c^2}{(1+\omega_c\Gamma_A)^3}\right) > 0$ with $\frac{\partial\Gamma_A}{\partial\rho_{\text{\tiny C,AF}}} < 0$ is given by \eqref{eq:der_gamma_rho_CBf_AF}.
Thus, it follows that $\rho_{\text{\tiny C,AF}}^*=\gamma_{eff}$ is the unique stationary point and the global optimizer. Plugging back $\rho_{\text{\tiny C,AF}}$  into \eqref{eq:limSINR_coord_AF} yields \eqref{eq:opt_limSINR_coord}.
\end{IEEEproof}

Similar to the MCP case, the corollary above shows that the optimal regularization parameter adapts to the changes of CSIT quality and 
it is a decreasing function of $\delta_t$. 

Finding $\nu$ that maximizes the limiting SINR of the CBf is more complicated than in the MCP case. It is equivalent to maximizing $\omega_d\Gamma$, such that $\nu\in [0,1]$: this is a non-convex program. However, the maximizer $\nu^*$ is one of followings: 
the boundaries of the feasible set ($\nu=\{0,1\}$) or the stationary point, denoted by $\nu^\circ$,
which is the solution of 
\be\label{eq:nuStat_coord}
\nu^\circ= - \frac{\Gamma_A^*}{\frac{\partial \Gamma_A^*}{\partial \rho_{\text{\tiny C,AF}}^*}(1+\nu^\circ\bar{\gamma}_u)}.
\ee
The point $\nu=0$ can be eliminated from the feasible set since the derivative of the limiting SINR with respect to $\nu$ at this point is always positive.

\subsection{Numerical Results}

Since propagation channels fluctuate, the SINR expressions in \eqref{eq:sinr_mcp} and \eqref{eq:sinr_coord} are random quantities. Consequently, the average sum-rates are also random. Figure \ref{fig:conv_sumrate_analog} illustrates how the random average sum-rates approach the limiting sum-rates as the dimensions of the system increase. This is quantified by the normalized  sum-rate difference which is defined in \eqref{eq:def_ThDiff}.
The average sum-rate is obtained by  averaging the sum-rates over 1000 channel realizations. The optimal regularization parameter and power splitting obtained in the large system analysis are used in computing the limiting and average sum-rates. 
We can see that as the system size increases, the normalized sum-rate difference becomes smaller and this hints
that the approximation of the average sum-rate by the limiting sum-rate becomes more accurate. The difference is already about $1.3\%$ and $0.5\%$ for the MCP and CBf respectively for $N=60,K=36$. 


\begin{figure}[t]
\centering
\begin{tikzpicture}
	\pgfplotsset{grid style=dashed}
	\begin{axis}[small, width=11cm, xlabel=$N$, ylabel=$\Delta R_\text{sum}$,
	yticklabel style={/pgf/number format/fixed},
	xmin=1, xmax=61, ymin=0,
	grid=major
	 ]
		\addplot table {av_ThdiffvsN_MCP_analog.dat};
		\addplot[color=red, mark=square*, mark size=1.5] table {av_ThdiffvsN_CBf_analog.dat};
		\legend{\footnotesize MCP, \footnotesize CBf}
	\end{axis}
\end{tikzpicture}
\caption{The normalized sum-rate difference for different system dimensions with $\beta=0.6$, $\epsilon =0.5$ $\gamma_d = 10$ dB and $\gamma_u = 0$ dB.}
\label{fig:conv_sumrate_analog}
\end{figure}
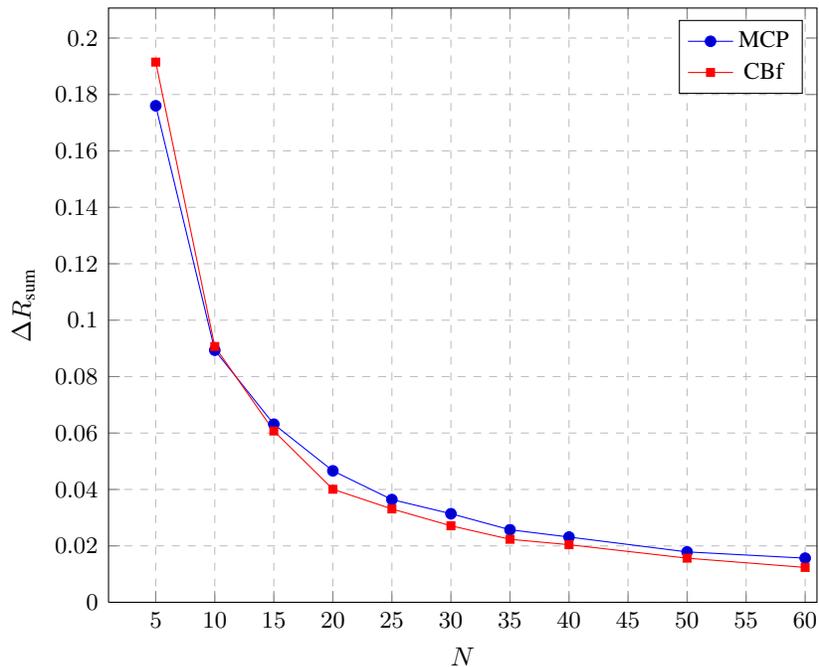

Figure \ref{fig:ThDiff_analog} describes the applicability of the large system results into finite-size systems. We choose a reasonable system-size in practice, i.e., $N=10, K=6$. Then, 250 channel realizations are generated. For each channel realization, with a fixed regularization parameter of the precoder, the optimal $\nu$, denoted by $\nu_\text{\tiny FS}^*$, is computed. Then the resulting average sum-rate is compared to the average sum-rate that using $\nu^*$ from the large system analysis, i.e., \eqref{eq:opt_nu_mcp} and \eqref{eq:nuStat_coord}, for different values of $\epsilon$. We can see that the normalized average sum-rate difference, i.e., $\frac{\mathbb{E}\left[|R_\text{sum}(\nu_\text{\tiny FS})-R_\text{sum}(\nu^*)|\right]}{R_\text{sum}(\nu_\text{\tiny FS})}$ , for CBf has a peak around $4\%$ that can be considered as a reasonable value for the chosen system size. For MCP, it is less than $0.47\%$. To this end, our simulation results indicate that the large system results discussed earlier approximate the finite-system quite well.

\begin{figure}[t]
\centering
\begin{tikzpicture}
	\pgfplotsset{grid style=dashed}
	\begin{semilogyaxis}[small, width=11cm, xlabel=$\epsilon$, ylabel=$\frac{\mathbb{E}[ R_{\text{sum}}(\nu_{\text{FS}}^*) - R_\text{sum}(\nu^*)]}{\mathbb{E}[R_{\text{sum}}(\nu_{\text{FS}}^*)]}$,
	xmin=0, xmax=1, ymin=0, ymax=0.13,
	minor x tick num=1,
	xmajorgrids, yminorgrids,
	 ]
		\addplot table {av_ThdiffvsEpsil_MCP_analog.dat};
		\addplot[color=red, mark=square*, mark size=1.5] table {av_ThdiffvsEpsil_CBf_analog.dat};
		\legend{\footnotesize MCP, \footnotesize CBf}
	\end{semilogyaxis}
\end{tikzpicture}
\caption{The normalized average sum-rate difference of the finite-size system by using the $\nu_\text{FS}$ and $\nu^*$  with $N=10, \beta=0.6$, $\gamma_d = 10$ dB and $\gamma_u = 0$ dB.}%
\label{fig:ThDiff_analog}
\end{figure}
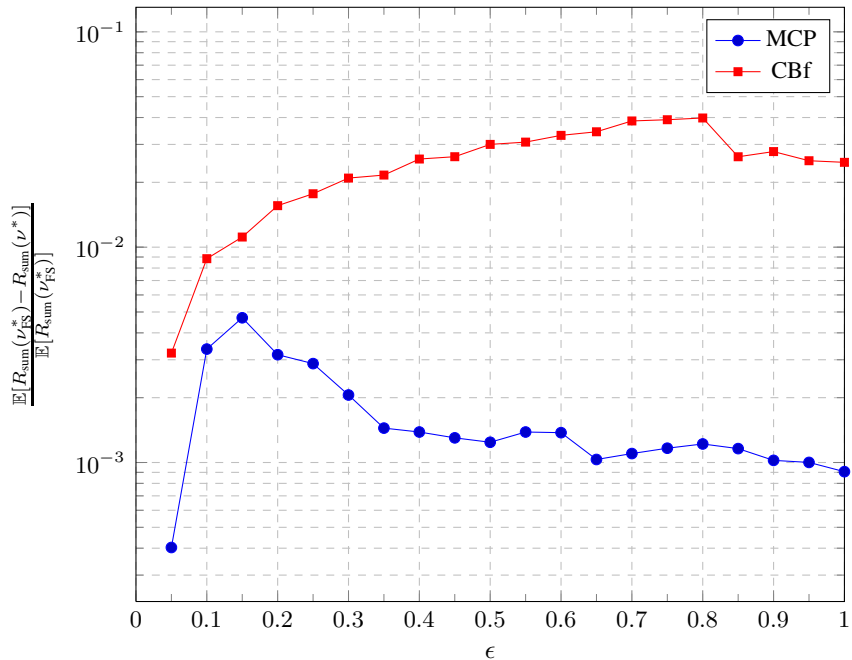

%

In the following, we present some numerical simulations that visualize the characteristics of the optimal $\nu^*$ (in the large system limit) and the corresponding limiting SINR for each cooperation scheme. We are primarily interested in their characteristics when the interfering channel gain $\epsilon$ varies, as depicted in Figure \ref{fig1}. 
In general, we can see that for the same system parameters, the CBf scheme allocates more power to feed back the direct channel compared to the MCP. 
From Figure \ref{fig1}(a), we can see that for values of $\epsilon$ ranging from $0$ up to a certain threshold (denoted by $\epsilon_\text{\tiny M}^\text{\tiny th}=\epsilon_{\max}^\text{\tiny SCP}$ and $\epsilon_\text{\tiny C}^\text{\tiny th}$ for MCP and CBf respectively), the optimal $\nu$  is 1: in other words, it is optimal in this range for the BSs not to try to get information about the cross channels and to construct the precoder based on the direct channel information only. Effectively, the two schemes reduce to the SCP scheme when $\nu^* = 1$: as a result, the same limiting SINR is achieved by both schemes. 

In Figure \ref{fig1}(b), we can observe a peculiar behavior of the limiting SINR of MCP, which we already highlighted in the analysis of Section \ref{ss:Network_MIMO_AF}.  When $\sqrt{\epsilon} \le \frac{1}{\bar{\gamma}_u+1}$, i.e. when  $\nu^*=1$, the SINR is decreasing as $\epsilon$ increases. After that the SINR is still decreasing until $\epsilon$ reaches $\epsilon^\text{\tiny M}_\text{\tiny AF}$ and then increasing: this reflects the trade-off between $\delta_c$ and $\delta_d$. Note that this initial decrease does not occur in the perfect CSI case where the SINR is strictly increasing in $\epsilon$ for MCP.
Similar to the MCP case, we can see that the limiting SINR of CBf is decreasing in $\epsilon$ when $\nu^*=1$ (SCP). Moreover, it is still decreasing when both BSs perform CBf. 

\begin{figure}[t]
\centering
\begin{tabular}{cc}
\begin{tikzpicture}[baseline]
	\pgfplotsset{grid style=dashed}
	\begin{axis}[small, width=8.5cm,  xlabel=$\epsilon$, ylabel=$\nu^*$,
	yticklabel style={/pgf/number format/fixed}, ytick={0.4,0.5,...,1.1},
	xmin=0, xmax=1, ymin=0.4, ymax=1.15,
	grid=major,	no markers, legend pos=south west
	 ]
		\addplot[color=blue, line width=1] table {nu_vs_epsil_beta06_yup0_MCP.dat};
		\addplot[color=red, dashed, line width=1] table {nu_vs_epsil_beta06_yup0_CBf.dat};
		\legend{\footnotesize MCP, \footnotesize CBf}
		\draw[->,thick,>=stealth] (axis cs:0.17,1.08) -- (axis cs:0.1,1);
		\node at (axis cs:0.2,1.1) {\footnotesize $\epsilon_\text{\tiny M}^\text{\scriptsize th}$};
		\draw[->,thick,>=stealth] (axis cs:0.88,1.08) -- (axis cs:0.81,1);
		\node at (axis cs:0.91,1.1) {\footnotesize $\epsilon_\text{\tiny C}^\text{\scriptsize th}$};
	\end{axis}
\end{tikzpicture}
&
\begin{tikzpicture}[baseline]
	\pgfplotsset{grid style=dashed}
	\begin{axis}[small, width=8.5cm, xlabel=$\epsilon$, ylabel=$\tSINR^\infty$,
	yticklabel style={/pgf/number format/fixed},
	xmin=0, xmax=1, ymin=0.6,
	grid=major, no markers, legend pos=south west
	 ]
		\addplot[color=blue, line width=1] table {SINR_vs_epsil_beta06_yup0_MCP.dat};
		\addplot[color=red, dashed, line width=1] table {SINR_vs_epsil_beta06_yup0_CBf.dat};
		\legend{\footnotesize MCP, \footnotesize CBf}
		\draw[->,thick,>=stealth] (axis cs:0.3,1.5) -- (axis cs:0.25,1.362);
		\node at (axis cs:0.33,1.55) {\footnotesize $\epsilon_\text{\tiny AF}^\text{\tiny M}$};
	\end{axis}
\end{tikzpicture}\\
$\text{\footnotesize (a)}$ & $\text{\footnotesize (b)}$
\end{tabular}
\caption{
(a) The optimal $\nu^*$ and (b) the limiting SINR for the MCP and CBf scheme as $\epsilon$ varies in $[0, 1]$ with $\beta=0.6$, $\gamma_d=10 \text{ dB}$, $\gamma_u=0 \text{ dB}$.}%
\label{fig1}%
\end{figure}
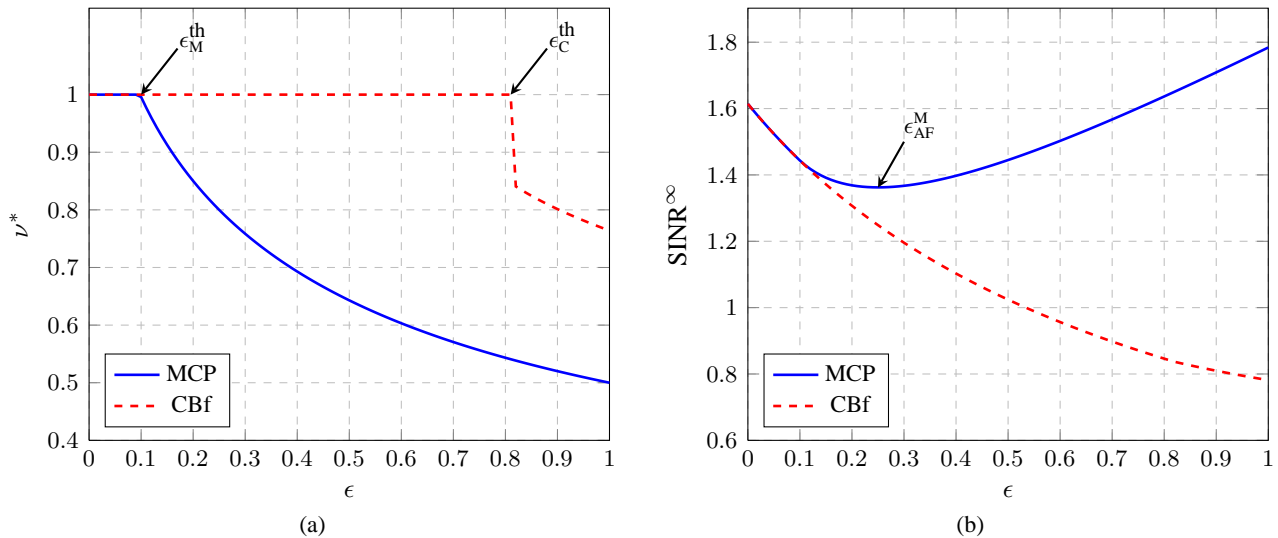

\section{Quantized Feedback  via Random Vector Quantization (RVQ)}\label{sec:RVQ_feedback}
In this section, we will derive the approximations of the SINR for the MCP \eqref{eq:sinr_mcp} and CBf \eqref{eq:sinr_coord} by analyzing them in the large system limit.
We use these approximations to optimize the feedback bit allocation, and regularization parameter. This joint optimization problem can be split into two steps. 
First, we derive the optimal bit allocation, i.e., the optimal $\bar{B}_d=\frac{B_d}{N}$ and $\bar{B}_c=\frac{B_c}{N}$. Plugging the optimal bit allocation back into  the limiting SINR expression, we can then proceed to the second step where we obtain the optimal regularization parameter. At the end of the section, some comparisons of the limiting SINR and bit allocation values for the two schemes are illustrated.

\subsection{MCP}\label{ss:LSA_MCP_RVQ}
\begin{theorem}\label{th:sinr_mcp_rvq}
Let $\rho_{\text{\tiny M,Q}}=(1+\epsilon)^{-1}\alpha/N$ and $g(\beta,\rho)$ be the solution of $g(\beta,\rho)=\left(\rho + \frac{\beta}{1+g(\beta,\rho)}\right)^{-1}$. In the large system limit, the SINR converges in probability to a deterministic quantity
given by
\be\label{eq:limSINR_MCP_RVQ}
	\tSINR_\text{\tiny MCP,Q}^\infty = \gamma_e g(\beta,\rho_\text{\tiny M,Q}) \frac{1+\frac{\rho_\text{\tiny M,Q}}{\beta}(1+g(\beta,\rho_\text{\tiny M,Q}))^2}{\gamma_e+(1+g(\beta,\rho_\text{\tiny M,Q}))^2},
\ee
where
\be\label{eq:eff_snr_mcp_rvq}
\gamma_e = \frac{d^2}{1-d^2+\frac{1}{\gamma_d(1+\epsilon)}}
\ee
is defined as the effective SNR and
\be\label{eq:d}
d=\frac{\sqrt{1-2^{-\bar{B}_d}}+\epsilon\sqrt{1-2^{-\bar{B}_c}}}{1+\epsilon}.
\ee
\end{theorem}


\begin{proof}
Refer to Appendix \ref{App:limSINR_mcp_rvq}.
\end{proof}
   
Theorem \ref{th:sinr_mcp_rvq} shows that the limiting SINR is the same for all users in both cells. 
This is not surprising given the symmetry in their channel statistics and feedback mechanisms. 
Moreover, the only dependence of the limiting SINR on the bit allocation is via $\gamma_e$, which itself is a function of $d$:
 $d$ can be interpreted as a measure of the total quality of the channel estimates;
In fact, given that $\bar{B}_d$ and $\bar{B}_c$ are constrained to sum up to $\bar{B}_t$, $d$ in \eqref{eq:d} highlights a trade-off between increasing feedback bits for direct channel and cross channel. Comparing \eqref{eq:lim_sinr_mcp} and \eqref{eq:limSINR_MCP_RVQ}, we can immediately recognize an identical structure between them. The effective SNR expressions \eqref{eq:def_effSNR_MCPanalog} and \eqref{eq:eff_snr_mcp_rvq} also share a similar construction, where $(1+\epsilon)d^2$ in \eqref{eq:d} can be thought to be equivalent to $\omega_d+\omega_c$.

Now, we move tho the first step of the joint optimization i.e., determining the optimal bit allocation that maximizes \eqref{eq:limSINR_MCP_RVQ}. It is clear from \eqref{eq:limSINR_MCP_RVQ} that $\bar{B}_d$ and $\bar{B}_c$ contributes to the limiting SINR through $d$.
It is easy to check that the limiting SINR is an increasing and a convex function of  $d$. Thus, maximizing $\tSINR^\infty_{\text{\tiny MCP,Q}}$ is equivalent to maximizing $d$, i.e. solving (cf. Eq. \eqref{eq:d}). 
\be\label{OP:xd_mcp}
\underset{x_d \in [X_t,1]}{\max} \sqrt{1-x_d}+\epsilon\sqrt{1-\frac{X_t}{x_d}}.
\ee
where $X_t=2^{-\bar{B}_t}$, $\bar{B}_t=\frac{B_t}{N}$ and $x_d=2^{-\bar{B}_d}$.
The solution of   \eqref{OP:xd_mcp} is presented in the following theorem.
\begin{theorem}\label{th:optbit_MCP}
 $\tSINR^\infty_{\text{\tiny MCP,Q}}$ is maximized by allocating 
$\bar{B}_d=-\log_2(x_d^*)$ bits to feed back the direct channel information, and $\bar{B}_c = \bar{B}_t - \bar{B}_d$ to feed back the interfering channel information,  where $x_d^*$ is the positive (real) solution of the following quartic equation
\be\label{eq:optbit_MCP}
x_d^4-X_tx_d^3 + (\epsilon X_t)^2(x_d-1) = 0.  
\ee 
\end{theorem}
\begin{IEEEproof}
The first derivative of the objective function over $x_d$ is given by
\be\label{eq:derEd_MCP}
(1+\epsilon)\frac{\partial \mathbb{E}[d]}{\partial x_d}=\half\left(-\frac{1}{\sqrt{1-x_d}}+\frac{1}{x_d^2}\frac{\epsilon X_t}{\sqrt{1-\frac{X_t}{x_d}}}\right)
\ee
and $\lim_{x_d \to X_t} \frac{\partial \mathbb{E}[d]}{\partial x_d} = \infty, \lim_{x_d \to 1} \frac{\partial \mathbb{E}[d]}{\partial x_d} = -\infty$. Moreover, the objective function is a concave function in $x_d$ since
\[
(1+\epsilon)\frac{\partial^2 \mathbb{E}[d]}{\partial x_d^2}=\half\left(-\half(1-x_d)^{-3/2}-\frac{2}{x_d^3}\frac{\epsilon X_t}{\sqrt{1-\frac{X_t}{x_d}}}-\half\frac{\epsilon X_t}{x_d^4}\left(1-\frac{X_t}{x_d}\right)^{-3/2}\right) < 0, \quad x_d \in [X_t,1].
\]
The stationary point, $x_d^*$, is obtained by setting the derivative equal to 0 and it is the non-negative  solution of 
\[
x_d^4-X_tx_d^3 + (\epsilon X_t)^2(x_d-1) = 0.
\] 
Since the objective function is concave over $x_d$, then $x_d^*$ gives the global optimum point.   
\end{IEEEproof}

Now, let us discuss how the optimal bit allocation vary with $\epsilon$. Since $x_d=x_d^*$ satisfies \eqref{eq:optbit_MCP}, then by taking the (implicit) derivative of \eqref{eq:optbit_MCP} w.r.t. $\epsilon$, we have
\[
\frac{\partial x_d^*}{\partial \epsilon}=\frac{2\epsilon X_t^2(1-x_d)}{4x_d^3 -3X_tx_d^2 +(\epsilon X_t)^2} > 0, \quad \text{for }X_t \leq x_d^* \leq 1.
\]  
This implies that as $\epsilon$ increases, $x_d^*$ ($\bar{B}_d^*$) increases (decreases). This is consistent with the intuition that for higher $\epsilon$,  more resources would be allocated to quantize the cross channel information. At one of the extremes, i.e., $\epsilon = 0$, $x^*_d=X_t$, or $\bar{B_d}=\bar{B}_t$. 
If $\epsilon = 0$, $x^*_d=X_t$, so that when there is no interference from the neighboring BS, all feedback bits are used to convey the direct channel states, as expected.
At the other extreme, when $\epsilon \to \infty$, $x_d^* \to 1$ or $\bar{B}_d \to 0$.
This can be shown by setting the derivative \eqref{eq:derEd_MCP} equal to zero and we have
\[
\frac{1}{\epsilon}=\frac{X_t\sqrt{1-x_d}}{x_d^2\sqrt{1-\frac{X_t}{x_d}}}.
\] 
As $\epsilon \to \infty$, the left hand side goes to zero and the stationarity is achieved by setting $x_d=1$. 

It is also interesting to see how $d$, after optimal bit allocation, behaves as the cross channel gain varies.
Let $d^*$ is $d$ evaluated at $x_d=x_d^*$. By taking $\frac{\partial d^*}{\partial \epsilon}$, we can show the following property.
\begin{proposition}\label{Prop:MCP_RVQ_dVsepsil}
For $\epsilon \leq 1$, $d^*$ is decreasing in $\epsilon$ and increasing for $\epsilon \geq 1$. Consequently, $d^*$ is minimum at $\epsilon=1$. 
\end{proposition}
As mentioned previously, $x_d^*$ increases and consequently $1-x_d^*$ decreases as $\epsilon$ increases. On the other side, $\epsilon\sqrt{1-X_t/x_d^*}$ is getting larger.
So, from the calculation we can conclude that $d^*$ is mostly affected by $\sqrt{1-x_d^*}$ for $\epsilon \leq 1$, while for the other values of $\epsilon$, the other term takes over.

We now proceed to find the optimal $\rho_\text{\tiny M,Q}$ that maximizes $\tSINR^\infty_{\text{\tiny MCP,Q}}$. The result is summarized below.
\begin{theorem}\label{th:opt_rho_MCP_RVQ}
Let $\gamma_e^*$ be $\gamma_e$ evaluated at $d=d^*$. The optimal $\rho_{\text{\tiny M}}$ that maximizes $\tSINR^\infty_{\text{\tiny MCP}}(d^*)$ is 
\be\label{eq:opt_rho_MCP_RVQ}
\rho_{\text{\tiny M,Q}}^*=\frac{\beta}{\gamma_e^*}.
\ee
The corresponding limiting SINR is given by
\[
\tSINR^{*,\infty}_{\text{\tiny MCP}}=g\left(\beta,\rho_{\text{\tiny M,Q}}^*\right).
\]
\end{theorem}
\begin{IEEEproof}
The equation \eqref{eq:limSINR_MCP_RVQ} has the same structure as \eqref{eq:lim_sinr_mcp} and thus, \eqref{eq:opt_rho_MCP_RVQ} follows. 
\end{IEEEproof}

From Theorem \ref{th:opt_rho_MCP_RVQ}, $d^*$ affects the regularization parameter and the limiting SINR via effective SNR $\gamma_e^*$. The latter grows with $d^*$ (cf. \eqref{eq:eff_snr_mcp_rvq}). Thus, $\rho_{\text{\tiny M,Q}}^*$ declines as the CSIT quality, $d^*$, increases and this behavior is also observed for the cooperation schemes with the analog feedback.

In Proposition \ref{Prop:MCP_RVQ_dVsepsil}, we established how $d^*$ changes with $\epsilon$. We can show that $\gamma_e^*$ has a similar behavior but reaches its minimum at a different value of $\epsilon$ due to the last term in the denominator in \eqref{eq:eff_snr_mcp_rvq}. 
For $\tSINR^{*,\infty}_{\text{\tiny MCP}}$, it attains its minimum at $\epsilon=\epsilon_{\text{\tiny Q}}^\text{\tiny M}$, as described in the next proposition. 
\begin{proposition}\label{Prop:MCP_RVQ_rhovsepsil}
Suppose that $\epsilon=\epsilon_{\text{\tiny Q}}^\text{\tiny M}$ satisfies
\be\label{eq:epsilTh_MCP_RVQ}
(x_d^*)^2=\frac{\gamma_d(1+\epsilon)-\half}{\epsilon X_t\left[\gamma_d(1+\epsilon)+1+\frac{\epsilon}{2}\right]}.
\ee
Then,  $\tSINR^{*,\infty}_{\text{\tiny MCP,Q}}$ decreasing for $\epsilon \leq \epsilon_{\text{\tiny Q}}^\text{\tiny M}$ and  increasing for $\epsilon \geq \epsilon_{\text{\tiny Q}}^\text{\tiny M}$.
\end{proposition}

 The characterization of $\tSINR^{*,\infty}_{\text{\tiny MCP,Q}}$ above reminds us a similar behavior of $\tSINR^{*,\infty}_{\text{\tiny MCP,AF}}$ after optimal power allocation. We can conclude that \textit{the limiting SINR of MCP under both feedback schemes has a common behavior as $\epsilon$ varies}. 
 
	
\subsection{Coordinated Beamforming}
\begin{theorem}\label{th:sinr_coord_rvq}
Let $\rho_\text{\tiny C,Q}=\alpha/N$ and $\Gamma_{Q}$ be the solution of the following cubic equation
\begin{align}
\Gamma_Q =  \frac{1}{\rho_\text{\tiny C,Q} + \frac{\beta}{1+\Gamma_Q}+\frac{\beta\epsilon}{1+\epsilon\Gamma_Q}}.
\end{align}
Let $\phi_d=1-2^{-\bar{B}_d}$, $\phi_c=1-2^{-\bar{B}_c}$, $\delta_d=2^{-\bar{B}_d}$ and $\delta_c=\epsilon 2^{-\bar{B}_c}$. In the large system limit, the SINR \eqref{eq:sinr_coord} for the quantized feedback via RVQ converges weakly to a deterministic quantity given by
\begin{align}\label{eq:limSINR_CBf_RVQ}
\tSINR^\infty_\text{\tiny CBf,Q} = -\frac{\phi_d\Gamma_Q^2}{\beta\left(\frac{1}{\gamma_d}+\frac{\phi_d}{(1+\Gamma_Q)^2}+\frac{\phi_c\epsilon}{(1+\epsilon\Gamma_Q)^2}+\delta_d + \delta_c\right)\frac{\partial \Gamma_Q}{\partial \rho}}
\end{align}
where
\[
-\frac{\partial \Gamma_Q}{\partial \rho_\text{\tiny C,Q}} =\frac{\Gamma_Q}{\rho_\text{\tiny C,Q}+\frac{\beta\epsilon}{(1+\epsilon\Gamma_Q)^2}+\frac{\beta}{(1+\Gamma_Q)^2}}.
\]
\end{theorem}
\begin{IEEEproof}
See Appendix \ref{App:limSINR_CBf_RVQ}
\end{IEEEproof}

As in Theorem \ref{th:sinr_mcp_rvq}, Theorem \ref{th:sinr_coord_rvq} shows that that the limiting SINR is the same for all users.
The quantization error variance of estimating the direct channel, $\delta_d$, affects both the signal  strength (via $\phi_d$) and the interference energy, in which  it captures the effect of the \emph{intra-cell} interference. 
$\delta_c$, on the other hand, only contributes to the interference term: It represents the quality of the cross channel and  determines the strength of the \emph{inter-cell} interference. Since $\bar{B}_t$ is fixed, increasing $\bar{B}_d$, or equivalently reducing $\bar{B}_c$, will strengthen the desired signal and reduce the intra-cell interference: it does so, however, at the expense of strengthening the inter-cell interference. Thus, feedback bits' allocation is important in order to improve the performance of the system. 
To solve the joint optimization problem, it is useful to write \eqref{eq:limSINR_CBf_RVQ} as follows
\[
\tSINR^\infty_\text{\tiny CBf,Q} =G_1\frac{1-x_d}{\frac{1}{\gamma_d}+(1-x_d)(G_2-1)+\epsilon\left(1-\frac{X_t}{x_d}\right)(G_3-1)+1+\epsilon},
\]
where $x_d$ and $X_t$ are defined as in the previous subsection and for brevity, we denote: $G_1=-\Gamma_Q^2\left(\beta\frac{\partial \Gamma_Q}{\rho_{\text{\tiny C,Q}}}\right)^{-1}, G_2=(1+\Gamma_Q)^{-2}$ and $G_3=(1+\epsilon\Gamma_Q)^{-2}$. The optimal bit allocation can be found by solving the following optimization problem.
\be\label{eq:optProb_bitAlloc_RVQ}
\underset{x_d \in [X_t,1]}{\max.} \quad \tSINR^\infty_\text{\tiny CBf,Q} .
\ee
The solution of \eqref{eq:optProb_bitAlloc_RVQ} is summarized in the following theorem.
\begin{theorem}
For a fixed $\bar{B_t}$, the optimal bit allocation, in term of $x_d=2^{-\bar{B}_d}$, that maximizes $\tSINR^\infty_\text{\tiny CBf,Q} $ is given by
\be\label{eq:opt_bit}
x_d^*=\begin{cases}
X_t, & \epsilon \leq \frac{X_t(\frac{1}{\gamma_d}+1)}{1-G_3-X_t(2-G_3)}=\epsilon_\text{th} \\
X_d=\frac{\epsilon X_t (G_3-1) + \sqrt{\epsilon^2X_t^2(G_3-1)^2-\epsilon X_t\left(\frac{1}{\gamma_d}+1+\epsilon G_3\right)(G_3-1)}}{\frac{1}{\gamma_d}+1+\epsilon G_3}, & \text{otherwise}.
\end{cases}
\ee
\end{theorem}
\begin{IEEEproof}
Differentiating the objective function \eqref{eq:optProb_bitAlloc_RVQ}, we get
\[
\frac{\partial \tSINR^\infty_\text{\tiny CBf,Q}}{\partial x_d}=G_1\frac{-x_d^2(\frac{1}{\gamma_d}+1+\epsilon G_3)+\epsilon(G_3-1)(2X_tx_d - X_t)}{x_d^2\left(\frac{1}{\gamma_d}+(1-x_d)(G_2-1)+\epsilon\left(1-\frac{X_t}{x_d}\right)(G_3-1)+1+\epsilon\right)^2}
\]
and the stationary is given by
\begin{align}
x_d^\circ = \frac{\epsilon X_t (G_3-1) + \sqrt{\epsilon^2X_t^2(G_3-1)^2-\epsilon X_t(\frac{1}{\gamma_d}+1+\epsilon G_3)(G_3-1)}}{\frac{1}{\gamma_d}+1+\epsilon G_3}. \label{eq:opt_xd1}
\end{align}
Now let us consider the term $Z=-x_d^2(\frac{1}{\gamma_d}+1+\epsilon G_3)+\epsilon(G_3-1)(2X_tx_d - X_t)$ in the numerator. 
It can be verified that the sign of $Z$ is the same as the sign of $\frac{\partial \tSINR^\infty_\text{\tiny CBf,Q}}{\partial x_d}$.
Thus, $X_d=x_d^\circ$ will be the unique positive solution of the quadratic equation $Z=0$. 

It can be also checked that $\frac{\partial Z}{\partial x_d}=-2x_d(\frac{1}{\gamma_d}+1+\epsilon G_3)+\epsilon(G_3-1)(2X_t) < 0$ and thus, $Z$ is decreasing in $x_d$. Since at $x_d=1$, $Z < 0$, we should never allocate $x_d^*=1$. We will allocate $x_d=X_t$ if $Z \leq 0$ at $x_d=X_t$ (this condition is satisfied whenever $\epsilon \leq \epsilon_\text{th}$).
\end{IEEEproof}

Unlike the MCP case where $x_d^*=X_t$ \emph{only when $\epsilon=0$}, in the CBf, it is optimal for a user to allocate all $B_t$ to the direct channel when $0 \leq \epsilon \leq \epsilon_\text{th}$. 
Note that 
$x_d^*=X_t$ does not imply that the cooperation breaks down or that both BSs perform single-cell processing. It is easy to check that $\epsilon_\text{th}$ increases when $\bar{B}_t$ or $\gamma_d$ is decreased. This suggests that when the resource for the feedback bits is scarce or the received SNR is low then it is preferable for the user to allocate all the feedback bits to quantize the direct channel. So, in this situation, quantizing the cross channel does more harm to the performance the system.    
However, as $\epsilon$ increases beyond $\epsilon_\text{th}$, quantizing the cross channel will improve the SINR. We can show that $x_d^*$, particularly $X_d$,  is increasing  in $\epsilon$. 
In doing that, we need to take the derivative of $X_d$ over $\epsilon$. It is easy to show that $\Gamma_Q$ is decreasing in $\epsilon$. Then, it follows that $G_3$ is decreasing in $\epsilon$.
Using this fact, we can then show $\frac{\partial X_d}{\partial \epsilon} > 0$. 
So, as in the case of MCP, this suggests that more resources are allocated to  feedback the cross-channel when $\epsilon$ increases.

Once we have the optimal bit allocation, we can find the optimal $\rho_\text{C,Q}$, as we did for the MCP. For that purpose, we can rewrite \eqref{eq:opt_bit} w.r.t $\rho_\text{\tiny C,Q}$ as follows
\[
x_d^*=\begin{cases}
X_t, & \rho_\text{\tiny C,Q} \geq \rho_\text{th} \\
X_d, & \text{otherwise}.
\end{cases}
\]
where for given $X_t,\epsilon$ and $\gamma_d$,  the threshold $\rho_\text{th}$ satisfies $\epsilon=\epsilon_\text{th}$.
So, we have $\tSINR^\infty_\text{\tiny CBf,Q} (X_d)$ for $\rho_\text{\tiny C,Q} < \rho_\text{th}$ and $\tSINR^\infty_\text{\tiny CBf,Q} (X_t)$ for other values of $\rho_\text{\tiny C,Q}$. 

Now, let us investigate the optimal $\rho_\text{\tiny C,Q}$ when $x_d^* = X_d$. By evaluating $\frac{\partial \tSINR^\infty_\text{\tiny CBf,Q} (X_d)}{\partial \rho_\text{\tiny C,Q}}=0$, we can determine the stationary point, which is given by
\begin{align*}
\rho^\circ_{X_d}&=\left\{(1-X_d)\left((G_2'+\epsilon G_3')\left[\frac{1}{\gamma_d}+X_d+\epsilon\frac{X_t}{X_d}\right] +\epsilon (G_2G_3'-G_3G_2')\left[-X_d + \frac{X_t}{X_d}\right]\right)  -X_d'(G_2+\epsilon G_3)\gamma_e  \right\} \\
&\quad \times \frac{\beta}{X_d'\gamma_e+(1-X_d)\left((1-X_d)G_2'+\epsilon \left(1-\frac{X_t}{X_d}\right)G_3'\right)},
\end{align*}
where $\gamma_e=\frac{1}{\gamma_d}+1+\epsilon+\epsilon (G_3-1)\left(1-\frac{2X_t}{X_d}+\frac{X_t}{X_d^2}\right)$ and $G_2'=\frac{\partial G_2}{\partial \rho_\text{\tiny C,Q}}$
and $G_3'=\frac{\partial G_3}{\partial \rho_\text{\tiny C,Q}}$. 
 
We can show that the derivative is positive for $\rho_\text{\tiny C,Q} \in [0,\rho^\circ_{X_d})$ and negative for $\rho_\text{\tiny C,Q} \in (\rho^\circ_{X_d},\infty)$. Since $\tSINR^\infty_\text{\tiny CBf,Q} (X_d)$ is defined for $\rho_\text{\tiny C,Q} \leq \rho_\text{th}$, if $\rho^\circ_{X_d} < \rho_\text{th}$ then  $\tSINR^\infty_\text{\tiny CBf,Q} (X_d)$ is increasing for $\rho_\text{\tiny C,Q} \in [0,\rho^\circ_{X_d}]$ and decreasing for $\rho_\text{\tiny C,Q} \in [\rho^\circ_{X_d},\rho_\text{th})$.  If $\rho^\circ_{X_d} \geq \rho_\text{th}$ then $\tSINR^\infty_\text{\tiny CBf,Q} (X_d)$ is increasing for $\rho_\text{\tiny C,Q} \in [0,\rho_\text{th})$. 

Then, we move to the case when $x_d^*=X_t$. 
%
By setting $\frac{\partial \tSINR^\infty_\text{\tiny CBf,Q} (X_t)}{\partial \rho_\text{\tiny C,Q}}=0$, the stationary point  is then given by
\begin{align*}
\rho^\circ_{X_t} &=
\beta\frac{(G_2'+\epsilon G_3')(X_t+1/\gamma_d+\epsilon)+\epsilon(1-X_t)( G_2 G_3'- G_3 G_2')}{(1-X_t)G_2'}.
\end{align*}

We can also show that the derivative is positive for $\rho_\text{\tiny C,Q} \in [0,\rho^\circ_{X_t})$ and negative for $\rho_\text{\tiny C,Q} \in (\rho^\circ_{X_t},\infty)$. Since $\tSINR^\infty_\text{\tiny CBf,Q} (X_t)$ is defined for $\rho_\text{\tiny C,Q} \geq \rho_\text{th}$, if $\rho^\circ_{X_t} > \rho_\text{th}$ then  $\tSINR^\infty_\text{\tiny CBf,Q} (X_t)$ is increasing for $\rho_\text{\tiny C,Q} \in [\rho_\text{th},\rho^\circ_{X_t}]$ and decreasing for $\rho_\text{\tiny C,Q} \in [\rho^\circ_{X_t},\infty)$.  If $\rho^\circ_{X_t} \leq \rho_\text{th}$ then $\tSINR^\infty_\text{\tiny CBf,Q} (X_t)$ is decreasing for $\rho_\text{\tiny C,Q} \in [\rho_\text{th},\infty)$.

In what follows, by knowing the stationary point in both regions of $\rho$, we will investigate how to obtain the optimal $\rho_\text{\tiny C,Q}$, denoted by $\rho_\text{\tiny C,Q}^*$, for $\rho_\text{\tiny C,Q} \in [0,\infty)$. By inspecting $\partial \tSINR^\infty_\text{\tiny CBf,Q} (X_d)/\partial\rho_\text{\tiny C,Q}$ and $\partial \tSINR^\infty_\text{\tiny CBf,Q} (X_t)/\partial\rho_\text{\tiny C,Q}$  we can see that that $\tSINR^\infty_\text{\tiny CBf,Q} (x_d^*)$ is continuously differentiable for the region, $\rho_\text{\tiny C,Q} \in [0,\rho_\text{th})$ and
$\rho_\text{\tiny C,Q} \in [\rho_\text{th},\infty)$, respectively. To show $\tSINR^\infty_\text{\tiny CBf,Q} (x_d^*)$ is continuously differentiable for $\rho_\text{\tiny C,Q} \in [0,\infty)$ we need to establish  $\partial \tSINR^\infty_\text{\tiny CBf,Q} (x_d^*)/\partial\rho_\text{\tiny C,Q}$ to be continuous at $\rho_\text{\tiny C,Q}=\rho_\text{th}$, or equivalently 
\be\label{eq:cont_diff_sinr}
\lim_{\rho_\text{\tiny C,Q} \to \rho_\text{th}^-} \frac{\partial \tSINR^\infty_\text{\tiny CBf,Q} (X_d)}{\partial \rho_\text{\tiny C,Q}} = \lim_{\rho_\text{\tiny C,Q} \to \rho_\text{th}^+} \frac{\partial \tSINR^\infty_\text{\tiny CBf,Q} (X_t)}{\partial \rho_\text{\tiny C,Q}}=\left. \frac{\partial \tSINR^\infty_\text{\tiny CBf,Q} (X_t)}{\partial \rho_\text{\tiny C,Q}}\right|_{\rho_\text{\tiny C,Q}=\rho_\text{th}}
\ee 
When $\rho_\text{\tiny C,Q} \to \rho_\text{th}^-$, $X_d \to X_t$ and therefore the denominator of $\partial \tSINR^\infty_\text{\tiny CBf,Q} (X_d)/\partial\rho_\text{\tiny C,Q}$ and $\partial \tSINR^\infty_\text{\tiny CBf,Q} (X_t)/\partial\rho_\text{\tiny C,Q}$ are equal. Let $\mathcal{N}(f)$ denote the numerator of $f$. As $X_d \to X_t$, we have
\begin{align*}
\lim_{\rho_\text{\tiny C,Q} \to \rho_\text{th}^-}\mathcal{N}\left(\partial \tSINR^\infty_\text{\tiny CBf,Q} (X_d)/\partial\rho_\text{\tiny C,Q} \right) &= \left[\beta(G_2'+\epsilon G_3')(1/\gamma_d+1+\epsilon - (1-X_t))+\beta\epsilon (1-X_t)(G_2G_3' - G_3 G_2') \right. \\
&\qquad \left. -\rho_\text{th}(1-X_t)G_2'\right]\Gamma_Q (1-X_t) - \lim_{\rho_\text{\tiny C,Q} \to \rho_\text{th}^-} X_d'(\beta G_2+\beta\epsilon G_3 +\rho)\gamma_e \\
 &= \mathcal{N}\left(\left.\frac{\partial \tSINR^\infty_\text{\tiny CBf,Q} (X_t)}{\partial\rho_\text{\tiny C,Q}}\right|_{\rho_\text{\tiny C,Q}=\rho_\text{th}}\right)- \lim_{\rho_\text{\tiny C,Q} \to \rho_\text{th}^-} X_d'(\beta G_2+\beta\epsilon G_3 +\rho)\gamma_e,
\end{align*}
 where $X_d'=\partial X_d/\partial \rho_\text{\tiny C,Q}$. We should note that $\lim_{\rho_\text{\tiny C,Q} \to \rho_\text{th}^-} X_d'= -\half\epsilon G_3'\frac{1-X_t}{\frac{1}{\gamma_d}+1+\epsilon}\neq 0$. This shows that $x_d^*$ is not continuously differentiable over $\rho_\text{\tiny C,Q}$. Now let us look at
\begin{align*}
\lim_{\rho_\text{\tiny C,Q} \to \rho_\text{th}^-} \gamma_e &= \frac{1}{\gamma_d}+1+\epsilon+\epsilon (G_3-1)\left(-1+\frac{1}{X_t}\right)\\
&= \frac{1}{X_t}\left[\left(\frac{1}{\gamma_d}+1\right)X_t +\epsilon(2X_t-1)-\epsilon G_3 (X_t-1)\right]=0
\end{align*}
since as $\rho_\text{\tiny C,Q} \to \rho_\text{th}^-$, from the (equivalent) condition $\epsilon=\epsilon_\text{th}$, the term in the bracket becomes 0.
This concludes \eqref{eq:cont_diff_sinr} and therefore  $\tSINR^\infty_\text{\tiny CBf,Q} (x_d^*)$ is continuously differentiable for $\rho_\text{\tiny C,Q} \in [0,\infty)$.

By using the property above and the facts that the  $\tSINR^\infty_\text{\tiny CBf,Q} (X_d)$ and $\tSINR^\infty_\text{\tiny CBf,Q} (X_t)$ are quasi-concave (unimodal), we can determine the optimal $\rho_\text{\tiny C,Q}^*$ and $x_d^*$ jointly as described in Algorithm \ref{alg:comp_rho_cbf_rvq}.
\begin{algorithm}
\caption{Calculate $\rho_\text{\tiny C,Q}^*$ and $x_d^*$} \label{alg:comp_rho_cbf_rvq}
\begin{algorithmic}[1] 
\STATE Compute $\rho_\text{th}$
\IF{$\rho_\text{th} \leq 0$}
	\STATE $x_d^*=X_t$.
	\STATE $\rho_\text{\tiny C,Q}^*=\rho_{X_t}^\circ$
\ELSE
	\STATE Compute $\rho_{X_t}^\circ$
	\IF{$\rho_{X_t}^\circ \geq \rho_\text{th}$}
			\STATE $x_d^*=X_t$
			\STATE $\rho_\text{\tiny C,Q}^*=\rho_{X_t}^\circ$
	\ELSE
			\STATE $x_d^*=X_d$
			\STATE $\rho_\text{\tiny C,Q}^*=\rho_{X_d}^\circ$
	\ENDIF
\ENDIF
\end{algorithmic}
\end{algorithm}
We can verify the steps 6-13 in the algorithm by using the following arguments:  If $\rho^\circ_{X_t} > \rho_\text{th}$, then the derivate of
$\tSINR^\infty_\text{\tiny CBf,Q} (X_t)$ is positive at $\rho_\text{\tiny C,Q}=\rho_\text{th}$ because  $\tSINR^\infty_\text{\tiny CBf,Q} (X_t)$ is quasi-concave. Since the $\tSINR^\infty_\text{\tiny CBf,Q} (x_d^*)$ is continuously differentiable, then the derivative of $\tSINR^\infty_\text{\tiny CBf,Q} (X_d)$ is also positive when $\rho_\text{\tiny C,Q} \to \rho_\text{th}$. Since $\tSINR^\infty_\text{\tiny CBf,Q} (X_d)$ is also quasi-concave, consequently $\tSINR^\infty_\text{\tiny CBf,Q} (X_d)$ is increasing for $\rho_\text{\tiny C,Q} \in  [0,\rho_\text{th})$. This implies that $\rho^*_\text{\tiny C,Q}= \rho^\circ_{X_t}$.
Similar types of arguments can be also used to verify that if  $\rho^\circ_{X_t} < \rho_\text{th}$ then $\rho^*_\text{\tiny C,Q}= \rho^\circ_{X_d}$.

\subsection{Numerical Results}

The first two figures in this section are obtained by using a similar procedure to that followed in the analog feedback case.
Figure \ref{fig:conv_sumrate_rvq} shows how well the limiting sum-rate (equivalently the limiting SINR) approximates the finite-size system sum-rate. The optimal regularization parameter and bit allocation are applied in computing the limiting and average sum-rates.
As $N$ grows, the normalized sum-rate difference become smaller. For $N=60, K=36$, it is arleady about $3.1\%$ and $1.6\%$ for MCP and CBf, respectively. Figure \ref{fig:ThDiff_rvq} shows the average sum-rate difference, with a fixed regularization parameter, between the system that uses $B_{d,\text{FS}}^*$ and $\bar{B_d}^*$ to feed back the direct channel states. $B_{d,\text{FS}}^*$ denotes the optimal bit allocation of the finite-size system. For each channel realization, it is obtained by a grid search. With $N=10, K=6$, the maximum normalized average sum-rate difference reaches $0.22\%$ for MCP. It is about four-times bigger for CBf, which is approximately $0.86\%$.  Thus, from those simulations, similar to the analog feedback case, the conclusions we can reach for the limiting regime are actually useful for the finite system case.

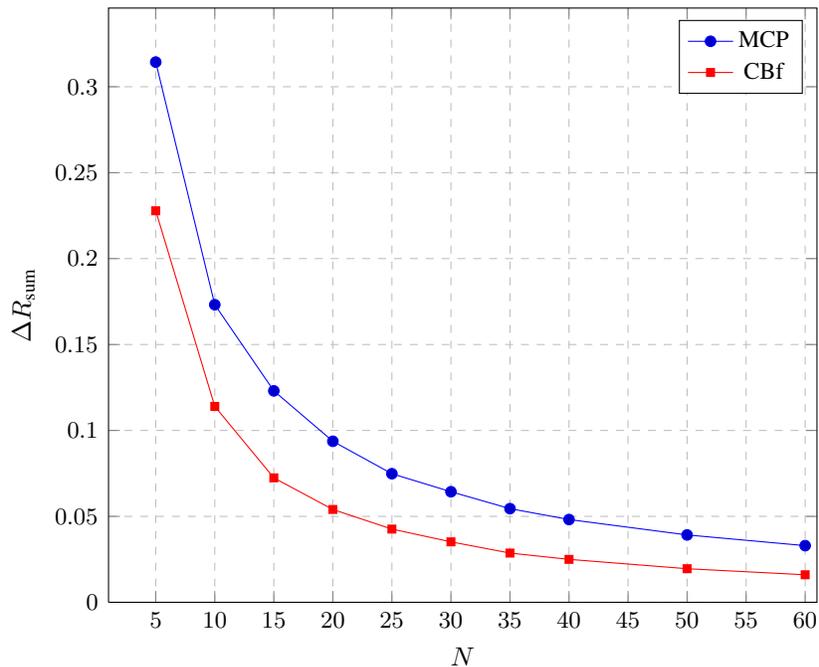
\begin{figure}[t]
\centering
\begin{tikzpicture}
	\pgfplotsset{grid style=dashed}
	\begin{axis}[small, width=11cm, xlabel=$N$, ylabel=$\Delta R_\text{sum}$,
	yticklabel style={/pgf/number format/fixed},
	xmin=1, xmax=61, ymin=0,
	grid=major
	 ]
		\addplot table {av_ThdiffN_MCP_RVQ.dat};
		\addplot[color=red, mark=square*, mark size=1.5] table {av_ThdiffN_CBf_RVQ.dat};
		\legend{\footnotesize MCP, \footnotesize CBf}
	\end{axis}
\end{tikzpicture}
\caption{The total sum-rate difference for different system dimensions with $\beta=0.6$, $\epsilon =0.5$ $\gamma_d = 10$ dB and $\bar{B}_t = 4$.}
\label{fig:conv_sumrate_rvq}
\end{figure}
 
\begin{figure}[t]
\centering
\begin{tikzpicture}
	\pgfplotsset{grid style=dashed}
	\begin{semilogyaxis}[small, width=11cm, xlabel=$\epsilon$, ylabel=$\frac{\mathbb{E}[ R_{\text{sum}}(\bar{B}_{d,\text{FS}}^*) - R_\text{sum}(\bar{B}_d^*)]}{\mathbb{E}[R_{\text{sum}}(\bar{B}_{d,\text{FS}}^*)]}$,
	xmin=0, xmax=1, ymin=0, ymax=0.1,
	minor x tick num=1,
	xmajorgrids, yminorgrids,
	 ]
		\addplot table {av_ThdiffvsEpsil_MCP_RVQ.dat};
		\addplot[color=red, mark=square*, mark size=1.5] table {av_ThdiffvsEpsil_CBf_RVQ.dat};
		\legend{\footnotesize MCP, \footnotesize CBf}
	\end{semilogyaxis}
\end{tikzpicture}
\caption{
The (normalized) average sum-rate difference of the finite-size system by using the $\bar{B}_{d,\text{FS}}^*$ and $\bar{B}_d^*$  with $N=10, \beta=0.6$, $\gamma_d = 10$ dB and $\bar{B}_t=4$.}%
\label{fig:ThDiff_rvq}%
\end{figure}
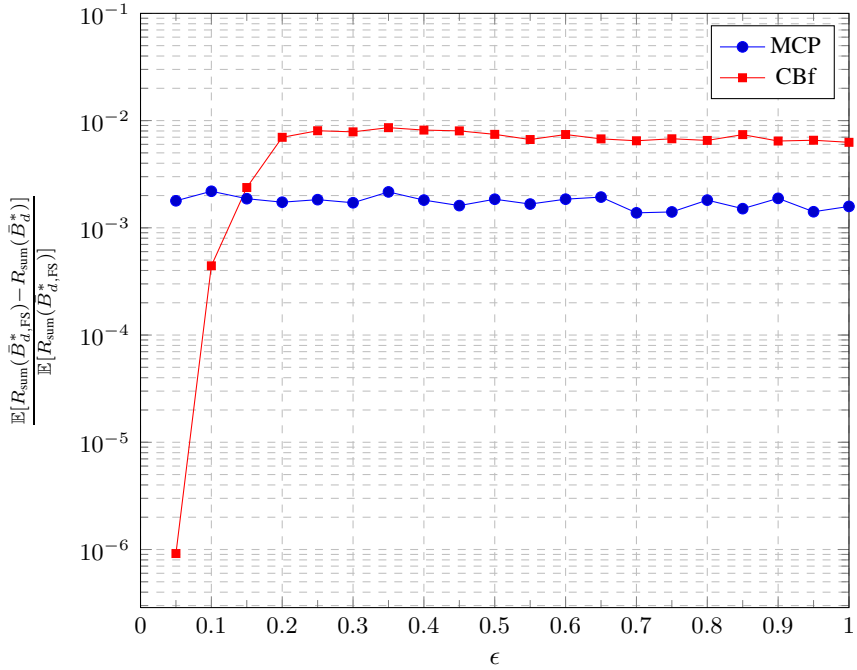

In the following, we present numerical simulations that show the behavior of the limiting SINR and optimal bit allocation for MCP and CBf as $\epsilon$ varies. 
The optimal bit allocation is illustrated in Figure \ref{fig:comparison_DF}. As shown in Section \ref{sec:RVQ_feedback}, the optimal $B_d$ for MCP is decreasing in $\epsilon$ and $B_d^*=B_t$ when $\epsilon=0$. For CBf, $B_d^*=B_t$ when $\epsilon\leq 0.19$, and after that decreases as $\epsilon$ grows. Overall, for given $\epsilon$, 
$B_d^*$ for CBf is larger than for MCP, implying the quality of the direct channel information is more important for CBf. 

In Figure \ref{fig:comparison_DF}, the optimal values for the regularization parameter and bit allocation are used. From that figure, it is obvious that  $\tSINR^\infty_\text{\tiny CBf,Q}$ decreases as $\epsilon$ increases. In the case of MCP, as predicted by the analysis, the  limiting SINR is decreasing until $\epsilon^*_{\text{\tiny M,RVQ}}\approx 0.72$ and is increasing after that point. By comparing the limiting SINR for both cooperation schemes, it is also interesting to see that  for some values of $\epsilon$, i.e., in the interval when CBf has $\bar{B}_c^*=0$, the CBf slightly outperforms MCP.  We should note that within the current scheme, when $\bar{B}_c^*=0$, CBf and MCP are not the same as single-cell processing (SCP): under RVQ, there is still a quantization vector in the codebook that is used to represent the cross channel (although it is uncorrelated with the actual channel vector being quantized).     

Motivated by the above facts, we investigate whether SCP provides some advantages over MCP and CBf for some (low) values of $\epsilon$. 
In SCP, we use $B_{k,j,j}=B_t$ bits ($\forall k,j$) to quantize the direct channel. The cross channels in the precoder are represented by vectors with zero entries. By following the steps in deriving Theorem \ref{th:sinr_mcp_rvq} and \ref{th:sinr_coord_rvq}, we can show that the limiting SINR is given by 
\begin{align*}
\tSINR^\infty_{\text{\tiny SCP,Q}}=\gamma_e g(\beta,\rho_{\text{\tiny S}}) \frac{1+\frac{\rho_{\text{\tiny S}}}{\beta}(1+g(\beta,\rho_{\text{\tiny S}}))^2}{\gamma_e+(1+g(\beta,\rho_{\text{\tiny S}}))^2},
\end{align*}
where $\rho_{\text{\tiny S}}=N^{-1}\alpha$ and $\gamma_e = \frac{1-2^{-\bar{B}_t}}{2^{-\bar{B}_t}+\epsilon+\frac{1}{\gamma_d}}$.
It follows that the optimal $\rho_{\text{\tiny S}}$ maximizing $\tSINR^\infty_{\text{\tiny SCP,Q}}$  is $\rho_{\text{\tiny S}}^*=\frac{\beta}{\gamma_e}$
and the corresponding the limiting SINR is $\tSINR^{*,\infty}_{\text{\tiny SCP,Q}}=g(\beta,\rho_{\text{\tiny S}}^*)$.

From Figure \ref{fig:comparison_DF}, it is obvious that the SCP outperforms MCP and CBf for some values of $\epsilon$, that is, $\epsilon \leq 0.13$.
Surprisingly, the CBf is still beaten by  SCP until $\epsilon\approx 0.82$. This means that the SCP still gives advantages over the CBf even in a quite strong interference regime with this level of feedback.

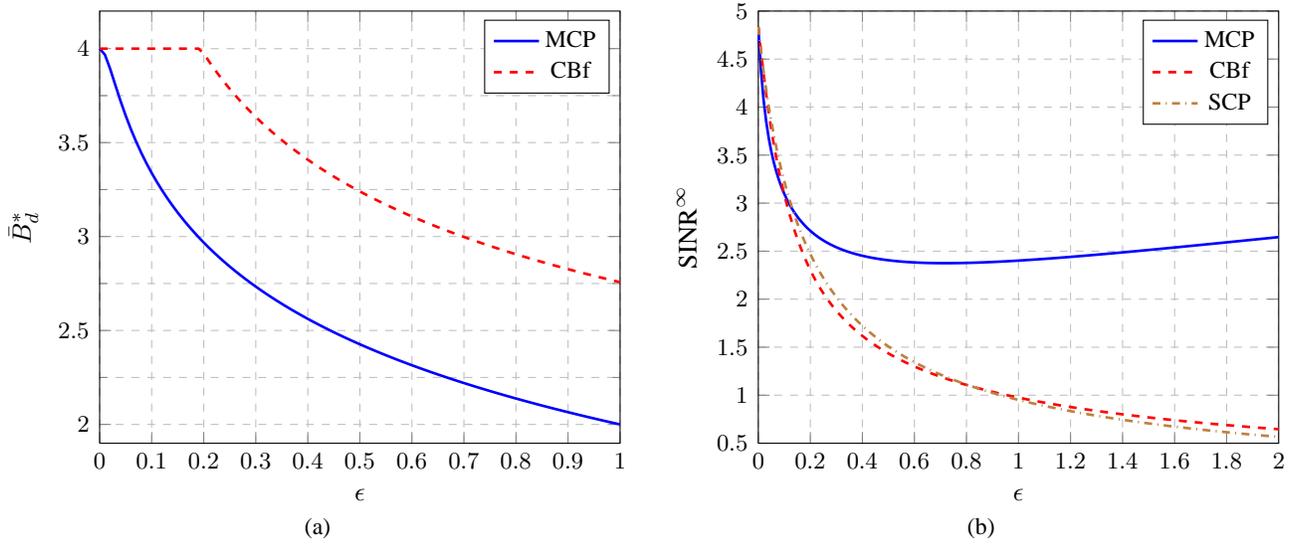
\begin{figure}[t]
\centering
\begin{tabular}{cc}
\begin{tikzpicture}[baseline]
	\pgfplotsset{grid style=dashed}
	\begin{axis}[small, width=8.5cm,  xlabel=$\epsilon$, ylabel=$\bar{B}_d^*$,
	yticklabel style={/pgf/number format/fixed}, 
	xmin=0, xmax=1, ymin=1.9, ymax=4.2,
	minor y tick num=1,
	xmajorgrids,	yminorgrids, ymajorgrids, no markers
	 ]
		\addplot[color=blue, line width=1] table {Bd_vs_epsil_beta06_Xt4Pd10_MCP.dat};
		\addplot[color=red, dashed, line width=1] table {Bd_vs_epsil_beta06_Xt4Pd10_CBf.dat};
		\legend{\footnotesize MCP, \footnotesize CBf}
	\end{axis}
\end{tikzpicture}
&
\begin{tikzpicture}[baseline]
	\pgfplotsset{grid style=dashed}
	\begin{axis}[small, width=8.5cm, xlabel=$\epsilon$, ylabel=$\tSINR^\infty$,
	yticklabel style={/pgf/number format/fixed},
	xmin=0, xmax=2, ymin=0.5, ymax=5, ytick={0.5,1,1.5,...,5},
		grid=major, no markers
	 ]
		\addplot[color=blue, line width=1] table {SINR_vs_epsil_beta06_Xt4Pd10_MCP.dat};
		\addplot[color=red, dashed, line width=1] table {SINR_vs_epsil_beta06_Xt4Pd10_CBf.dat};
		\addplot[color=brown, dashdotted, line width=1] table {SINR_vs_epsil_beta06_Xt4Pd10_SCP.dat};
		\legend{\footnotesize MCP, \footnotesize CBf, \footnotesize SCP}
	\end{axis}
\end{tikzpicture}\\
$\text{\footnotesize (a)}$ & $\text{\footnotesize (b)}$
\end{tabular}
\caption{(a) Optimal bit allocation vs $\epsilon$  (b) Limiting SINR vs. $\epsilon$. Parameters: $\gamma_d$=10 dB, $\bar{B}_t=4$.}  %
\label{fig:comparison_DF}%
\end{figure}

\section{Analog vs. digital feedback}\label{sec:compare_ana_rvq}
In this section we will compare the performance of the analog and quantized feedback for each cooperation scheme. For the quantized feedback, we follow the approach in \cite{Caire_it10, Koba_comm11,JZhang_arxiv11,Santipach_it10} that  translates feedback bits to symbols for a fair comparison with the analog feedback. In this regard, there are two approaches \cite{JZhang_arxiv11}:
\begin{enumerate}
	\item  By assuming that the feedback channel is error free and transmitted at the uplink rate (even though this assumption could be unrealistic in practice), we can write
\be\label{eq:Bit_to_Symbols_adapt}
\bar{B}_t=\frac{B_t}{N}= 2\kappa\log_2\left(1+(1+\epsilon)\gamma_u\right).
\ee 
This approach is introduced in \cite{Caire_it10, Koba_comm11}. \eqref{eq:Bit_to_Symbols_adapt} is obtained by assuming that each feedback bit is received by both base stations in different cells where the path-gains from a user to its own BS and other BS are different i.e, 1 and $\epsilon$ respectively. We can think the feedback transmission from a user to both BSs as a Single-Input Multi-Output (SIMO) system. The BSs linearly combine the feedback signal from the user and the corresponding maximum SNR is  $(1+\epsilon)\gamma_u$ (see \cite{Larsson_book08}). The pre-log factor $2\kappa N$ for $B_t$ in \eqref{eq:Bit_to_Symbols_adapt} presents the channel uses (symbols) for transmitting the feedback bits which are the same as those for the analog feedback. $\kappa$ follows the discussion in Section \ref{ss:AF_model}. Our approach is different from the approach in \cite{JZhang_arxiv11} in which the user $k$ in cell $j$ sends the feedbacks only to its own BS $j$. In that case,
\eqref{eq:Bit_to_Symbols_adapt} becomes  $\bar{B}_t= 2\kappa\log_2\left(1+\gamma_u\right)$.  

\item Following \cite{Santipach_it10}, the second approach translates the feedback bits to symbols based on the modulation scheme used in the feedback transmission.  In the analog feedback, the feedback takes $2\kappa N$ channel uses per user. Let $\eta$ be a conversion factor that links the bits and symbols and it depends on the modulation scheme. As an example, for the binary phase shift keying (BPSK), $\eta=1$.
Thus, we can write (see also \cite{JZhang_arxiv11})
\be\label{eq:Bit_to_Symbols_fixed}
\eta B_t= 2\kappa N.
\ee
We should note that using this approach, for a fixed $\kappa$ there is no link between $\bar{B}_t$ and $\gamma_u$ as we can see in \eqref{eq:Bit_to_Symbols_adapt}. 
\end{enumerate}

  Let us assume that $\kappa=1$. Thus, with the first approach, we have $X_t=2^{-\bar{B}_t}=\frac{1}{\left(1+(1+\epsilon)\gamma_u\right)^2}$. The comparison of the limiting SINR based on the analog and quantized feedback for MCP and CBf can be seen in Figure \ref{fig:comparison_AFandDF}(a). It shows that the quantized feedback beats the analog feedback in both MCP and CBf for $\epsilon$ less than about 1. A similar situation still occurs for CBf even for $\epsilon \in [0,2]$. The opposite happens for MCP when $\epsilon$ is above 1.5.  
 The comparison of the analog and quantized feedback with the second approach, also with $\kappa=1$, is illustrated in Figure \ref{fig:comparison_AFandDF}(b). Similar to the previous, one can see that the quantized feedback outperforms the analog feedback if $\epsilon$ is below a certain threshold. Otherwise, the analog feedback gives better performance. Those observations can be explained by verifying whether the feedback scheme that provides better CSIT will give a better performance. This is easier to check by looking at the MCP scheme because from our discussions in Section \ref{sec:ana_feedback} and \ref{sec:RVQ_feedback}, its performance can be measured by the total CSIT quality, i.e., $\omega_c+\omega_d$ in  the analog feedback and $(1+\epsilon)d^2$ in the digital feedback. Plotting those over $\epsilon$, not shown here, will give the same behaviors for the MCP as we observed in Figure \ref{fig:comparison_AFandDF}. Thus, from our simulations above, the CSIT quality of the quantized feedback is better than that of analog feedback when the cross channel gain is below a certain threshold. The above plots also confirm that more feedback resources will increase the system performance: for a fixed $\gamma_u=0$ dB, $\bar{B}_t$ in the left plot is larger than that in the right plot and hence gives a higher (limiting) SINR for the quantized feedback scheme. 

%

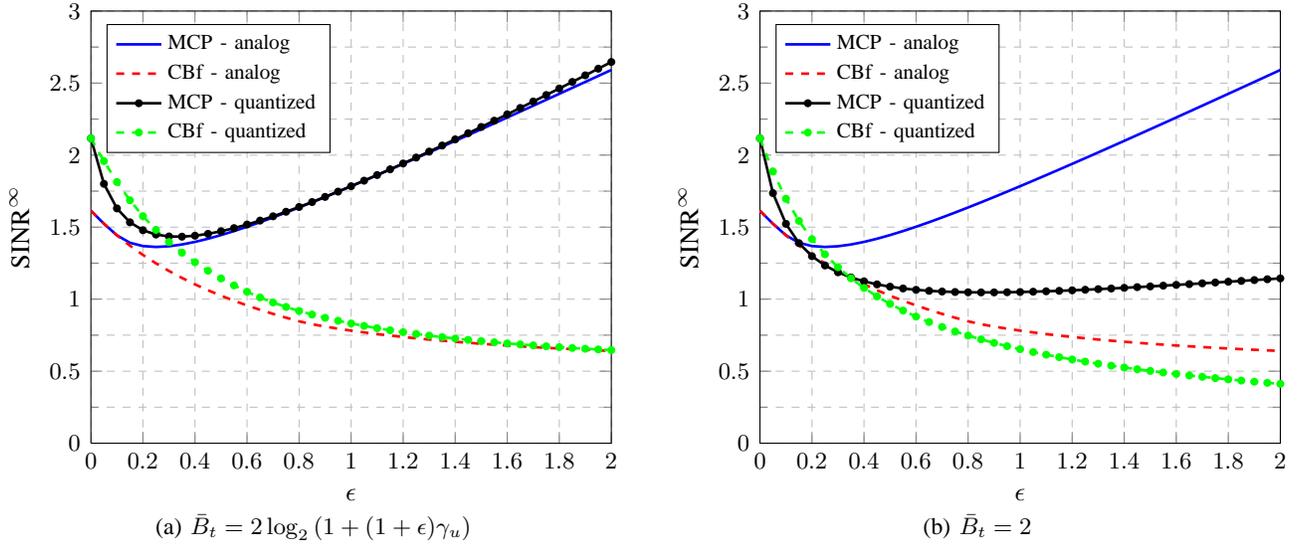
\begin{figure}
\centering
\begin{tabular}{cc}
\begin{tikzpicture}[baseline]
	\pgfplotsset{grid style=dashed, /tikz/every mark/.append style={solid}}
	\begin{axis}[small, width=8.5cm,  xlabel=$\epsilon$, ylabel=$\tSINR^\infty$,
	yticklabel style={/pgf/number format/fixed}, 
	xmin=0, xmax=2, ymin=0, ymax=3,
	minor y tick num=1,
	xmajorgrids,	yminorgrids, ymajorgrids,
	legend pos = north west,
	legend cell align = left
	 ]
		\addplot[color=blue, line width=1] table[x=epsil, y=sinr_ana_mcp] {dataSINRvsEpsil_compareAFvsDF_appr1.dat};
		\addplot[color=red, dashed, line width=1] table[x=epsil, y=sinr_ana_cbf] {dataSINRvsEpsil_compareAFvsDF_appr1.dat};
		\addplot[mark=*, mark size=1, color=black, line width=1] table[x=epsil, y=sinr_dig_mcp] {dataSINRvsEpsil_compareAFvsDF_appr1.dat};
		\addplot[mark=*, mark size=1, color=green, dashed, line width=1] table[x=epsil, y=sinr_dig_cbf] {dataSINRvsEpsil_compareAFvsDF_appr1.dat};
		\legend{\scriptsize MCP - analog, \scriptsize CBf - analog, \scriptsize MCP - quantized, \scriptsize CBf - quantized}
	\end{axis}
\end{tikzpicture}
&
\pgfplotsset{grid style=dashed, /tikz/every mark/.append style={solid}}
\begin{tikzpicture}[baseline]
	\begin{axis}[small, width=8.5cm, xlabel=$\epsilon$, ylabel=$\tSINR^\infty$,
	yticklabel style={/pgf/number format/fixed},
	xmin=0, xmax=2, ymin=0, ymax=3, 
		minor y tick num=1,
	xmajorgrids,	yminorgrids, ymajorgrids,
	legend pos = north west,
	legend cell align = left
	 ]
		\addplot[no markers, color=blue, line width=1] table[x=epsil, y=sinr_ana_mcp] {dataSINRvsEpsil_compareAFvsDF_appr2.dat};
		\addplot[no markers, color=red, dashed, line width=1] table[x=epsil, y=sinr_ana_cbf] {dataSINRvsEpsil_compareAFvsDF_appr2.dat};
		\addplot[mark=*, mark size=1, color=black, line width=1] table[x=epsil, y=sinr_dig_mcp] {dataSINRvsEpsil_compareAFvsDF_appr2.dat};
		\addplot[mark=*, mark size=1, color=green, dashed, line width=1] table[x=epsil, y=sinr_dig_cbf] {dataSINRvsEpsil_compareAFvsDF_appr2.dat};
		\legend{\scriptsize MCP - analog, \scriptsize CBf - analog, \scriptsize MCP - quantized, \scriptsize CBf - quantized}
	\end{axis}
\end{tikzpicture}\\
{\footnotesize (a) $\bar{B}_t= 2\log_2\left(1+(1+\epsilon)\gamma_u\right)$} & {\footnotesize (b) $\bar{B}_t=2$}
\end{tabular}

\caption{The comparison of the limiting SINR of the analog and quantized feedback for different cooperation schemes. Parameters: $\beta=0.6$, $\gamma_d=10$ dB, $\gamma_u=0$ dB}  %
\label{fig:comparison_AFandDF}%
\end{figure}

Figure \ref{fig:comparison_AFandDF_vsBt} depicts the limiting SINR of the analog and quantized feedback for different values of feedback rate. For the analog feedback, the values of feedback rate/bit is converted by using the previous approaches: $\kappa=\frac{\bar{B}_t}{2\log_2(1+(1+\epsilon)\gamma_u)}$ and $\kappa=\bar{B}_t/2$ respectively. 
For MCP, we can see that initially the analog feedback scheme outperforms the quantized feedback in both plots. However, after a certain value (threshold) of $\bar{B}_t$, the opposite happens. A similar observation also holds for the CBf scheme. The explanations for those phenomena follow the discussions for Figure \ref{fig:comparison_AFandDF}. We should note that in generating the figures, the values for $\bar{B}_t$ are already determined. So, the limiting SINRs for the digital feedback are the same in both sub-figures. For the analog feedback, since $\kappa$ with the approach \eqref{eq:Bit_to_Symbols_fixed} is larger (with $\gamma_u=0$ dB) than that with the approach \eqref{eq:Bit_to_Symbols_adapt}, then the training period in the former is  longer and  will result in a better CSIT. Thus,
the limiting SINRs for the analog feedback in Figure \ref{fig:comparison_AFandDF_vsBt}(b) are larger compared to those in \ref{fig:comparison_AFandDF_vsBt}(a). 
       

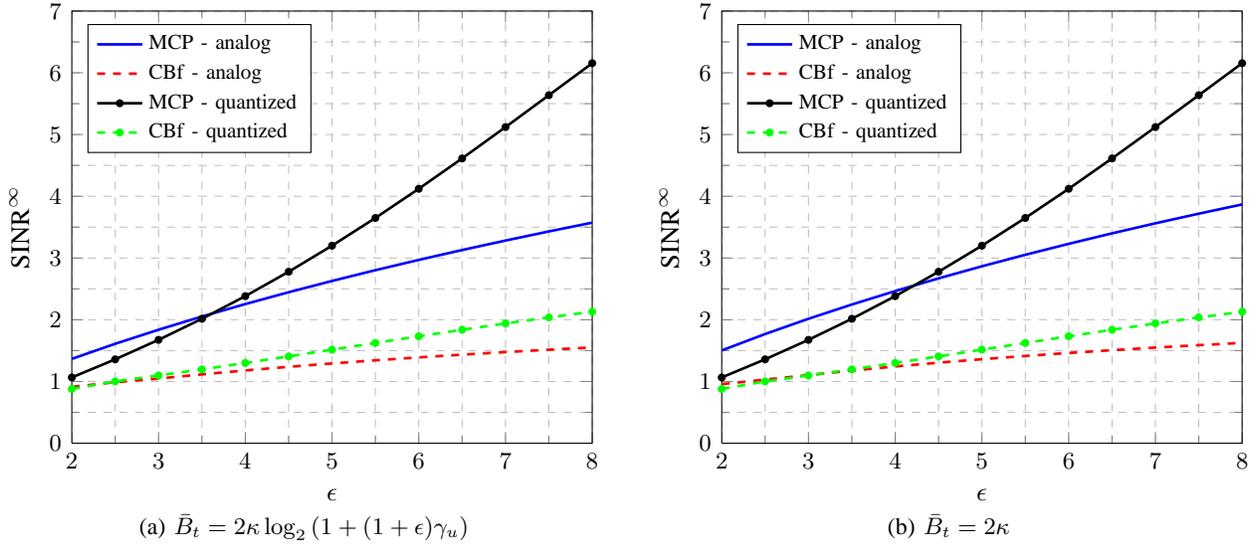
\begin{figure}
\centering
\begin{tabular}{cc}
\begin{tikzpicture}[baseline]
	\pgfplotsset{grid style=dashed, /tikz/every mark/.append style={solid}}
	\begin{axis}[small, width=8.5cm,  xlabel=$\epsilon$, ylabel=$\tSINR^\infty$,
	yticklabel style={/pgf/number format/fixed}, 
	xmin=2, xmax=8, ymin=0, ymax=7,
	minor x tick num=1,
	minor y tick num=1,
	xmajorgrids,	xminorgrids, yminorgrids, ymajorgrids,
	legend pos = north west,
	legend cell align = left
	 ]
		\addplot[color=blue, line width=1] table[x=epsil, y=sinr_ana_mcp] {dataSINRvsBt_adapt.dat};
		\addplot[color=red, dashed, line width=1] table[x=epsil, y=sinr_ana_cbf] {dataSINRvsBt_adapt.dat};
		\addplot[mark=*, mark size=1, color=black, line width=1] table[x=epsil, y=sinr_dig_mcp] {dataSINRvsBt_adapt.dat};
		\addplot[mark=*, mark size=1, color=green, dashed, line width=1] table[x=epsil, y=sinr_dig_cbf] {dataSINRvsBt_adapt.dat};
		\legend{\scriptsize MCP - analog, \scriptsize CBf - analog, \scriptsize MCP - quantized, \scriptsize CBf - quantized}
	\end{axis}
\end{tikzpicture}
&
\pgfplotsset{grid style=dashed, /tikz/every mark/.append style={solid}}
\begin{tikzpicture}[baseline]
	\begin{axis}[small, width=8.5cm, xlabel=$\epsilon$, ylabel=$\tSINR^\infty$,
	yticklabel style={/pgf/number format/fixed},
	xmin=2, xmax=8, ymin=0, ymax=7, 
		minor y tick num=1,
		minor x tick num=1,
	xmajorgrids,	xminorgrids, yminorgrids, ymajorgrids,
	legend pos = north west,
	legend cell align = left
	 ]
		\addplot[no markers, color=blue, line width=1] table[x=epsil, y=sinr_ana_mcp] {dataSINRvsBt_fixed.dat};
		\addplot[no markers, color=red, dashed, line width=1] table[x=epsil, y=sinr_ana_cbf] {dataSINRvsBt_fixed.dat};
		\addplot[mark=*, mark size=1, color=black, line width=1] table[x=epsil, y=sinr_dig_mcp] {dataSINRvsBt_fixed.dat};
		\addplot[mark=*, mark size=1, color=green, dashed, line width=1] table[x=epsil, y=sinr_dig_cbf] {dataSINRvsBt_fixed.dat};
		\legend{\scriptsize MCP - analog, \scriptsize CBf - analog, \scriptsize MCP - quantized, \scriptsize CBf - quantized}
	\end{axis}
\end{tikzpicture}\\
{\footnotesize (a) $\bar{B}_t= 2\kappa\log_2\left(1+(1+\epsilon)\gamma_u\right)$} & {\footnotesize (b) $\bar{B}_t=2\kappa$}
\end{tabular}

\caption{The comparison of the limiting SINR of the analog and quantized feedback for different cooperation schemes vs. the feedback rates. Parameters: $\beta=0.6$, $\epsilon=0.6$,$\gamma_d=10$ dB, $\gamma_u=0$ dB}  %
\label{fig:comparison_AFandDF_vsBt}%
\end{figure}


\section{Conclusion}\label{sec:conclusion}
In this paper, we perform feedback optimization for the analog and quantized feedback schemes in a symmetric two-cell network with different levels of cooperation between base stations. In both cooperation schemes, it is shown that more resources, uplink transmit power in the case of analog feedback or feedback bits in the case of quantized feedback, are allocated to feeding back the interfering channel information as the interfering channel gain increases. Moreover, if the interfering channel gain is below a certain threshold, the conventional network with no cooperation between base stations is preferable. Our analysis also shows that the limiting SINR for MCP, in both analog and quantized feedback, improves in $\epsilon$ if $\epsilon$ is above certain threshold. This also implies that above that threshold the (total) quality of the channel at the base stations is also getting better. Although our analysis is performed in the asymptotic regime, our numerical results hint to their validity in the finite-size system cases.  Future works could consider a more general channel model such as analog feedback through MAC channels. Furthermore, feedback reduction problems in which the users or groups of users have different path-loss gains can be interesting to explore.

\appendices
\section{Some Results in Random Matrix theory}
For the clarity in presentation, in this  section we will list some results in random matrix theory that will be used to derive the large system results in this work. 

\begin{lemma}[{\cite[Lemma 1]{Evans_it00}}]\label{th:Lemma1_jse}
Let $\mathbf{A}$ be a deterministic  $N\times N$ complex matrix with uniformly bounded spectral radius for all $N$. Let $\mathbf{q}=\frac{1}{\sqrt{N}}[q_1,q_2,\cdots,q_N]^T$ where the $q_i$'s are i.i.d with zero mean, unit variance and finite eight moment. Let $\mathbf{r}$ be a similar vector independent of $\mathbf{q}$. Then, we have
\[
\mathbf{qAq}^H - \frac{1}{N}\TR{\mathbf{A}}  \as 0, \text{ and } \mathbf{qAr}^H \as 0.
\]
\end{lemma}

\begin{theorem}[\cite{Hanly_it01}]\label{th:Girko}
Let $\bH$ be a $\lfloor cN \rfloor \times \lfloor dN \rfloor $ random matrix with independent entries $[\bH]_{ij}$ which are zero mean and variance $\mathbb{E}\left[|[\bH]_{ij}|^2\right]=N^{-1}\mathsf{P}_{ij}$, such that $\mathsf{P}_{ij}$ are uniformly bounded from above. For each $N$, let
\[
v_N(x,y) : [0,c]\times [0,d] \to \mathbb{R}
\]
be the \textit{variance profile} function given by
\[
v_N(x,y)=\mathsf{P}_{ij}, \quad \frac{i}{N} \leq x \leq \frac{i+1}{N},  \quad \frac{j}{N} \leq y \leq \frac{j+1}{N}.
\]
Suppose that $v_N(x,y)$ converges uniformly to a limiting bounded function $v(x,y)$. Then, for each $a,b \in [0,c], a<b$, and $z\in \mathbb{C}^+$
\[
\frac{1}{N}\sum_{i=\lfloor aN \rfloor}^{\lfloor bN \rfloor} \left[\left(\bH\bH^H - z \bI\right)^{-1}\right]_{ii} \ip \int_a^b u(x,z)\ dz
\] 
where $u(x,z)$ satisfies 
\[
u(x,z)=\frac{1}{-z+\int_0^{d} \frac{v(x,y) dy}{1+\int_0^c u(w,z)v(w,y)dw}}
\]
for every $x \in [0,c]$. The solution always exists and is unique in the class of functions $u(x,z) \geq 0$, analytic on $z \in \mathbb{C}^+$ and continuous on $x\in [0,c]$.  
Moreover, almost surely, the empirical eigenvalue distribution of $\bH\bH^H$ converges weakly to a limiting distribution whose Stieltjes transform is given by $\int_0^1 u(x,z)\ dx$
\end{theorem}
In the theorem above $x$-axis and $y$-axis refer to the rows and columns of the matrix $\bH$, respectively.

\section{Large system results for the Network MIMO}
First, we will expand the SINR expression in \eqref{eq:sinr_mcp}. Let $\bPhi_{k,j}=\text{diag}\{\phi_{k,j,1}, \phi_{k,j,2}\}$. Based on \eqref{eq:ch_model} we can write $\bh_{k,j}=\hh_{k,j}\bPhi_{k,j}^\half+\ih_{k,j}$.
Consequently, the SINR can be expressed as
\be\label{eq:sinr_mcp_app}
\tSINR_{k,j}=\frac{c^2\left|(\hh_{k,j}\bPhi_{k,j}^\half+\ih_{k,j})\left(\hH^H\hH+\alpha \mathbf{I}\right)^{-1}\hh^H_{k,j}\right|^2}{c^2(\hh_{k,j}\bPhi^\half_{k,j}+\ih_{k,j})\left(\hH^H\hH+\alpha \mathbf{I}\right)^{-1}\hH^H_{k,j}\hH_{k,j}\left(\hH^H\hH+\alpha \mathbf{I}\right)^{-1}(\hh_{k,j}\bPhi^\half_{k,j}+\ih_{k,j})^H + \sigma^2_d}.
\ee
It holds that $\left(\hH^H\hH+\alpha \mathbf{I}\right)^{-1}=\left(\hH^H_{k,j}\hH_{k,j}+ \hh_{k,j}^H\hh_{k,j} + \alpha \mathbf{I}\right)^{-1}$.
By applying the matrix inversion lemma (MIL) to the RHS of the (previous) equation, we obtain
\be\label{eq:inversion_mcp}
\left(\hH^H\hH+\alpha \mathbf{I}\right)^{-1}=\left(\hH^H_{k,j}\hH_{k,j}+\alpha \mathbf{I}\right)^{-1} - \frac{\left(\hH^H_{k,j}\hH_{k,j}+\alpha \mathbf{I}\right)^{-1}\hh_{k,j}^H\hh_{k,j}\left(\hH^H_{k,j}\hH_{k,j}+\alpha \mathbf{I}\right)^{-1}}{1+\hh_{k,j}\left(\hH^H_{k,j}\hH_{k,j}+\alpha \mathbf{I}\right)^{-1}\hh_{k,j}^H}.
\ee
Let $\rho=\frac{\alpha}{N}$ and $\mathbf{O}_{k,j} = \left(\frac{1}{N}\hH^H_{k,j}\hH_{k,j}+\rho \mathbf{I}\right)^{-1}$. Then, we can write \eqref{eq:inversion_mcp} as $\frac{1}{N}\mb{Z}_{k,j}$, where 
\[
\mb{Z}_{k,j}=\mathbf{O}_{k,j}-\frac{\mathbf{O}_{k,j}\left(\frac{1}{N}\hh_{k,j}^H\hh_{k,j}\right)\mathbf{O}_{k,j}}{1+\frac{1}{N}\hh_{k,j}\mb{O}_{k,j}\hh_{k,j}^H}.
\] 
Thus, \eqref{eq:sinr_mcp_app} can be expressed as 
\be\label{eq:sinr_mcp_expand}
\tSINR_{k,j}=\frac{c^2\left|\displaystyle \frac{\breve{A}_{k,j}+F_{k,j}}{1+A_{k,j}}\right|^2}{c^2\left(B_{k,j}+2\Re\left[D_{k,j}\right]+E_{k,j} \right) + \sigma^2_d},
\ee
where
\begin{align*}
\breve{A}_{k,j} &= \frac{1}{N}\hh_{k,j}\bPhi_{k,j}^\half\mb{O}_{k,j}\hh^H_{k,j}\\
A_{k,j} &= =\frac{1}{N}\hh_{k,j}\mb{O}_{k,j}\hh^H_{k,j}\\
F_{k,j} &= =\frac{1}{N}\ih_{k,j}\bPhi_{k,j}^\half\mb{O}_{k,j}\hh^H_{k,j}\\
B_{k,j} &= \frac{1}{N}\hh_{k,j}\bPhi_{k,j}^\half\mathbf{Z}_{k,j}\left(\frac{1}{N}\hH^H_{k,j}\hH_{k,j}\right)\mathbf{Z}_{k,j}\bPhi_{k,j}^\half\hh_{k,j}^H\\
D_{k,j} &= \frac{1}{N}\hh_{k,j}\bPhi_{k,j}^\half\mathbf{Z}_{k,j}\left(\frac{1}{N}\hH^H_{k,j}\hH_{k,j}\right)\mathbf{Z}_{k,j}\ih_{k,j}^H\\
E_{k,j} &= \frac{1}{N}\ih_{k,j}\mathbf{Z}_{k,j}\left(\frac{1}{N}\hH^H_{k,j}\hH_{k,j}\right)\mathbf{Z}_{k,j}\ih_{k,j}^H.
\end{align*}
Note that for the analog feedback, $\bPhi_{k,j}=\bI, \forall k,j$.
In the following subsections, the large system limit of each term in the SINR \eqref{eq:sinr_mcp_expand}, for the analog and quantized feedback cases, will be derived. For brevity in the following presentation, we also denote $\mathbf{Q}_{k,j} = \mathbf{O}_{k,j} \left(\frac{1}{N}\hH^H_k\hH_{k,j}\right)\mathbf{O}_{k,j}$.

\subsection{Proof of Theorem 1: Analog Feedback case}\label{App:SINR_MCP}
In the analog feedback case, as mentioned in Section \ref{sec:ana_feedback}, $\hh_{k,i,j}\sim\mathcal{CN}(0,\omega_{ij}\bI_N)$ and $\ih_{k,i,j}\sim\mathcal{CN}(0,\delta_{ij}\bI_N)$ are independent.  We can rewrite those vectors respectively as 
\[
\hh_{k,j}=\bg_{k,j}\sG_{k,j}^\half  \text{ and } \ih_{k,j}=\bd_{k,j} \sD_{k,j}^\half,
\]
where $\bg_{k,j} \sim \mathcal{CN}(0,\bI_{2N})$ and $\bd_{k,j} \sim \mathcal{CN}(0,\bI_{2N})$ are independent. The diagonal matrices $\sG_{k,j}$ and $\sD_{k,j}$ are given by
$\sG_{k,j}=\text{diag}\{\omega_{j1}\bI_N,\omega_{j2}\bI_N\} \text{ and }  \sD_{k,j}=\text{diag}\{\delta_{j1}\bI_N,\delta_{j2}\bI_N\}$, respectively. 

Since in the analysis below,  we heavily use Theorem \ref{th:Girko}, let us define the asymptotic variance profile for the matrix $\frac{1}{N}\hH^H$ which is a $2N \times 2\beta N$ complex random matrix. Following Theorem \ref{th:Girko}, in our case, we have $x \in [0,2]$, $y \in [0,2\beta]$ and the asymptotic variance profile is given by 
\[
v(x,y)=\begin{cases}
\omega_d & 0 \leq x < 1, 0 \leq y < \beta \\
\omega_c & 1 \leq x < 2, 0 \leq y < \beta \\
\omega_c & 0 \leq x < 1, \beta \leq y < 2\beta \\
\omega_d & 1 \leq x < 2, \beta \leq y < 2\beta. 
\end{cases}
\]  

In what follows, we will derive the large system limit for the terms in \eqref{eq:sinr_mcp_expand}.

\subsubsection{\texorpdfstring{$\breve{A}_{k,j}$}{LSA for breve Akj}}
 It can rewritten as $\frac{1}{N}\bg_{k,j}\sG_{k,j}^\half\bPhi_{k,j}^\half \bO_{k,j} \sG_{k,j}^\half\bg_{k,j}^H.$
Applying Lemma \ref{th:Lemma1_jse} yields 
\[
\breve{A}_{k,j} - \frac{1}{N}\Tr\left(\sG_{k,j}\bPhi_{k,j}^\half\bO_{k,j}\right)\as 0.
\]
The second term in the LHS can be written as
$
\frac{1}{N}\left[ \sum_{i=1}^N \omega_{j1}\sqrt{\phi_{k,j,1}}[\bO_{k,j}]_{ii} +	\sum_{i=N+1}^{2N} \omega_{j2}\sqrt{\phi_{k,j,2}}[\bO_{k,j}]_{ii}\right]. 
$
By applying Theorem \ref{th:Girko}, it converges in probability to
\[
	\omega_{j1}\sqrt{\phi_{k,j,1}}\int_0^1 u(x,-\rho)\ dx + \omega_{j2}\sqrt{\phi_{k,j,2}}\int_1^2 u(x,-\rho)\ dx,
\]
where for $0 \leq x \leq 1$,
\[
		u(x,-\rho)=u_1=\frac{1}{\rho+\frac{\beta\omega_d }{1+u_1\omega_d +u_2\omega_c }+ \frac{\beta\omega_c }{1+u_1\omega_c + u_2\omega_d }}
\]
and for $1 < x \leq 2$,
\[
	u(x,-\rho)=u_2=\frac{1}{\rho+\frac{\beta\omega_c }{1+u_1\omega_d +u_2\omega_c}+ \frac{\beta\omega_d }{1+u_1\omega_c  + u_2\omega_{d}}}.
\]
The solution of the equations above is $u_1=u_2=u$ where $u$ is the positive solution of 
\[
	u=\frac{1}{\rho+\frac{\beta(\omega_d + \omega_c)}{1+ u(\omega_d+\omega_c)}}.
\]

Let $g(\beta,\rho)$,  be the solution of $g(\beta,\rho)=\left(\rho+\frac{\beta}{1+g(\beta,\rho)}\right)^{-1}$.	
Then, we can express $u$ in term of $g(\beta,\rho)$ as
\[
		u=\frac{1}{\omega_d+ \omega_c}g(\beta,\bar{\rho}), \quad\quad \bar{\rho}=\frac{\rho}{\omega_d+ \omega_c}.
\]
We should note that $\omega_{j1}+ \omega_{j2}=\omega_d+ \omega_c,\ \forall j$. Thus,
\[
		\breve{A}_{k,j} - \breve{A}^\infty_{k,j} \ip 0, \text{ with } \breve{A}^\infty_{k,j}=\frac{\sqrt{\phi_{k,j,1}}\omega_{j1}+ \sqrt{\phi_{k,j,2}}\omega_{j2}}{\omega_d+ \omega_c}g(\beta,\bar{\rho}).
\]
Since in the current feedback scheme $\bPhi_{k,j}=\bI_{2N}$  then $\breve{A}^\infty_{k,j}=g(\beta,\bar{\rho})$.	Moreover, $\breve{A}^\infty_{k,j}=\breve{A}^\infty$ is the same for all users in both cells.
\subsubsection{\texorpdfstring{$A_{k,j}$}{LSA for Akj}}
This term is $\breve{A}_{k,j}$ with $\bPhi_{k,i}=\bI_{2N}$. Thus, it follows that $A_{k,j} - g(\beta,\bar{\rho}) \ip 0$.

\subsubsection{\texorpdfstring{$F_{k,j}$}{LSA for Fkj}}\label{ss:Fkj}
 It can be rewritten as $\frac{1}{N}\bd_{k,j}\sD_{k,j}^\half\bO_{k,j}\sG_{k,j}^\half\bg_{k,j}^H.$ Conditioning on $\hH_{k,j}$, it is obvious that $\bd_{k,j}$, $\bg_{k,j}$ and $\sD_{k,j}^\half\bO_{k,j}\sG_{k,j}^\half$ are independent of each other. By Lemma 1, it follows that $F_{k,j}\as 0$.

\subsubsection{\texorpdfstring{$D_{k,j}$}{LSA for Dkj}}
Expanding $D_{k,j}$, we have
\begin{align}
D_{k,j} &= \frac{1}{N}\hh_{k,j}\bPhi^\half\bQ_{k,j}\ih_{k,j}^H + 
\frac{\breve{A}_{k,j}F_{k,j}^*\left(\frac{1}{N}\hh_{k,j}\bQ_{k,j}\hh_{k}^H\right)}{(1+A_{k,j})^2} -\frac{F_{k,j}^*\left(\frac{1}{N}\hh_{k,j}\bPhi^\half\bQ_{k,j}\hh_{k,j}^H\right)}{1+A_{k,j}}
         -\frac{\breve{A}_{k,j}\left(\frac{1}{N}\hh_{k,j}\bQ_{k,j}\ih_{k,j}^H\right)}{1+A_{k,j}} \notag\\
					&=D_{k,j}^{(1)} + D_{k,j}^{(2)} - D_{k,j}^{(3)} - D_{k,j}^{(4)}. \label{eq:Dkj_MCP_AF}
\end{align}
Following the arguments in $\textit{4)}$, it can be checked that  $D_{k,j}^{(1)} \as 0$.
Similarly, $D_{k,j}^{(4)}\as 0$. Moreover,  since $ F_{k,j} \as 0$ then it follows that
$D_{k,j}^{(2)}\as 0$ and $D_{k,j}^{(3)}\as 0$. Combining the results yields $D_{k,j} \as 0$.

\subsubsection{\texorpdfstring{$B_{k,j}$}{LSA for Bkj}}
It can be rewritten as 
\begin{align}
B_{k,j}&=\frac{1}{N}\hh_{k,j}\bPhi_{k,j}^\half\bQ_{k,j}\bPhi_{k,j}^\half\hh_{k,j}^H + \frac{|\breve{A}_{k,j}|^2\left(\frac{1}{N}\hh_{k,j}\bQ_{k,j}\hh_{k,j}^H\right)}{(1+A_{k,j})^2}-\frac{2\Re\left[\breve{A}_{k,j}^*\left(\frac{1}{N}\hh_{k,j}\bPhi_{k,j}^\half\bQ_{k,j}\hh_{k,j}^H\right)\right]}{1+A_{k,j}} \notag\\
&=B_{k,j}^{(1)}+\frac{|\breve{A}_{k,j}|^2B_{k,j}^{(2)}}{(1+A_{k,j})^2}-\frac{2\Re\left[\breve{A}_{k,j}^*B_{k,j}^{(3)}\right]}{1+A_{k,j}}. \label{eq:Bkj_MCP_AF}
\end{align}
From Lemma \ref{th:Lemma1_jse}, we can show $B_{k,j}^{(1)} -\frac{1}{N}\TR{\sG_{k,j}\bPhi_{k,j}\bQ_{k,j}} \as 0$.
%
It can be shown that the following equality holds 
\begin{align}
(\bX\bX^H+\zeta \bI)^{-1}\bX\bX^H(\bX\bX^H+\zeta \bI)^{-1}
=(\bX\bX^H+\zeta \bI)^{-1}+\zeta\frac{\partial}{\partial \zeta}(\bX\bX^H+\zeta \bI)^{-1}. \label{eq:Qkj_expand}
\end{align}
Thus, we have $\bQ_{k,j}=\bO_{k,j}+\rho\frac{\partial }{\partial \rho}\bO_{k,j}$. By applying Theorem \ref{th:Girko}, we obtain
\[
\frac{1}{N}\sum_{i=\lfloor aN \rfloor}^{\lfloor bN \rfloor} \left[\bQ_{k,j}\right]_{ii} \ip \int_a^b u(x,-\rho) \ dx + 	\rho\frac{\partial}{\partial \rho}\int_a^b u(x,-\rho) \ dx.
\]
Consequently, we can show $B_{k,j}^{(1)} - B_{k,j}^{(1),\infty} \ip 0$, where
\begin{align}
B_{k,j}^{(1),\infty} 
		&=\frac{\omega_{j1}\phi_{k,j,1}+\omega_{j2}\phi_{k,j,2}}{\omega_d+\omega_c}\left[g(\beta,\bar{\rho}) + \bar{\rho}\frac{\partial}{\partial \bar{\rho}}g(\beta,\bar{\rho})\right]\ .	
\end{align}
Similarly, we can also show that $B_{k,j}^{(2)} - B_{k,j}^{(2),\infty} \ip 0$ and $B_{k,j}^{(3)} - B_{k,j}^{(3),\infty} \ip 0$, where
\begin{align*}
B_{k,j}^{(2),\infty} 
		&=g(\beta,\bar{\rho}) + \bar{\rho}\frac{\partial}{\partial \bar{\rho}}g(\beta,\bar{\rho}),	\\
B_{k,j}^{(3),\infty} 
		&=\frac{\omega_{j1}\sqrt{\phi_{k,j,1}}+\omega_{j2}\sqrt{\phi_{k,j,2}}}{\omega_d+\omega_c}\left[g(\beta,\bar{\rho}) + \bar{\rho}\frac{\partial}{\partial \bar{\rho}}g(\beta,\bar{\rho})\right]\ .	
\end{align*}
 For the analog feedback case, i.e., $\bPhi_{k,j}=\bI_{2N}$,  it follows that $B_{k,j}^{(1),\infty}=B_{k,j}^{(2),\infty}=B_{k,j}^{(3),\infty}$.
Therefore, $B_{k,j}$ converges in probability to 
 \[
\frac{1}{(1+g(\beta,\bar{\rho}))^2}\left[g(\beta,\bar{\rho}) + \bar{\rho}\frac{\partial}{\partial \bar{\rho}}g(\beta,\bar{\rho})\right].
\]

\subsubsection{\texorpdfstring{$E_{k,j}$}{LSA for Ekj}}
Expanding this term gives
\begin{align}
E_{k,j}&=\frac{1}{N}\ih_{k,j}\bQ_{k,j}\ih_{k,j}^H - 2\Re\left[ \frac{F_{k,j}\left(\frac{1}{N}\hh_{k,j}\bQ_{k,j}\ih_{k,j}^H\right)}{1+A_{k,j}} \right] +	\frac{\left|F_{k,j}\right|^2\frac{1}{N}\hh_{k,j}\bQ_{k,j}\hh_{k,j}^H}{(1+A_{k,j})^2} \notag \\
	&=E_{k,j}^{(1)}-E_{k,j}^{(2)}+E_{k,j}^{(3)}. \label{eq:Ekj_MCP_AF}
\end{align}
By using the previous results, we can show that $E_{k,j}^{(2)} \as 0$ and $E_{k,j}^{(3)}\ip 0$. From Lemma \ref{th:Lemma1_jse}, it follows that $E_{k,j}^{(1)} - \frac{1}{N}\TR{\sD_{k,j}\bQ_{k,j}} \as 0$.
Following the steps in obtaining $B_{k,j}^{(1),\infty}$, it is easy to show that $E_{k,j}^{(1)}- E_{k,j}^{(1),\infty} \ip 0 $, where
\[
E_{k,j}^{(1),\infty} = \frac{\delta_{k,j,1}+\delta_{k,j,2}}{\omega_d+\omega_c}\left[g(\beta,\bar{\rho}) + \bar{\rho}\frac{\partial}{\partial \bar{\rho}}g(\beta,\bar{\rho})\right].
\]
Thus, we have  $E_{k,j}- E_{k,j}^{(1),\infty} \ip 0 $.

\subsubsection{\texorpdfstring{$c^2$}{LSA for c2}}
The denominator of $c^2$ can be written as follows
\begin{align*}
\frac{1}{N}\TR{\left(\frac{1}{N}\hH^H\hH+\rho \mathbf{I}\right)^{-2}\frac{1}{N}\hH^H\hH}=\int \frac{\lambda}{(\lambda+\rho)^2}\ dF_{\hH^H\hH}(\lambda),
\end{align*}
where $F_{\hH^H\hH}$ is the empirical eigenvalue distribution of $\hH^H\hH$. From Theorem \ref{th:Girko}, $F_{\hH^H\hH}$
converges almost surely to a limiting distribution $G^*$ whose Stieltjes transform $m(z)=\int_0^\infty \frac{1}{\lambda-z}\ dG^*(\lambda)=\int_0^1 u(x,z)\ dx$. 
Therefore, 
\begin{align}
\int \frac{\lambda}{(\lambda+\rho)^2}\ dF_{\hH^H\hH}(\lambda) &= \int \frac{1}{\lambda+\rho} + \rho \frac{\partial}{\partial \rho} \frac{1}{\lambda+\rho}\ dF_{\hH^H\hH}(\lambda) \notag\\
&\as m(-\rho)+\rho \frac{\partial}{\partial \rho}m(-\rho) \notag\\
&=\int_0^1 u(x,-\rho) +\rho \frac{\partial}{\partial \rho} u(x,-\rho)\ dx. \label{eq:lsa_c2_mcp_analog} 
\end{align}
Previously, we have shown that $\int_0^1 u(x,-\rho)=2(\omega_d+\omega_c)^{-1}g(\beta,\bar{\rho})$ with $\bar{\rho}=2\rho(\omega_d+\omega_c)^{-1}$. Hence, \eqref{eq:lsa_c2_mcp_analog} is equal to $2(\omega_d+\omega_c)^{-1}\left( g(\beta,\bar{\rho}) + \bar{\rho}\frac{\partial}{\partial \bar{\rho}}g(\beta,\bar{\rho})\right)$ and the large system result for $c^2$ is given by
\[
c^2 - \frac{\half (\omega_d+\omega_c)P_t}{g(\beta,\bar{\rho}) + \bar{\rho}\frac{\partial}{\partial \bar{\rho}}g(\beta,\bar{\rho})} \as 0.
\]

The large system results in $\textit{1) - 3)}$ and $\textit{7)}$ show that the signal strength, i.e, the numerator of \eqref{eq:sinr_mcp_expand}, converges to 
\be\label{eq:lim_ss_MCP_AF}
\frac{(\omega_d+\omega_c)Pg^2(\beta,\bar{\rho})}{(1+g(\beta,\bar{\rho}))^2\left(g(\beta,\bar{\rho})+\bar{\rho}\frac{\partial}{\partial \bar{\rho}}g(\beta,\bar{\rho})\right)},
\ee
where we already substitute $P=\half P_t$. Similarly, it follows that the interference (energy) converges to
\be\label{eq:lim_int_MCP_AF}
\frac{P}{\left(1+g(\beta,\bar{\rho})\right)^2}\left(\omega_d+\omega_c - (\delta_d+\delta_c)\left(1+g(\beta,\bar{\rho})\right)^2\right).
\ee
To simplify the expression for the limiting $\tSINR_{k,j}$, we need the following result (see \cite{RusdhaPrep})
\[
g(\beta,\bar{\rho}) + \bar{\rho}\frac{\partial}{\partial \bar{\rho}}g(\beta,\bar{\rho})=\frac{\beta g(\beta,\bar{\rho})}{\beta+\bar{\rho}(1+g(\beta,\bar{\rho}))^2}.
\]
By combining the large system results  and also denoting $\rho_\text{\tiny M,AF}=\bar{\rho}$,
we can express the limiting SINR as in \eqref{eq:lim_sinr_mcp}. This completes the proof.

\subsection{Proof of Theorem \ref{th:sinr_mcp_rvq}: Quantized feedback (via RVQ) case}\label{App:limSINR_mcp_rvq}
Since the $\phi_{k,j,i}$ is a random variable, while it is not in the analog feedback,  the derivation for the limiting SINR is quite different from the previous subsection. For a given  $2N \times 2N$ matrix $\bX$, we can partition it as follows 
\[
\bX_{k,j}=\ba{cc}
\bX_{k,j}^{11} & \bX_{k,j}^{12} \\
\bX_{k,j}^{21} & \bX_{k,j}^{22}
\ea,
\]
where  $\bX_{k,j}^{11}=[\bX_{k,j}]_{lm}, l=1,\cdots,N, m=1,\cdots,N$, $\bX_{k,j}^{12}=[\bX_{k,j}]_{lm}, l=1,\cdots,N, m=N+1,\cdots,2N$, $\bX_{k,j}^{21}=[\bX_{k,j}]_{lm}, l=N+1,\cdots,2N, m=1,\cdots,N$, and $\bX_{k,j}^{22}=[\bX_{k,j}]_{lm}, l=N+1,\cdots,2N, m=N+1,\cdots,2N$. 

In the following, we derive the large system limit for each term in the SINR.
\subsubsection{\texorpdfstring{$\breve{A}_{k,j}$}{LSA for brAkj}} 
We can write $\breve{A}_{k,j}$ as
\be\label{eq:brAkj}
\breve{A}_{k,j}=\frac{1}{N}\left(\phi_{k,j,1}^\half\hh_{k,j,1}\bO_{k,j}^{11}\hh_{k,j,1}^H+ \phi_{k,j,1}^\half\hh_{k,j,1}\bO_{k,j}^{12}\hh_{k,j,2}^H+\phi_{k,j,2}^\half\hh_{k,j,2}\bO_{k,j}^{21}\hh_{k,j,1}^H + \phi_{k,j,2}^\half\hh_{k,j,2}\bO_{k,j}^{22}\hh_{k,j,2}^H\right).
\ee

Since $\hh_{k,j,1}$ and $\hh_{k,j,2}$ are independent then the second and third terms converge almost surely to 0. For the first term
\[
\frac{1}{N}\hh_{k,j,1}\bO_{k,j}^{11}\hh_{k,j,1}^H - \frac{\omega_{j1}}{N}\TR{\bO_{k,j}^{11}} \as 0
\]
or equivalently,
\[
\frac{1}{N}\hh_{k,j,1}\bO_{k,j}^{11}\hh_{k,j,1}^H - \frac{\omega_{j1}}{N}\sum_{i=1}^{N} [\bO_{k,j}]_{ii} \as 0.
\]
By using the Girko's result (Theorem \ref{th:Girko}), we have
\[
\frac{1}{N}\sum_{i=1}^{N} [\bO_{k,j}]_{ii} \ip \frac{1}{1+\epsilon}g(\beta,\bar{\rho}),
\]
where $\bar{\rho}=\frac{\frac{1}{N}\alpha}{1+\epsilon}$. By using the same techniques as in \cite{Santipach_it09}, we can show that
\[
\sqrt{\phi_{k,j,i}} = \sqrt{1-\tau^2_{k,j,i}} \ms \begin{cases} \sqrt{1-2^{-\bar{B}_d}} & j=i \\  \sqrt{1-2^{-\bar{B}_c}} & \text{otherwise.}\end{cases}
\]
We should also note that the convergence in mean square sense implies the convergence in probability. By doing the same steps as for the last term of \eqref{eq:brAkj}, we have
\[
\breve{A}_{k,j} - \frac{\sqrt{1-2^{-\bar{B}_d}} + \epsilon\sqrt{1-2^{-\bar{B}_c}}}{1+\epsilon} g(\beta,\bar{\rho}) \ip 0.
\]

\subsubsection{\texorpdfstring{$A_{k,j}$}{LSA for Akj}} This term is $\breve{A}_{k,j}$ with $\bPhi_{k,j}=\bI_{2N}$. Hence, $A_{k,j} - g(\beta,\bar{\rho}) \ip 0$.

\subsubsection{\texorpdfstring{$F_{k,j}$}{LSA for Fkj}} We can expand the term as follows
\be\label{eq:Fkj}
F_{k,j}=\frac{1}{N}\left(\ih_{k,j,1}\bO_{k,j}^{11}\hh_{k,j,1}^H+ \ih_{k,j,1}\bO_{k,j}^{12}\hh_{k,j,2}^H+\ih_{k,j,2}\bO_{k,j}^{21}\hh_{k,j,1}^H + \ih_{k,j,2}\bO_{k,j}^{22}\hh_{k,j,2}^H.\right)
\ee
By using the same steps as in deriving \eqref{eq:LSA_ih_hh}, it follows that each term on the RHS of the equation above converges in probability to 0. Hence,
$F_{k,j} \ip 0$.

\subsubsection{\texorpdfstring{$B_{k,j}$}{LSA for Bkj}} 
By following the representation \eqref{eq:Bkj_MCP_AF} for $B_{kj}$ , we can express the first term on the RHS as
\[
B_{k,j}^{(1)}=\frac{1}{N}\left(\phi_{k,j,1}\hh_{k,j,1}\bQ_{k,j}^{11}\hh_{k,j,1}^H+ \phi_{k,j,1}\hh_{k,j,1}\bQ_{k,j}^{12}\hh_{k,j,2}^H+\phi_{k,j,2}\hh_{k,j,2}\bQ_{k,j}^{21}\hh_{k,j,1}^H + \phi_{k,j,2}\hh_{k,j,2}\bQ_{k,j}^{22}\hh_{k,j,2}^H\right).
\]
The second term and the third term of the equation above converge (almost surely) to 0. For the first term,
\[
\frac{1}{N}\hh_{k,j,1}\bQ_{k,j}^{11}\hh_{k,j,1}^H - \frac{\omega_{j1}}{N}\TR{\bQ_{k,j}^{11}} \as 0. 
\]
From \eqref{eq:Qkj_expand}, we have $\bQ_{k,j}=\bO_{k,j}+\rho\frac{\partial }{\partial \rho}\bO_{k,j}$. Hence, we can show 
\[
\frac{1}{N}\TR{\bQ_{k,j}^{11}}=\frac{1}{N}\sum_{i=1}^{ N } \left[\bQ_{k,j}^{11}\right]_{ii} \ip \int_0^1 u(x,-\rho) \ dx + 	\rho\frac{\partial }{\partial \rho}\int_0^1 u(x,-\rho) \ dx=\frac{1}{1+\epsilon}\left[g(\beta,\bar{\rho}) + \bar{\rho}\frac{\partial}{\partial \bar{\rho}}g(\beta,\bar{\rho})\right].
\]
By doing the same steps for the last term of $B_{k,j}^{(1)}$, it follows that
\[
 B_{k,j}^{(1)} - \frac{1-2^{-\bar{B}_d}+\epsilon(1-2^{-\bar{B}_c})}{1+\epsilon}\left[g(\beta,\bar{\rho}) + \bar{\rho}\frac{\partial}{\partial \bar{\rho}}g(\beta,\bar{\rho})\right] 	\ip 0.
\]
Similarly, we can also show that
\begin{align*}
& B_{k,j}^{(2)} 
		- g(\beta,\bar{\rho}) + \bar{\rho}\frac{\partial}{\partial \bar{\rho}}g(\beta,\bar{\rho}) \ip 0,	\\
& B_{k,j}^{(3)} 
		- \frac{\sqrt{1-2^{-\bar{B}_d}}+\epsilon\sqrt{1-2^{-\bar{B}_c}}}{1+\epsilon}\left[g(\beta,\bar{\rho}) + \bar{\rho}\frac{\partial}{\partial \bar{\rho}}g(\beta,\bar{\rho})\right]	 \ip 0.
\end{align*}
Combining the results together, we obtain
\[
 B_{k,j} -  \left(\frac{1-2^{-\bar{B}_d}+\epsilon(1-2^{-\bar{B}_c})}{1+\epsilon} - \frac{d^2(2+g(\beta,\bar{\rho}))g(\beta,\bar{\rho})}{(1+g(\beta,\bar{\rho}))^2}\right)\left[g(\beta,\bar{\rho}) + \bar{\rho}\frac{\partial}{\partial \bar{\rho}}g(\beta,\bar{\rho})\right]	 \ip 0,
\]
where
\[
d=\frac{\sqrt{1-2^{-\bar{B}_d}}+\epsilon\sqrt{1-2^{-\bar{B}_c}}}{1+\epsilon}.
\]

\subsubsection{\texorpdfstring{$D_{k,j}$}{LSA for Dkj}}
We can expand $D_{k,j}$ as expressed in \eqref{eq:Dkj_MCP_AF}.
Using the previous results, it follows immediately that $D_{k,j}^{(2)}$ and $D_{k,j}^{(3)}$ converge to 0. We can expand the first term  as follows
\[
D_{k,j}^{(1)} =\frac{1}{N}\left(\phi_{k,j,1}^\half\hh_{k,j,1}\mathbf{Q}_{k,j}^{11}\ih_{k,j,1}^H+ \phi_{k,j,1}^\half\hh_{k,j,1}\mathbf{Q}_{k,j}^{12}\ih_{k,j,2}^H+\phi_{k,j,2}^\half\hh_{k,j,2}\mathbf{Q}_{k,j}^{21}\ih_{k,j,1}^H + \phi_{k,j,2}^\half\hh_{k,j,2}\mathbf{Q}_{k,j}^{22}\ih_{k,j,2}^H\right).
\]
Again, by following the same steps as in deriving \eqref{eq:LSA_ih_hh}, each term converges in probability to 0 and hence $D_{k,j}^{(1)} \ip 0$. Similarly, $D_{k,j}^{(4)} \ip 0$. 
Therefore,  $D_{k,j}\ip 0$.

\subsubsection{\texorpdfstring{$E_{k,j}$}{LSA for Ekj}} The expansion of $E_{kj}$ follows \eqref{eq:Ekj_MCP_AF}. 
By applying the previous results, $E_{k,j}^{(2)}$ and $E_{k,j}^{(3)}$ converge in probability to 0. $E_{k,j}^{(1)}$ can rewritten as
\[
E_{k,j}^{(1)}=\frac{1}{N}\left(\ih_{k,j,1}\mathbf{Q}_{k,j}^{11}\ih_{k,j,1}^H+ \ih_{k,j,1}\mathbf{Q}_{k,j}^{12}\ih_{k,j,2}^H+\ih_{k,j,2}\mathbf{Q}_{k,j}^{21}\ih_{k,j,1}^H + \ih_{k,j,2}\mathbf{Q}_{k,j}^{22}\ih_{k,j,2}^H\right).
\]
Since $\ih_{k,j,1}$ and $\ih_{k,j,2}$ are independent then the second and third term converge to 0. The first term in the equation above can be written as
\begin{align*}
\frac{1}{N}\ih_{k,j,1}\mathbf{Q}_{k,j}^{11}\ih_{k,j,1}^H&=\frac{\tau^2_{k,j,1}\|\bh_{k,j,1}\|^2}{N\|\bv_{k,j,1}\bPi_{\hh_{k,j,1}}^\bot\|^2}\left(\bv_{k,j,1} - \frac{(\bv_{k,j,1}\hh^H_{k,j,1})\hh_{k,j,1}}{\|\hh_{k,j,1}\|^2}\right)\bQ_{k,j}^{11}\left(\bv_{k,j,1} - \frac{(\bv_{k,j,1}\hh^H_{k,j,1})\hh_{k,j,1}}{\|\hh_{k,j,1}\|^2}\right)^H \\
&=\frac{\tau^2_{k,j,1}\|\bh_{k,j,1}\|^2}{\|\bv_{k,j,1}\bPi_{\hh_{k,j,1}}^\bot\|^2}\left(\frac{1}{N}\bv_{k,j,1}\bQ_{k,j}^{11}\bv_{k,j,1}^H + \frac{|\frac{1}{N}\bv_{k,j,1}\hh^H_{k,j,1}|^2\frac{1}{N}\hh_{k,j,1}\bQ_{k,j}^{11}\hh_{k,j,1}^H }{\frac{1}{N^2}\|\hh_{k,j,1}\|^4} \right.\\
&\qquad\qquad\left. -2\Re\left[\frac{(\frac{1}{N}\bv_{k,j,1}\hh^H_{k,j,1})^*\frac{1}{N}\bv_{k,j,1}\bQ_{k,j}^{11}\hh_{k,j,1}^H }{\frac{1}{N}\|\hh_{k,j,1}\|^2} \right]\right).
\end{align*}
Since $\bv_{k,j,1}$ and $\hh_{k,j,1}$ are independent,  the second and third term in the bracket converge to 0. It can be shown that 
\[
\frac{1}{N}\bv_{k,j,1}\bQ_{k,j}^{11}\bv_{k,j,1}^H - \frac{1}{1+\epsilon}\left[g(\beta,\bar{\rho}) + \bar{\rho}\frac{\partial}{\partial \bar{\rho}}g(\beta,\bar{\rho})\right] \ip 0.	
\]
The large system limit for the last term of $E_{k,j}$ can be done in the same way. Thus,
\[
E_{k,j} - \frac{2^{-\bar{B}_d}+\epsilon 2^{-\bar{B}_c}}{1+\epsilon}\left[g(\beta,\bar{\rho}) + \bar{\rho}\frac{\partial}{\partial \bar{\rho}}g(\beta,\bar{\rho})\right] \ip 0.	
\]

\subsubsection{\texorpdfstring{$c^2$}{LSA for c2}} We can show that
\[
c^2 - \frac{\half P_t(1+\epsilon)}{g(\beta,\bar{\rho}) + \bar{\rho}\frac{\partial}{\partial \bar{\rho}}g(\beta,\bar{\rho})} \as 0.
\]
Since $\as$ implies $\ip$ then $c^2$ also converges weakly to the same quantity as above.

Combining the results, it follows that the signal strength and the interference converge to
\be\label{eq:lim_ss_mcp_rvq}
\frac{P(1+\epsilon)d^2g^2(\beta,\bar{\rho})}{(1+g(\beta,\bar{\rho}))^2\left(g(\beta,\bar{\rho})+\bar{\rho}\frac{\partial}{\partial \bar{\rho}}g(\beta,\bar{\rho})\right)}
\ee
and
\be\label{eq:lim_int_mcp_rvq}
P(1+\epsilon)\left(1-\frac{d^2g(\beta,\bar{\rho})(2+g(\beta,\bar{\rho}))}{(1+g(\beta,\bar{\rho}))^2}\right),
\ee
respectively. Therefore, \eqref{eq:limSINR_MCP_RVQ}  follows immediately with $\rho_\text{\tiny M,Q}=\bar{\rho}$.

\section{Large System Results for the Coordinated Beamforming}\label{App:SINR_Coord}
For brevity in the proofs, we define the following (see also \cite{Zakhour_it12}) 
\begin{align*}
\bA_{j}  &= \left(\rho + \frac{1}{N}\sum_{m=1}^2\sum_{l=1}^K \hh_{l,m,j}^H\hh_{l,m,j}\right)^{-1} \\
\bA_{kj} &= \left(\rho + \frac{1}{N}\sum_{(l,m)\neq(k,j)} \hh_{l,m,j}^H\hh_{l,m,j}\right)^{-1} \\
\bA_{kj,k'j',j} &= \left(\rho + \frac{1}{N}\sum_{(l,m)\neq(k,j),(k',j')} \hh_{l,m,j}^H\hh_{l,m,j}\right)^{-1}. 			
\end{align*}
From the definitions above, we can write the numerator of the $\tSINR_{k,j}$ \eqref{eq:sinr_coord} excluding $c_j^2$, as follows
\begin{align*}
|\bh_{k,j,j}\bw_{kj}|^2 
												&= \left|\frac{\sqrt{\phi_{k,j,j}}}{N}\hh_{k,j,j}\bA_{kj}\hh_{k,j,j}^H\right|^2 + \left|\frac{1}{N}\ih_{k,j,j}\bA_{kj}\hh_{k,j,j}^H\right|^2 \\
												&\quad+ 2\Re\left[\frac{1}{N^2}(\ih_{k,j,j}\bA_{kj}\hh_{k,j,j}^H)(\hh_{k,j,j}\bA_{kj}\hh_{k,j,j}^H)\right] = \phi_{k,j,j}\left|S_{kj}^{(1)}\right|^2 + |S_{kj}^{(2)}|^2 +S_{kj}^{(3)}.														
\end{align*}

In denominator, let us consider the term $|\bh_{k,j,j'}\bw_{k'j'}|^2$ which can be expanded as follows
\begin{align}
|\bh_{k,j,j'}\bw_{k'j'}|^2
													&=\frac{1}{N^2}\phi_{k,j,j'}\hh_{k,j,j'}\bA_{k'j'}\hh_{k',j',j'}^H\hh_{k',j',j'}A_{k'j'}\hh_{k,j,j'}^H	+ \frac{1}{N^2}\ih_{k,j,j'}\bA_{k'j'}\hh_{k',j',j'}^H\hh_{k',j',j'}A_{k'j'}\ih_{k,j,j'}^H \\
													&\qquad -2\Re\left[\frac{\sqrt{\phi_{k,j,j'}}}{N^2}\ih_{k,j,j'}\bA_{k'j'}\hh_{k',j',j'}^H\hh_{k',j',j'}\bA_{k'j'}\hh_{k,j,j'}^H\right]\\
&=\frac{1}{N}I_{kj,k'j'}^{(1)} + \frac{1}{N}I_{kj,k'j'}^{(2)} - \frac{1}{N}I_{kj,k'j'}^{(3)}=\frac{1}{N}\mathcal{I}. \label{eq:int_cbf}													
\end{align}

Now, we are going to derive the large system limit for $\frac{1}{N}\TR{\bA_j}$ since it will be used frequently in this section. Let $\widehat{\mathsf{H}}_j=[\hh_{1,1,j}\ \cdots\ \hh_{K,1,j} \ \hh_{1,2,j}\ \cdots\ \hh_{K,2,j}]^T$ and $\hh_{k,i,j} \sim \mathcal{CN}(0,\omega_{ij}\bI)$. Then, $\bA_j=\left(\rho + \frac{1}{N}\widehat{\mathsf{H}}_j^H\widehat{\mathsf{H}}_j\right)^{-1}$ and
\[
\frac{1}{N}\TR{\bA_j}=\int \frac{1}{\lambda+\rho}\ dF_{\widehat{\mathsf{H}}_j^H\widehat{\mathsf{H}}_j}
\]
where $F_{\widehat{\mathsf{H}}_j^H\widehat{\mathsf{H}}_j}$ is the empirical eigenvalue distribution of ${\widehat{\mathsf{H}}_j^H\widehat{\mathsf{H}}_j}$.
From Theorem \ref{th:Girko}, this distribution converges almost surely to a limiting distribution $\mathsf{F}$ whose Stieltjes transform $m_\mathsf{F}(z)$. It can be shown that
\[
\frac{1}{N}\TR{\bA_j} \as m_\mathsf{F}(-\rho) = \int_0^1 u(x,-\rho)\ dx.
\] 
where
\[
u(x,-\rho)=u(-\rho)=\frac{1}{\rho+\frac{\beta\omega_{1j}}{1+\omega_{1j}u(-\rho)}+\frac{\beta\omega_{2j}}{1+\omega_{2j}u(-\rho)}}=\frac{1}{\rho + \frac{\beta\omega_d}{1+\omega_du(-\rho)}+\frac{\beta\omega_c}{1+\omega_cu(-\rho)}}.
\]
for $0\leq x \leq 1$. Let $\Gamma=u(-\rho)$, then $\frac{1}{N}\TR{\bA_j} \as \Gamma$.

\subsection{Analog Feedback}\label{App:limSINR_CBf_AF}
Based on the channel model \eqref{eq:ch_model}, we have $\phi_{k,j,j}=\phi_{k,j,j'}=1$. The definitions for other terms such as $\omega_{\bullet}$ and $\delta_{\bullet}$ can be seen in Section \ref{ss:AF_model}. Now, let us first derive the large system limit for the numerator of the $\tSINR_{kj}$. We start with the term $S_{kj}^{(1)}$. From Lemma \ref{th:Lemma1_jse} and by applying \cite[Lemma 5.1]{Liang_it07}, we can show that
\[
\underset{j=1,2,k \leq K}{\max}\ \left|S_{kj}^{(1)}  - \frac{\omega_d}{N}\TR{\bA_{kj}} \right| \as 0. 
\]
By applying rank-1 perturbation lemma (see e.g., \cite[Lemma 3]{Zakhour_it12}, \cite[Lemma 14.3]{Couillet_book11}), we have
\[
\underset{j=1,2,k \leq K}{\max}\ \left|S_{kj}^{(1)} - \frac{\omega_d}{N}\TR{\bA_{j}} \right| \as 0 
\]
where $\frac{1}{N}\TR{\bA_j} \as \Gamma$. 


Since $\hh_{k,j,j}$, $\bA_{kj}$ and $\ih_{k,j,j}$ are independent then it follows that  
\[
\underset{j=1,2,k \leq K}{\max}\ \left|\ih_{k,j,j}\bA_{kj}\hh_{k,j,j}^H \right| \as 0. 
\]
Consequently,
\[
\underset{j=1,2,k \leq K}{\max}\ \left|S_{kj}^{(2)}\right| \as 0 \text{ and }\ \underset{j=1,2,k \leq K}{\max}\ \left|S_{kj}^{(3)}\right| \as 0.
\]
In summary,
\[
\underset{j=1,2,k \leq K}{\max}\ \left||\bh_{k,j,j}\bw_{kj}|^2 - \omega_d^2\Gamma^2 \right| \as 0. 
\]

Now, let us move in analyzing the interference term. By using the matrix inversion lemma, we can rewrite $I_{kj,k'j'}^{(1)}$ as
\[
I_{kj,k'j'}^{(1)}=\frac{\frac{1}{N}\phi_{k,j,j'}\hh_{k,j,j'}\bA_{k'j',kj,j'}\hh_{k',j',j'}^H\hh_{k',j',j'}\bA_{k'j',kj,j'}\hh_{k,j,j'}^H}{\left(1+\frac{1}{N}\hh_{k,j,j'}\bA_{k'j',kj,j'}\hh_{k,j,j'}^H\right)^2}.
\]  
By applying Lemma \ref{th:Lemma1_jse}, \cite[Lemma 5.1]{Liang_it07} and  rank-1 perturbation lemma (R1PL) twice, we can show
\[
\underset{j,j'=1,2,k,k' \leq K,(k,j)\neq(k',j')}{\max}\ \left|\frac{1}{N}\hh_{k,j,j'}\bA_{k'j',kj,j'}\hh_{k,j,j'}^H - \frac{\omega_{jj'}}{N}\TR{\bA_{j'}} \right| \as 0. 
\]
Similarly,
\[
\underset{j,j'=1,2,k,k' \leq K,(k,j)\neq(k',j')}{\max}\ \left|\frac{1}{N}\hh_{k,j,j'}\bA_{k'j',kj,j'}\hh_{k',j',j'}^H\hh_{k',j',j'}\bA_{k'j',kj,j'}\hh_{k,j,j'}^H - \frac{\omega_{jj'}\omega_d}{N}\TR{\bA^2_{j'}} \right| \as 0. 
\]
Since $\frac{1}{N}\TR{\bA_{j'}} \to \Gamma$ and $\frac{1}{N}\TR{\bA_{j'}^2} \to -\frac{\partial \Gamma}{\partial \rho}$, we have
\[
\underset{j,j'=1,2,k,k' \leq K,(k,j)\neq(k',j')}{\max}\ \left|I_{kj,k'j'}^{(1)} - \omega_d\left(-\frac{\omega_{jj'}}{(1+\omega_{jj'}\Gamma)^2}\frac{\partial \Gamma}{\partial \rho}\right) \right| \as 0. 
\]
 
By following the same steps, we obtain
\[
\underset{j,j'=1,2,k,k' \leq K,(k,j)\neq(k',j')}{\max}\ \left|I_{kj,k'j'}^{(2)}   - \left(-\delta_{jj'}\omega_d\frac{\partial \Gamma}{\partial \rho}\right)\right| \as 0. 
\]
and 
\[
\underset{j,j'=1,2,k,k' \leq K,(k,j)\neq(k',j')}{\max}\ \left|I_{kj,k'j'}^{(3)}\right| \as 0.
\]
Combining the results, we have the large system limit for $\mathcal{I}$ in \eqref{eq:int_cbf}
\be\label{eq:lim_CompInt_CBf}
\underset{j,j'=1,2,k,k' \leq K,(k,j)\neq(k',j')}{\max}\ \left|\mathcal{I}-\omega_d\left(-\frac{\omega_{jj'}}{(1+\omega_{jj'}\Gamma)^2}-\delta_{jj'} \right)\frac{\partial \Gamma}{\partial \rho} \right| \as 0. 
\ee
Using \eqref{eq:lim_CompInt_CBf}, the large system result for the interference term can be written as follows
\begin{align*}
\sum_{(k',j')\neq (k,j)} |\bh_{k,j,j'}\bw_{k'j'}|^2 &= \sum_{l=1,l\neq k}^K  |\bh_{k,j,j}\bw_{lj}|^2 + \sum_{l=1}^K  |\bh_{k,j,\bar{j}}\bw_{l\bar{j}}|^2\\
   & \as -\beta\omega_d\left(\frac{\omega_{d}}{(1+\omega_{d}\Gamma)^2}+\frac{\omega_{c}}{(1+\omega_{c}\Gamma)^2}+\delta_{d} + \delta_{c} \right)\frac{\partial \Gamma}{\partial \rho} .
\end{align*}
Now, we just need to derive the large system limit for 
$c_j^2=P\left(\sum_{k=1}^K \|\bw_{kj}\|^2\right)^{-1}$, where we can express $\|\bw_{kj}\|^2=\frac{1}{N^2}\hh_{k,j,j}\bA_{kj}^2\hh_{k,j,j}^H$.
We can show that
\[
\underset{j=1,2,k \leq K}{\max}\ \left|\frac{1}{N^2}\hh_{k,j,j}\bA_{kj}^2\hh_{k,j,j}^H -\frac{\omega_d}{N}\TR{\bA_{j}^2} \right| \as 0. 
\]
Thus,
\[
c_j^2 \as \frac{P}{-\beta\omega_d\frac{\partial \Gamma}{\partial \rho} },
\]
where we can show that
\be\label{eq:der_gamma_rho_CBf_AF}
-\frac{\partial \Gamma}{\partial \rho}=-\Gamma'=\frac{\Gamma}{\rho+\frac{\beta\omega_c}{(1+\omega_c\Gamma)^2}+\frac{\beta\omega_d}{(1+\omega_d\Gamma)^2}}.
\ee

To sum up, from the analyses above, we can express the limiting signal energy as
\be\label{eq:lim_ss_cbf_af}
\frac{1}{\beta} P\omega_d \Gamma  \left(\rho+\frac{\beta\omega_c}{(1+\omega_c\Gamma)^2}+\frac{\beta\omega_d}{(1+\omega_d\Gamma)^2}\right)
\ee
and the limiting interference energy as
\be\label{eq:lim_int_cbf_af}
P\left(\frac{\omega_{d}}{(1+\omega_{d}\Gamma)^2}+\frac{\omega_{c}}{(1+\omega_{c}\Gamma)^2}+\delta_{d} + \delta_{c} \right)
\ee
Finally, the limiting SINR can be expressed as \eqref{eq:limSINR_coord_AF}, with $\Gamma_A=\Gamma$ and $\rho_\text{\tiny C,AF}=\rho$.

\subsection{Proof of Theorem \ref{th:sinr_coord_rvq}: Quantized Feedback via RVQ}\label{App:limSINR_CBf_RVQ}
In the derivation of the large system limit SINR in this section, we use some of the results presented in the previous section. Here, we have $\omega_{jj}=\omega_d=1$ and $\omega_{j\bar{j}}=\omega_c=\epsilon$. From \eqref{eq:ch_model}, we have $\phi_{k,j,i}=1-\tau^2_{k,j,i}$.
 
First, let us consider the numerator of the SINR. By using the result from previous section, we have
\[
\underset{j=1,2,k \leq K}{\max}\ \left|S_{kj}^{(1)} - \frac{1}{N}\TR{\bA_{j}} \right| \as 0 
\]
where $\frac{1}{N}\TR{\bA_j} \as \Gamma$ and $\Gamma$ is the solution of 
\[
\Gamma=\frac{1}{\rho + \frac{\beta}{1+\Gamma}+\frac{\beta\epsilon}{1+\epsilon\Gamma}}.
\]
As stated in \cite{Santipach_it09}, we have,
\be\label{eq:Ed_MCP_RVQ}
\phi_{k,j,j}\ms 1-2^{-\bar{B}_d}.  
\ee 
 Since almost sure convergence and convergence in mean square imply the convergence in probability  then 
\[
\phi_{k,j,j}|S_{kj}^{(1)}|^2 -(1-2^{-\bar{B}_d})\Gamma^2 \ip 0.
\]

By using \eqref{eq:chmodel_rvq}, the term $\ih_{k,j,j}\bA_{kj}\hh_{k,j,j}^H$ in $S_{kj}^{(3)}$ can be rewritten as
\begin{align*}
\frac{1}{N}\ih_{k,j,j}\bA_{kj}\hh_{k,j,j}^H&=\frac{\tau_{k,j,j}\|\bh_{k,j,j}\|}{\|\bv_{k,j,j}\bPi_{\hh_{k,j,j}}^\bot\|}\left(\frac{1}{N}\bv_{k,i,j}\bPi_{\hh_{k,i,j}}^\bot\hh_{k,j,j}^H\right) \\
&=\frac{\tau_{k,j,j}\|\bh_{k,j,j}\|}{\|\bv_{k,j,j}\bPi_{\hh_{k,j,j}}^\bot\|}\left(\frac{1}{N}\bv_{k,j,j}\bA_{kj}\hh_{k,j,j}^H - \frac{(\frac{1}{N}\bv_{k,j,j}\hh^H_{k,j,j})\hh_{k,j,j}\bA_{kj}\hh_{k,j,j}^H}{\|\hh_{k,j,j}\|^2}\right).
\end{align*}
Since $\bv_{k,j,j}$ and $\hh^H_{k,j,j}$ are independent, then
\[
\underset{j=1,2,k \leq K}{\max}\ \left|\frac{1}{N}\bv_{k,j,j}\hh_{k,j,j}^H \right|\as 0, \text{ and } \underset{j=1,2,k \leq K}{\max}\ \left|\frac{1}{N}\bv_{k,j,j}\bA_{kj}\hh_{k,j,j}^H \right|\as 0.
\]
It can also be shown that
\[
\underset{j=1,2,k \leq K}{\max}\ \left|\frac{1}{N}\|\bh_{k,j,j}\|^2 - 1\right| \as 0, \text{ and }  \underset{j=1,2,k \leq K}{\max}\ \left|\frac{1}{N}\|\bv_{k,j,j}\bPi_{\hh_{k,j,j}}^\bot\|^2 - 1\right| \as 0.
\]
Hence,
\be\label{eq:LSA_ih_hh}
\frac{1}{N}\ih_{k,j,j}\bA_{kj}\hh_{k,j,j}^H \ip 0,
\ee
and thus,
\[
S_{kj}^{(2)} \ip 0, \text{ and }  S_{kj}^{(3)} \ip 0. 
\]

Putting the results together, we have the following for the numerator
\[
|\bh_{k,j,j}\bw_{kj}|^2 - (1-2^{-\bar{B}_d})\Gamma^2 \ip 0.
\]
Now, let us consider the interference terms. By using the same steps as in the previous section, we can show the followings
\[
\underset{j,j'=1,2,k,k' \leq K,(k,j)\neq(k',j')}{\max}\ \left|I_{kj,k'j'}^{(1)} - \left(-\frac{(1-2^{-\bar{B}_{jj'}})\omega_{jj'}}{(1+\omega_{jj'}\Gamma)^2}\frac{\partial \Gamma}{\partial \rho}\right)\right| \ip 0, 
\]
\[
\underset{j,j'=1,2,k,k' \leq K,(k,j)\neq(k',j')}{\max}\ \left|I_{kj,k'j'}^{(2)}   - \epsilon_{jj'}2^{-\bar{B}_{jj'}}\frac{\partial \Gamma}{\partial \rho}\right| \ip 0, 
\]
and 
\[
\underset{j,j'=1,2,k,k' \leq K,(k,j)\neq(k',j')}{\max}\ |I_{kj,k'j'}^{(3)}|\ip 0,
\] 
where $\bar{B}_{jj'}=\bar{B}_d$ when $j=j'$ and otherwise $\bar{B}_{jj'}=\bar{B}_c$.

Combining the results, we have
\be\label{eq:lim_CompInt_CBf_rvq}
\underset{j,j'=1,2,k,k' \leq K,(k,j)\neq(k',j')}{\max}\ \left|\mathcal{I}-\left(-\frac{(1-2^{-\bar{B}_{jj'}})\omega_{jj'}}{(1+\omega_{jj'}\Gamma)^2}-\epsilon_{jj'}2^{-\bar{B}_{jj'}}\right)\frac{\partial \Gamma}{\partial \rho} \right| \ip 0. 
\ee
Using \eqref{eq:lim_CompInt_CBf_rvq}, the large system result for the interference term can be written as follows
\begin{align*}
\sum_{(k',j')\neq (k,j)} |\bh_{k,j,j'}\bw_{k'j'}|^2 &= \sum_{l=1,l\neq k}^K  |\bh_{k,j,j}\bw_{lj}|^2 + \sum_{l=1}^K  |\bh_{k,j,\bar{j}}\bw_{l\bar{j}}|^2\\
   & \ip -\beta\left(\frac{1-2^{-\bar{B}_d}}{(1+\Gamma)^2}+\frac{\epsilon(1-2^{-\bar{B}_c})}{(1+\epsilon\Gamma)^2}+2^{-\bar{B}_d} + \epsilon 2^{-\bar{B}_c} \right)\frac{\partial \Gamma}{\partial \rho} .
\end{align*}

By using the result from the previous results straightforwardly, we have $c_j^2 \as \frac{P}{-\beta\frac{\partial \Gamma}{\partial \rho} }$. 
Putting all the large system results for each term, we can show that the limiting signal strength is
\be\label{eq:lim_ss_cbf_rvq}
 \frac{1}{\beta} P\phi_d \Gamma  \left(\rho+\frac{\beta\epsilon}{(1+\epsilon\Gamma)^2}+\frac{\beta}{(1+\Gamma)^2}\right)
\ee
and the limiting interference energy becomes
\be\label{eq:lim_int_cbf_rvq}
P\left(\frac{\phi_d}{(1+\Gamma)^2}+\frac{\epsilon\phi_c}{(1+\epsilon\Gamma)^2} + \delta_d+\delta_c\right).
\ee
Let $\rho_\text{\tiny C,Q}=\rho$ and $\Gamma_Q=\Gamma$. Then, we can obtain the limiting SINR given by \eqref{eq:limSINR_CBf_RVQ} from \eqref{eq:lim_ss_cbf_rvq} and \eqref{eq:lim_int_cbf_rvq} straightforwardly.

\bibliographystyle{IEEEtran}
\bibliography{multicell}

\end{document}